\documentclass[10pt,a4paper]{article}
\pdfoutput=1
\usepackage[utf8]{inputenc}
\usepackage[english]{babel}
\usepackage{csquotes}
\usepackage{newlfont}

\usepackage{amsmath}
\usepackage{amssymb}
\usepackage{amsthm}
\usepackage{color}
\usepackage{tikz-cd}
\usepackage{stmaryrd}
\usepackage{float}
\usepackage{centernot}
\usepackage{a4wide}
\usepackage{proof-at-the-end}
\usepackage[style=numeric]{biblatex}
\bibliography{main.bib}
\usepackage{enumitem}

\newcommand{\red}{\rightarrow}
\newcommand{\amred}{\Rightarrow}
\newcommand{\stred}{\rightharpoonup}
\newcommand{\converges}{\downarrow}
\newcommand{\diverges}{\uparrow}
\newcommand{\deadlocks}{\bot}
\newcommand{\nconverges}{\centernot\downarrow}
\newcommand{\ndiverges}{\centernot\uparrow}
\newcommand{\ndeadlocks}{\centernot\bot}

\newcommand{\tuple}[1]{\langle #1 \rangle}
\newcommand{\bang}{\;!}
\newcommand{\circt}{\operatorname{\mathsf{Circ}}}

\newcommand{\functorfam}{\mathbf{M}_\labels}
\newcommand{\deadlocked}{\mathcal{D}}

\newcommand{\wiretypes}{\mathcal{W}}
\newcommand{\labels}{\mathcal{L}}
\newcommand{\initialmachine}{\mathcal{I}_{ma}}
\newcommand{\initialstacked}{\mathcal{I}_{st}}

\newcommand{\lift }{\operatorname{\mathsf{lift}}}
\newcommand{\force}{\operatorname{\mathsf{force}}}
\newcommand{\boxt}{\operatorname{\mathsf{box}}}
\newcommand{\apply}{\operatorname{\mathsf{apply}}}
\newcommand{\letin}[3]{\operatorname{\mathsf{let}}\, #1 = #2 \operatorname{\mathsf{in}} #3}

\newcommand{\lookup}{\operatorname{\mathsf{lookup}}}

\newcommand{\freshlabels}{\operatorname{\mathsf{freshlabels}}}
\newcommand{\append}{\operatorname{\mathsf{append}}}
\newcommand{\build}{\operatorname{\mathsf{fromMachine}}}
\newcommand{\frommachine}{\build} 
\newcommand{\fromsmallstep}{\operatorname{\mathsf{fromSmallStep}}}
\newcommand{\loadmachine}{\operatorname{\mathsf{load}}}
\newcommand{\clen}{\operatorname{\mathsf{clen}}}

\newcommand{\obs}{\operatorname{ob}}

\newcommand{\freelabels}{\mathit{FL}}
\newcommand{\freevariables}{\mathit{FV}}

\newcommand{\contfarg}{\mathit{FArg}}
\newcommand{\contfapp}{\mathit{FApp}}
\newcommand{\contalabel}{\mathit{ALabel}}
\newcommand{\contacirc}{\mathit{ACirc}}
\newcommand{\conttright}{\mathit{TRight}}
\newcommand{\conttleft}{\mathit{TLeft}}
\newcommand{\contbox}{\mathit{Box}}
\newcommand{\contsub}{\mathit{Sub}}
\newcommand{\contforce}{\mathit{Force}}
\newcommand{\contlet}{\mathit{Let}}

\newtheorem{thm}{Theorem}[section]
\newtheorem{lem}[thm]{Lemma}
\newtheorem{prop}[thm]{Proposition}
\newtheorem{cor}{Corollary}[thm]
\newtheorem{df}{Definition}[section]

\newcommand{\cat}[1]{\mathbf{#1}}
\newcommand{\void}{\phantom{-}}
\newcommand{\serr}{\textnormal{Error}}
\newcommand{\sother}{\textnormal{Otherwise}}

\title{On Abstract Machine Semantics for Proto-Quipper-M
\footnote{
    The contents of this paper are taken and adapted from the author's Master's degree thesis.
}}
\author{Andrea Colledan}
\date{\footnotesize University of Bologna\\Bologna, Italy}

\begin{document}

\maketitle
    
\begin{abstract}
    Quipper is a domain-specific programming language for the description of quantum circuits. Because it is implemented as an embedded language in Haskell, Quipper is a very practical functional language. However, for the same reason, it lacks a formal semantics and it is limited by Haskell's type system. In particular, because Haskell lacks linear types, it is easy to write Quipper programs that violate the non-cloning property of quantum states. In order to formalize relevant fragments of Quipper in a type-safe way, the Proto-Quipper family of research languages has been introduced over the last years. In this paper we first review Proto-Quipper-M, an instance of the Proto-Quipper family based on a categorical model for quantum circuits, which features a linear type system that guarantees that the non-cloning property holds at compile time. We then derive a tentative small-step operational semantics from the big-step semantics of Proto-Quipper-M and we prove that the two are equivalent. After proving subject reduction and progress results for the tentative semantics, we build upon it to obtain a truly small-step semantics in the style of an abstract machine, which we eventually prove to be equivalent to the original semantics.
\end{abstract}
    
\section{Introduction}
As progress is made in the physical realization of new and more powerful quantum computers, the need for a quantum programming language that goes beyond a mere instruction set for quantum hardware and instead offers high-level features similar to those we are already used to in classical programming becomes more and more apparent. Today, it would be ridiculous if a programmer were to try and code a web application by defining every part of it in terms of logic gates. Similarly, it is unreasonable to expect that a quantum programmer in the future will describe all of its quantum algorithms in terms of elementary unitary transformations. At the time of writing, one of the most promising candidates for a quantum programming language suitable to real-world applications is undoubtedly Quipper \cite{quipper-informal,quipper}.

\paragraph{} Quipper is a functional programming language for the description of quantum circuits. What sets Quipper apart from the majority of the remaining quantum programming languages is that it is designed with the explicit objective of being practical, scalable and ultimately useful. To this effect, Quipper is implemented as an embedded language in Haskell, so that all of the advanced programming constructs that are available in Haskell are also available when programming in Quipper. The result is a powerful quantum programming language that does not limit the programmer to a gate-by-gate description of quantum computations, but rather treats circuits themselves as data and supports many higher-order operations to combine them together and manipulate them in their entirety. This allows for the implementation of many real-world quantum algorithms that would be practically inexpressible in other programming languages, if anything due to the sheer size of their circuits.

\paragraph{} Unfortunately, while Quipper inherits all of the qualities of Haskell, it also inherits its shortcomings. Namely, Quipper lacks linear types, which are critical to quantum programming, and more generally a formal operational semantics. As a consequence, it is difficult to reason rigorously about the behavior of Quipper programs, a fact that constitutes an obstacle to the application of otherwise valuable static analysis techniques to them. For example, the ability to statically infer bounds on the number of qubits required at run time by a Quipper program would be immensely useful in a time where quantum resources -- although increasingly available -- are still scarce. Fortunately, a number of research languages exist that formalize significant fragments of Quipper in a type-safe way. In this paper, we examine one such language, namely Rios and Selinger's Proto-Quipper-M \cite{proto-quipper-m}, and use its big-step semantics as a starting point to define a new operational semantics for Quipper which is inspired by abstract machines. We then prove that this new semantics is equivalent to the original one. Our hope is that our work will in turn serve as a valuable starting point for future research in the formalization of more advanced Quipper constructs and in the static analysis of Quipper programs.

\subsection{Contents of the Paper}

\paragraph{} In Section 2 we review Proto-Quipper-M, a type-safe formalization of a relevant fragment of Quipper. We start by introducing the categorical model upon which circuit construction in Proto-Quipper-M is built. Then we proceed to present the language itself, with a particular focus on its linear type system, which can prevent at compile time a number of mistakes that would result in a run time exception in Quipper, namely those related to the violation of the no-cloning property of quantum states. To conclude the section, we cover Proto-Quipper-M's big-step semantics and make a first attempt to define an equivalent small-step semantics. We give safety results for the resulting semantics and we assess its limitations, specifically as far as the circuit boxing operation is concerned.

\paragraph{} In Section 3 we present two incremental upgrades to the small-step semantics defined in the previous section. First, we propose a \textit{stacked semantics}, which overcomes the shortcomings of the previous semantics by introducing an explicit stack into the small-step semantics, to keep track of nested boxing operations. Next, we take the stack approach even further and formulate a proposal for a \textit{machine semantics} for Proto-Quipper-M. This semantics is heavily inspired by abstract machines, and particularly by the CEK machine \cite{cek}, as it models every phase of the evaluation of a Proto-Quipper-M program as a continuation on a stack.

\paragraph{}Finally, in the more technical Section 4 we analyze the three semantics in their relationship with one another, eventually proving that the proposed machine semantics is effectively equivalent to the starting small-step semantics and -- as a consequence -- to the original Proto-Quipper-M semantics given by Rios and Selinger.

\section{Proto-Quipper-M: a Formalization of Quipper}
\label{proto-quipper-m}

Quipper \cite{quipper-informal,quipper} is a functional programming language for the description of quantum circuits. Unlike other quantum programming languages, Quipper is designed with the goal of being first and foremost practical and scalable, allowing programmers to leverage the power of higher-order operators to describe quantum algorithms requiring order of trillions of gates. In order to provide this kind of power, Quipper is currently implemented as an embedded programming language in Haskell (that is, as a library and preferred idiom for Haskell), which means that it benefits from all of Haskell's features, including some advanced and experimental GHC extensions. The embedding route, however practical, comes at a price. As we mentioned in the introduction, Quipper lacks a formal semantics, which means that it is hard to reason formally about the behavior of Quipper programs. More importantly, Quipper lacks linear types, which are essential in quantum programming in that they prevent violations of the \textit{no-cloning theorem}, an ubiquitous result in quantum physics which asserts that it is impossible to duplicate an arbitrary unknown quantum state. In quantum computing, this constraint is also referred to as the \textit{no-cloning property} of quantum states and it entails that no quantum gate can create a copy of a qubit. Because Quipper cannot enforce this property at compile time through linear types, it follows that it is not a type-safe language.

\paragraph{}In order to still be able to study Quipper in a formal way, the \textit{Proto-Quipper} family of research languages has been introduced over the last years. Each language of this family formalizes a relevant fragment of Quipper in a type-safe way. The most prominent Proto-Quipper instances are Proto-Quipper-S \cite{proto-quipper-s}, Proto-Quipper-M \cite{proto-quipper-m} and Proto-Quipper-D \cite{proto-quipper-d}. In particular, Proto-Quipper-M is a lambda-calculus built upon a categorical model for circuit building, and features a full-fledged linear type system. In general, a linear type system guarantees that certain variables -- more technically, \textit{linear resources} -- are consumed exactly once. In the case of Proto-Quipper-M, the linear resources are the free wire ends in the circuit being built as a side-effect of the evaluation of a program. This makes it so that, unlike Quipper, Proto-Quipper-M \textit{can} enforce the no-cloning property of quantum states at compile time, guaranteeing type-safety. In this section we give an overview of Proto-Quipper-M, starting from its categorical model, which actually generalizes the notion of quantum circuit, and reviewing its syntax, type system and semantics.

\paragraph{}Note that throughout this section we assume that the reader is already somewhat familiar with quantum computing and specifically with the quantum circuit model for the description of quantum computations. The reader who is unfamiliar with quantum circuits can find a minimal introduction to the topic in the author's bachelor thesis \cite{bachthesis}. Alternatively, for a more thorough introduction to quantum mechanics and their application to computer science, we refer the reader to textbooks such as the ones by Yanofsky and Mannucci \cite{yanman} and Nielsen and Chuang \cite{nielchuang}.

\subsection{Generalizing Quantum Circuits}

In Proto-Quipper-M, a quantum circuit is modeled as a morphism in a symmetric monoidal category. To understand what this means exactly, we first need to know what a symmetric monoidal category is. Note that in the following pages we assume that the reader is already familiar with some basic concepts of category theory, such as morphisms, their composition, isomorphisms, functors and  bifunctors. For a proper introduction to category theory, refer to the excellent works of Riehl \cite{intro-cat-riehl} and Asperti and Longo \cite{intro-cat-asp}.

\subsubsection{Generalized Circuits}

\begin{df}[Monoidal Category]
A category $\cat C$  is said to be \emph{monoidal} if it is equipped with:
\begin{itemize}
    \item A bifunctor $\otimes: \cat C \times \cat C \to \cat C$, called \emph{tensor product},
    \item An object $I$, called \emph{identity object},
    \item Three natural isomorphism which guarantee that
    \begin{itemize}
        \item $\otimes$ is associative: for all $A,B,C$ in $\cat C$ there exists an isomorphism $\alpha_{A,B,C}:A \otimes (B \otimes C) \cong (A \otimes B) \otimes C$, called \emph{associator}, which is natural in $A,B$ and $C$,
        \item $I$ is a left identity for $\otimes$: for all $A$ in $\cat C$ there exists a natural isomorphism $\lambda_A : I \otimes A \cong A$, called \emph{left unitor},
        \item $I$ is a right identity for $\otimes$: for all $A$ in $\cat C$ there exists a natural isomorphism $\rho_A : A \otimes I \cong A$, called \emph{right unitor},
    \end{itemize}
    and such that the following diagrams commute:
    \begin{itemize}
        \item For all $A,B,C,D$ in $\cat C$:
        \begin{center}
        \begin{tikzcd}[column sep=huge]
            A \otimes (B \otimes (C \otimes D))
            \arrow[r,"id_A\otimes \alpha_{B,C,D}"]
            \arrow[d,"\alpha_{A,B,C\otimes D}"]
            & A \otimes ((B \otimes C) \otimes D)
            \arrow[dd,"\alpha_{A,B\otimes C, D}"]\\
            (A \otimes B) \otimes (C \otimes D) \arrow[d,"\alpha_{A\otimes B, C,D}"]\\
            ((A \otimes B) \otimes C ) \otimes D
            & (A \otimes (B \otimes C)) \otimes D
            \arrow[l,"\alpha_{A,B,C}\otimes id_D"]
        \end{tikzcd}
        
        \end{center} 
        \item For all $A,B$ in $\cat C$:
        \begin{center}
            \begin{tikzcd}
            A \otimes (I \otimes B) \arrow[rr,"\alpha_{A,I,B}"] \arrow[dr,"id_A\otimes \lambda_B"]
            && (A \otimes I) \otimes B \arrow[ld," \rho_A \otimes id_B"]
            \\
            &A \otimes B
            \end{tikzcd}
        \end{center}
    \end{itemize}
\end{itemize}
\end{df}

\paragraph{} We can start to see how quantum circuits can be reasoned about in terms of monoidal categories. Suppose the objects of a given monoidal category $\cat M$ represent collections of wires. If $A,B\in\obs(\cat M)$ are two such collections, then a circuit $C$ that takes as input the wires in $A$ and outputs the wires in $B$ can be clearly modelled by a morphism $C:A\to B$ in $\cat M$. Furthermore, if the wires output by a circuit $C$ coincide with the wires taken as input by a circuit $D$, the two circuits can be composed in series, as shown graphically:

\begin{figure}[H]
    \centering
    \includegraphics[scale=.9]{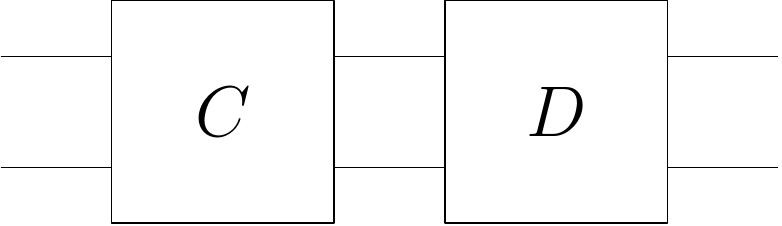}
\end{figure}

\noindent This is naturally modelled by the associative composition function $\circ$, since circuit composition is also associative. Also, for every collection of wires $A$, the circuit that does nothing to the wires in $A$ and returns them unaltered is a perfectly valid circuit. It is modelled by the identity morphism $id_A$ and appending or prepending it to any other circuit has no effect whatsoever on that circuit:

\begin{figure}[H]
    \centering
    \includegraphics[scale=.9]{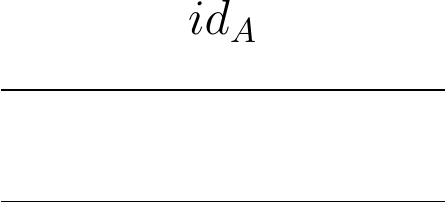}
\end{figure}

So far, these aspects of circuit building can be modelled by any category. Now we delve into the specifics of \textit{monoidal} categories. Whenever we have two circuits $C$ and $D$ we can compose them in parallel, regardless of their input or output wires, as shown graphically:

\begin{figure}[H]
    \centering
    \includegraphics[scale=.9]{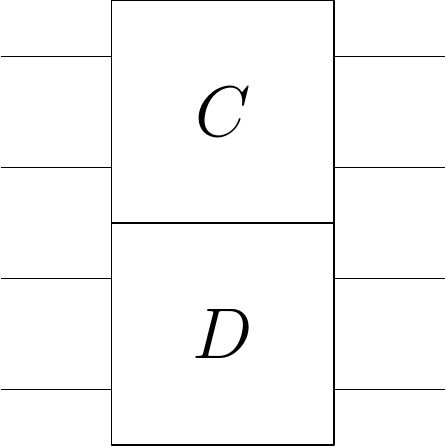}
\end{figure}

\noindent This is naturally modelled by the tensor product $\otimes$ of the monoidal category $\cat M$. Note that, being a functor, $\otimes$ can be applied to both morphisms (circuits) and objects (collections of wires), which means that whenever we have two collections of wires $A$ and $B$, we can put them together into a single collection $A\otimes B$. To this effect, the identity object $I$ represents the empty collection of wires. There is one last property that we would like to reflect in our categorical model, and it is that the wires of a quantum circuit can be rearranged freely (together with all the gates that act on them, naturally), without altering the fundamental nature of the circuit itself. A monoidal category is not enough to model this property, which is why we turn to \textit{symmetric monoidal categories}.

\begin{df}[Symmetric Monoidal Category]
A monoidal category $\cat C$ is said to be \emph{symmetric} when it is equipped, for all $A,B$ in $\cat C$, with an isomorphism
$$\gamma_{A,B}:A\otimes B \cong B \otimes A,$$
which is natural in both $A$ and $B$ and such that the following diagrams commute
\begin{itemize}
    \item For all $A$ in $\cat C$:
    \begin{center}
        \begin{tikzcd}
        A \otimes I \arrow[rr, "\gamma_{A,I}"] \arrow[dr, "\rho_A"]
        && I \otimes A \arrow[dl, "\lambda_A"]
        \\
        & A
        \end{tikzcd}
    \end{center}
    \item For all $A,B,C$ in $\cat C$:
        \begin{center}
            \begin{tikzcd}
                (A\otimes B) \otimes C \arrow[rr, "\gamma_{A,B}\otimes id_C"] \arrow[d, "\alpha_{A,B,C}"]
                && (B \otimes A) \otimes C \arrow[d, "\alpha_{B,A,C}"]
                \\
                A \otimes (B \otimes C) \arrow[d, "\gamma_{A,B\otimes C}"]
                && B \otimes (A \otimes C) \arrow[d, "id_B\otimes \gamma_{A,C}"]
                \\
                (B \otimes C) \otimes A \arrow[rr, "\alpha_{B,C,A}"]
                && B \otimes (C \otimes A)
            \end{tikzcd}
        \end{center}
        \item For all $A,B$ in $\cat C$:
        \begin{center}
            \begin{tikzcd}
            A \otimes B \arrow[d, bend left, "\gamma_{A,B}"]
            \\
            B\otimes A \arrow[u, bend left, "\gamma_{B,A}"]
            \end{tikzcd}
        \end{center}
\end{itemize}
\end{df}

\noindent This definition reflects the irrelevance of the order of the wires in a circuit precisely in the existence of the isomorphism $\gamma$. In conclusion, a symmetric monoidal category $\cat M$ offers a good mathematical model for quantum circuits. As we mentioned earlier, Proto-Quipper-M is a language designed specifically for describing morphisms in a symmetric monoidal category, which we call \textit{generalized circuits} from now on.

\begin{df}[Generalized Circuit]
Let $\cat M$ be a symmetric monoidal category. We call the morphisms of $\cat M$ \emph{generalized circuits}.
\end{df}

\noindent Note that because this definition of circuit is so general, Proto-Quipper-M can be used to describe \textit{any} instance of a symmetric monoidal category. This includes different quantum circuit representations (such as DAGs or unitary matrices), as well as other entities which are not necessarily quantum circuits, or circuits at all. In this respect, Proto-Quipper-M is more general than Quipper.

\subsubsection{Generalized Labelled Circuits}\label{generalized labelled circuits}

Although the definition that we just gave is by itself sufficient to characterize a quantum circuit categorically, it would be practical -- almost \textit{necessary}, from a programming point of view -- to have a way to identify and pick individual wires out of a collection, rather than treating said collection as an atomic object. We therefore introduce \textit{labels}, which behave as pointers to individual free wire ends, and we associate a \textit{wire type} to each one of them. In the case of quantum circuits, the types of wires are likely to be either \textit{Bit} or \textit{Qubit}, but for the sake of generality we assume that wire types come from an arbitrary set $\wiretypes$, which is a parameter of the model of the language.

\begin{df}[Wire Types]
Let $\cat M$ be a given symmetric monoidal category, and let $\wiretypes$ be a set equipped with an interpretation function
$$\llbracket \cdot \rrbracket : \wiretypes \to \obs(\cat M),$$
that is, a mapping from the elements of $\wiretypes$ to the objects of $\cat M$. We call the elements of $\wiretypes$ \emph{wire types}.
\end{df}

Wires can be considered individually or in bundles. In the second case, we assume that wire labels can be ordered and we refer to the resulting collection of mappings from labels to wire types as a \textit{label context}.

\begin{df}[Label Context]
Let $\labels$ be a fixed countably infinite set of \emph{label names}, which we assume to be totally ordered. A \emph{label context} $Q$ is a function of the form
$$Q:\labels\to \wiretypes.$$
Such a function that maps label names $\ell_1,\ell_2,\dots,\ell_n$ respectively to wire types $\alpha_1,\alpha_2,\dots,\alpha_n$ can be written as follows:
$$\ell_1:\alpha_1,\ell_2:\alpha_2,\dots,\ell_n:\alpha_n.$$
The interpretation of a label context $Q=\ell_1:\alpha_1,\ell_2:\alpha_2,\dots,\ell_n:\alpha_n$ is the following object of $\cat M$:
$$\llbracket Q \rrbracket = \llbracket\alpha_1\rrbracket \otimes \llbracket\alpha_2\rrbracket \otimes \dots \otimes \llbracket\alpha_n\rrbracket,$$
where $\ell_1 < \ell_2 < \dots < \ell_n$. In the case where $Q = \emptyset$, we have $\llbracket \emptyset \rrbracket = I$.
\end{df}

At this point, by instantiating the generic objects of $\cat M$ with label contexts, we get the category $\cat M_\labels$, which is a truly suitable model for a quantum circuit description language. Note that $\cat M$ and $\cat M_\labels$ are essentially the same category, the only difference between the two being that $\cat M_\labels$ is imbued with a labelling structure that allows us to identify individual wires and their type. To reflect this quality, we call the morphisms of $\cat M_\labels$ \textit{generalized labelled circuits}.

\begin{df}[Generalized Labelled Circuit]
Let $\cat M$ be a given symmetric monoidal category. Let $\cat M_\labels$ be a category in which
\begin{itemize}
    \item The objects are label contexts,
    \item A morphism $f:Q\to R$ is a morphism $g:\llbracket Q \rrbracket \to \llbracket R \rrbracket$ in $\cat M$.
\end{itemize}
We call the morphisms of $\cat M_\labels$ \emph{generalized labelled circuits}, or just \emph{labelled circuits}.
\end{df}

\subsection{Proto-Quipper-M's Syntax}

We are now ready to start examining Proto-Quipper-M, starting from its syntax. Although the original Proto-Quipper-M specification given by Rios and Selinger \cite{proto-quipper-m} is quite rich, in this paper we only consider a minimal fragment of the language, for the sake of simplicity. Our fragment can be described by the following grammar:
\begin{align*}
    M,N ::=\; &x \mid \ell \mid  \lambda x.M\mid MN \mid \tuple{M,N} \mid \letin{\tuple{x,y}}{M}{N} \\&
    \mid \lift  M \mid \force   M \mid \boxt_T M \mid \apply (M,N) \mid (\vec\ell, C, \vec{\ell'}),
\end{align*}
where $x$ ranges over variables names, $\ell$ ranges over the label names in $\labels$ and $C$ is a labelled circuit, that is, a morphism in $\cat M_\labels$. The ultimate goal of the evaluation of a Proto-Quipper-M program is the construction of a quantum circuit as a side-effect. Here, ``side-effect'' means that a circuit does not actually appear in the term that builds it, but rather lives ``behind the scenes'' and any changes made to it are, in a way, imperative in nature. For this reason, we often refer to the circuit being built by a program as the \textit{underlying} circuit.
That being said, it comes as no surprise that the most essential programming constructs of Proto-Quipper-M are those that allow to work with circuits. In particular, a term of the form $(\vec\ell,C,\vec{\ell'})$ is called a \textit{boxed circuit} and allows to treat quantum circuits as data: it corresponds to a labelled quantum circuit $C:Q\to Q'$ which exposes input labels $\vec\ell$ and output labels $\vec{\ell'}$ as an interface (where $\vec\ell$ and $\vec{\ell'}$ are all and only the labels occurring in $Q$ and $Q'$, respectively). New boxed circuits can be introduced via the $\boxt_T$ construct, which corresponds to Quipper's \texttt{box} operator. Informally, $\boxt_T$ takes a circuit-building function, executes it in a sandboxed environment (with new labels generated on-the-fly), boxes the resulting circuit and returns it as a result. On the other hand, the $\apply$ construct models the application of an existing boxed circuit to one or more exposed wires of the underlying circuit. Values are a subset of terms and they can be defined by the following grammar:

$$V,W::= \ell \mid \lambda x.M \mid \tuple{V,W} \mid \lift  M \mid (\vec\ell, C, \vec\ell').$$

\noindent Note that we often employ $\vec\ell$ as syntactic sugar to denote an arbitrary tuple of labels. More formally:

$$\vec\ell,\vec k ::= \ell \mid \tuple{\vec\ell,\vec k}.$$

\paragraph{}The original work by Rios and Selinger includes constants in the language. In particular, it assumes that for every quantum gate there exists a function constant that applies the corresponding gate to the underlying circuit. For simplicity, we decided to omit constants from the language. Instead, we assume that these constants exist as specific morphisms in the $\cat M_\labels$ category. For example, we assume that a morphism $H:(\ell:\mathsf{Qubit})\to (\ell':\mathsf{Qubit})$ exists and represents the circuit where the Hadamard gate is applied to a single qubit wire $\ell$ to obtain an output qubit wire $\ell'$. We can therefore refer to the Hadamard gate within our language with a term of the form $(\ell,H,\ell')$.
We now proceed to give some standard preliminary definitions which will be used in the coming sections.

\begin{df}[Free Labels]
The set of \emph{free labels} of a term $M$, denoted as $\freelabels(M)$, is defined as follows:
\begin{align*}
    \freelabels(x) = \freelabels((\vec\ell,C,\vec{\ell'})) &= \emptyset\\
    \freelabels(\ell) &= \{\ell\}\\
    \freelabels(\lambda x.N) = \freelabels(\lift N)  = \freelabels(\force N) = \freelabels(\boxt_T N) &= \freelabels(N)\\
    \freelabels(NP) = \freelabels(\tuple{N,P}) = \freelabels(\letin{\tuple{x,y}}{N}{P}) = \freelabels(\apply(N,P)) &= \freelabels(N)\cup \freelabels(P).
\end{align*}
\end{df}

\noindent Notice how basically all of the labels occurring in a term are free. This is because labels are not bound by the constructs of the language, but rather by the underlying circuit. We will consider this aspect in more detail in the coming sections, and especially in Section \ref{safety results}.

\begin{df}[Free Variables]
The set of \emph{free variables} of a term $M$, denoted as $\freevariables(M)$, is defined as follows:
\begin{align*}
    \freevariables(\ell) = \freevariables((\vec\ell,C,\vec{\ell'})) &= \emptyset\\
    \freevariables(x) &= \{x\}\\
    \freevariables(\lambda x.N) &= \freevariables(N)\setminus\{x\}\\
    \freevariables(\lift N) = \freevariables(\force N) = \freevariables(\boxt_T N) &= \freevariables(N)\\
    \freevariables(\letin{\tuple{x,y}}{N}{P}) &= \freevariables(N)\cup (\freevariables(P)\setminus\{x,y\})\\
    \freevariables(NP) = \freevariables(\tuple{N,P}) = \freevariables(\apply(N,P)) &= \freevariables(N)\cup \freevariables(P).
\end{align*}
\end{df}

\begin{df}[Capture-avoiding Substitution]
Let $M$ and $N$ be terms such that none of the variables occurring free in $N$ occur in $M$, free or bound. We define the \emph{substitution of $N$ for $x$ in $M$}, or $M[N/x]$, as follows:

\begin{align*}
    x[N/x] &= N\\
    y[N/x] &= y\\
    \ell[N/x] &= \ell\\
    (\lambda x.L)[N/x] &= \lambda x.L\\
    (LP)[N/x] &= L[N/x]P[N/x]\\
    \tuple{L,P}[N/x] &= \tuple{L[N/x],P[N/x]}\\
    (\letin{\tuple{x,y}}{L}{P})[N/x] &= \letin{\tuple{x,y}}{L[N/x]}{P}\\
    (\letin{\tuple{y,x}}{L}{P})[N/x] &= \letin{\tuple{y,x}}{L[N/x]}{P}\\
    (\letin{\tuple{y,z}}{L}{P})[N/x] &= \letin{\tuple{y,z}}{L[N/x]}{P[N/x]}\\
    (\lift L)[N/x] &= \lift L[N/x]\\
    (\force L)[N/x] &= \force L[N/x]\\
    (\boxt_T L)[N/x] &= \boxt_T L[N/x]\\
    (\apply(L,P))[N/x] &= \apply(L[N/x],P[N/x])\\
    (\vec\ell,C,\vec{\ell'})[N/x] &= (\vec\ell,C,\vec{\ell'}).
\end{align*}
\end{df}

\subsection{Type System}
\label{proto-quipper-m:types}

As we mentioned previously, Proto-Quipper-M is endowed with a linear type system that ensures that quantum states are never used more than once. In fact, there are two kinds of types in Proto-Quipper-M: \textit{parameter types} and \textit{linear types}. As the name suggests, parameter types refer to circuit parameters, which are not subjected to linearity constraints and can be used any number of times. Any type that is not a parameter type is a linear type. A variable of linear type is also referred to as a linear resource and, once introduced, can (and must) be consumed exactly once. Among linear types, we distinguish the \textit{simple M-types}, that is, the types of tuples of labels. For simplicity, we often refer to these as just \textit{M-types}. Ultimately, types can be described by the following grammar:

\begin{align*}
    \text{Types}\quad & A,B && ::= \alpha \mid A\otimes B \mid A \multimap B \mid \bang A \mid \circt(T,U),\\
    \text{Parameter types}\quad & P,R && ::= P\otimes R \mid \bang A \mid \circt(T,U),\\
    \text{Simple M-types}\quad & T,U && ::=\alpha \mid T\otimes U,\\
\end{align*}
where $\alpha$ comes from the set $\wiretypes$ of wire types. We note that $A\multimap B$ is the type of linear abstractions from $A$ to $B$, while $\circt(T,U)$ is the type of circuits from M-type $T$ to M-type $U$.

\paragraph{}We now define the notion of \textit{typing context}. In Proto-Quipper-M, a typing context can contain both parameter variables and linear variables. However, it is often useful to distinguish the case in which a typing context only contains parameter variables from the case in which it contains both kinds of variables. We therefore call a typing context a \textit{parameter context}, and denote it by $\Phi$, if it contains exclusively parameter types, whereas we call it a \textit{generic context}, and denote it by $\Gamma$, if it contains parameter \textit{and} linear types alike. Not that this distinction is in no way formal. In fact, a parameter variable may appear on one occasion in $\Phi$ and on another in $\Gamma$ in two rule applications within the same type derivation. Whereas variables are assigned a type by a typing context, labels are assigned a type by the very label contexts that we saw in Section \ref{generalized labelled circuits}. If a generic context $\Gamma$ and a label context $Q$ turn $M$ into a term of type $A$, then we write the following typing judgement:
$$\Gamma; Q \vdash M : A.$$
Typing judgements can be obtained by the following typing rules:
$$\frac{\void}{\Phi,x:A;\emptyset\vdash x:A}\textit{var}
\qquad
\frac{\void}{\Phi;\ell:\alpha\vdash \ell:\alpha}\textit{labels}
$$
\vspace{5pt}
$$
\frac{\Gamma,x:A;Q \vdash M:B}{\Gamma;Q \vdash \lambda x.M:A\multimap B}\textit{abs}
\qquad
\frac{\Phi,\Gamma_1;Q_1\vdash M : A\multimap B \quad \Phi,\Gamma_2;Q_2\vdash N: A}
{\Phi,\Gamma_1,\Gamma_2;Q_1,Q_2\vdash MN:B}\textit{app}
$$
\vspace{5pt}
$$
\frac{\Phi,\Gamma_1;Q_1\vdash M : A \quad \Phi,\Gamma_2;Q_2\vdash N: B}
{\Phi,\Gamma_1,\Gamma_2;Q_1,Q_2\vdash \tuple{M,N}:A\otimes B
}\textit{tuple}
$$
\vspace{5pt}
$$
\frac{\Phi,\Gamma_1;Q_1\vdash M : A\otimes B \quad \Phi,\Gamma_2,x:A,y:B;Q_2\vdash N: C}
{\Phi,\Gamma_1,\Gamma_2;Q_1,Q_2\vdash \letin{\tuple{x,y}}{M}{N}:C
}\textit{let}
$$
\vspace{5pt}
$$
\frac{\Phi;\emptyset\vdash M:A}
{\Phi;\emptyset\vdash \lift  M:\bang A}\textit{lift}
\qquad
\frac{\Gamma;Q\vdash  M:\bang A}
{\Gamma;Q\vdash  \force  M:\;A}\textit{force}
\qquad
\frac{\Gamma;Q\vdash M:\bang(T\multimap U)}
{\Gamma;Q\vdash \boxt_TM:\;\circt(T, U)}\textit{box}
$$
\vspace{5pt}
$$
\frac{\Phi,\Gamma_1;Q_1\vdash M:\circt(T,U) \quad \Phi,\Gamma_2;Q_2\vdash N:T}
{\Phi,\Gamma_1,\Gamma_2;Q_1,Q_2 \vdash \apply (M,N):U}\textit{apply}
$$
\vspace{6pt}
$$
\frac{\emptyset;Q_1\vdash \vec\ell:T \quad \emptyset;Q_2\vdash \vec{\ell'}:U\quad C\in\functorfam(Q_1,Q_2)}
{\Phi;\emptyset \vdash (\vec\ell,C,\vec{\ell'}):\circt(T,U)}\textit{circ}
$$
\vspace{5pt}

\noindent where we assume that $\Gamma_1$ and $\Gamma_2$ (as well as $Q_1$ and $Q_2$) are always disjoint and $\Gamma_1,\Gamma_2$ denotes the union of contexts $\Gamma_1$ and $\Gamma_2$. Note how the requirement that $\Gamma_1$ and $\Gamma_2$ be disjoint guarantees that a linear variable cannot be used more than once in a term, while the fact that the \textit{var} rule successfully derives $\Phi,\Gamma;\emptyset \vdash x:A$ exclusively if $x$ is the only linear variable in $\Gamma$ guarantees that no linear variable goes unused. Together, these two principles guarantee that every linear variable is used \textit{exactly} once, which is precisely the definition of \textit{linearity}. It is easy to see that this kind of constraint holds for labels too, and in this case the linearity property coincides with the no-cloning property of quantum states. Naturally, we have that substitution behaves well with respect to types. Specifically, by substituting a value of type $A$ for a variable of type $A$ in a term $M$, we do not alter the type of $M$. We assert this property in the following results.

\begin{theoremEnd}[all end]{lem}[Generation of Typing Judgements]The following hold:\label{generation lemma}
\begin{enumerate}
    \item If $\Phi,\Gamma;Q\vdash \ell : A$ then $\Gamma = \emptyset$ and there exists $\alpha\in\wiretypes$ such that $A\equiv \alpha$ and $Q=\ell:\alpha$.
    \item If $\Gamma;Q\vdash \lambda x.M:C$ then there exist $A$ and $B$ such that $\Gamma,x:A;Q\vdash M:B$ and $C\equiv A\multimap B$.
    \item If $\Phi,\Gamma;Q\vdash MN : B$ then there exists $A$, as well as $\Gamma_1,\Gamma_2$ and $Q_1,Q_2$, such that $ \Phi,\Gamma_1;Q_1\vdash M:A\multimap B$ and $\Phi,\Gamma_2;Q_2\vdash N:A$, where $\Gamma_1,\Gamma_2 = \Gamma$ and $Q_1,Q_2=Q$.
    \item If $\Phi,\Gamma;Q\vdash \tuple{M,N} : C$ then there exist $A$ and $B$, as well as $\Gamma_1,\Gamma_2$ and $Q_1,Q_2$, such that $ \Phi,\Gamma_1;Q_1\vdash M:A$,  $\Phi,\Gamma_2;Q_2\vdash N:B$ and $C\equiv A\otimes B$, where $\Gamma_1,\Gamma_2 = \Gamma$ and $Q_1,Q_2=Q$.
    \item If $\Phi,\Gamma;Q\vdash \letin{\tuple{x,y}}{M}{N} : C$ then there exist $A$ and $B$, as well as $\Gamma_1,\Gamma_2$ and $Q_1,Q_2$, such that $ \Phi,\Gamma_1;Q_1\vdash M:A\otimes B$ and $\Phi,\Gamma_2,x:A,y:B;Q_2\vdash N:C$, where $\Gamma_1,\Gamma_2 = \Gamma$ and $Q_1,Q_2=Q$.
    \item If $\Gamma;Q \vdash \force   M:A$, then $\Gamma;Q\vdash M:\bang A$.
    \item If $\Gamma;Q \vdash \boxt_T M:W$ then there exist $T$ and $U$ such that $\Gamma;Q\vdash M:\bang(T\multimap U)$ and $W\equiv\circt(T,U)$.
    \item If $\Phi,\Gamma;Q \vdash \lift  M:C$ then $\Gamma = Q = \emptyset$ and there exists $A$ such that $\Phi;\emptyset\vdash M:A$ and $C\equiv \bang A$
    \item If $\Phi,\Gamma;Q \vdash \apply (M,N):C$ then there exist $T$ and $U$, as well as $\Gamma_1,\Gamma_2$ and $Q_1,Q_2$, such that $\Phi,\Gamma_1;Q_1\vdash M:\circt(T,U)$, $\Phi, \Gamma_2;Q_2\vdash N:T$ and $C\equiv U$, where $\Gamma_1,\Gamma_2 = \Gamma$ and $Q_1,Q_2=Q$.
    \item If $\Phi,\Gamma;Q \vdash (\vec\ell,C,\vec{\ell'}):C$ then $\Gamma = Q = \emptyset$ and there exist $T$ and $U$, as well as $Q_1,Q_2$, such that $\emptyset;Q_1\vdash \vec\ell:T$, $ \emptyset;Q_2\vdash \vec{\ell'}:U$, $C\in\functorfam(Q_1,Q_2)$ and $C\equiv \circt(T,U)$.
\end{enumerate}
\end{theoremEnd}
\begin{proofEnd}
The claims all follow directly from the principle of inversion, since the rule system for types is syntax-directed.
\end{proofEnd}

\begin{theoremEnd}[all end]{lem}[Generation of Values]\label{value generation}
Suppose $V$ is a value. Then the following hold:
\begin{enumerate}
    \item If $\Phi,\Gamma;Q\vdash V:T$ for some simple M-type $T$, then $\Gamma=\emptyset$ and $V\equiv \vec\ell$ for some $\vec\ell$.
    \item If $\Phi,\Gamma;Q\vdash V:A\multimap B$ for some $A,B$, then $V\equiv \lambda x. N$ for some $x$ and $N$.
    \item If $\Phi,\Gamma;Q\vdash V: A\otimes B$ for some $A,B$, then $V\equiv \tuple{V_1,V_2}$ for some values $V_1,V_2$.
    \item If $\Phi,\Gamma;Q\vdash V:\bang A$ for some $A$, then $\Gamma=Q=\emptyset$ and $V\equiv \lift N$ for some $N$.
    \item If $\Phi,\Gamma;Q\vdash V:\circt(T,U)$ for some simple M-types $T$ and $U$, then $\Gamma=Q=\emptyset$ and $V\equiv (\vec\ell,D,\vec{\ell'})$ for some $\vec\ell,\vec{\ell'}$ and $D$.
\end{enumerate}
\end{theoremEnd}
\begin{proofEnd}
The claim follows immediately from the grammar for values and the rule system for types.
\end{proofEnd}

\begin{theoremEnd}{lem}[Type of Values]\label{value type}
Given a typing judgement $\Phi,\Gamma, Q\vdash V: A$, where $V$ is a value, then either one of the following holds:
\begin{itemize}
    \item $\Gamma = Q = \emptyset$.
    \item $A$ is a linear type.
\end{itemize}
\end{theoremEnd}
\begin{proofEnd}
By induction on the form of V:
\begin{itemize}
    \item Case $V\equiv \ell$. In this case, by Lemma \ref{generation lemma} we get $\Phi;\ell:\alpha\vdash \ell:\alpha$ and conclude that $A$ is a linear type.
    
    \item Case $V\equiv \lambda x. N$. In this case, by Lemma \ref{generation lemma} we get $\Phi,\Gamma;Q\vdash \lambda x. N: B\multimap C$ and conclude that $A$ is a linear type.
    
    \item Case $V\equiv \tuple{V_1,V_2}$. In this case, by Lemma \ref{generation lemma} we get $\Phi,\Gamma;Q\vdash \tuple{V_1,V_2}:B\otimes C$, $\Phi,\Gamma_1;Q_1\vdash V_1:B$ and $\Phi,\Gamma_2;Q_2\vdash V_2:C$, for some $\Gamma_1,\Gamma_2$ and $Q_1,Q_2$ such that $\Gamma_1,\Gamma_2=\Gamma$ and $Q_1,Q_2=Q$. By inductive hypothesis we know that either $B$ is a linear type or $\Gamma_1 = Q_1 = \emptyset$. In the former case, we conclude that $B\otimes C$ is also a linear type. In the latter case, we know by inductive hypothesis that either $C$ is a linear type or $\Gamma_2=Q_2=\emptyset$. In the former case, we conclude that $B\otimes C$ is also a linear type, while in the latter case we conclude that $\Gamma_1,\Gamma_2=\Gamma=\emptyset$ and $Q_1,Q_2=Q=\emptyset$.
    
    \item Case $V\equiv \lift N$. In this case, by Lemma \ref{generation lemma} we get $\Phi;\emptyset\vdash \lift N : \bang B$ and conclude that $\Gamma=Q =\emptyset$.
    
    \item Case $V\equiv (\vec\ell,C,\vec{\ell'})$. In this case, by Lemma \ref{generation lemma} we get $\Phi;\emptyset\vdash (\vec\ell,C,\vec{\ell'}) : \circt(T,U)$ and conclude that $\Gamma=Q =\emptyset$.
\end{itemize}
\end{proofEnd}

\begin{theoremEnd}{lem} [Parameter Substitution]
\label{parameter substitution}
Let $\Phi=\Phi',x:R$. If $\Phi,\Gamma;Q\vdash M:B$ and $\Phi';\emptyset\vdash V:R$, where $V$ is a value, then $\Phi',\Gamma;Q\vdash M[V/x]:B$.
\end{theoremEnd}
\begin{proofEnd}
By induction on the derivation of $\Phi,\Gamma;Q\vdash M:B$.
\begin{itemize}
    
    \item Case of \textit{var}. Suppose $M\equiv y$. If $x \not\equiv y$, we have $y[V/x]=y$ and the claim is trivially true. Otherwise, if $x\equiv y$, then $y[V/x]=x[V/x]=V$ and $\Phi';\emptyset\vdash V:R$ by hypothesis.
    
    \item Case of \textit{labels}. Suppose $M\equiv \vec\ell$. In this case $\vec\ell[V/x]= \ell'$ and the claim is trivially true.
    
    \item Case of \textit{abs}. Suppose $M\equiv \lambda y.N$. By Lemma \ref{generation lemma} we know that $\Phi,\Gamma,y:C;Q\vdash N:D$ for some $C$ and $D$ such that $B \equiv C\multimap D$. If $x\equiv y$ then by the definition of capture-free substitution we have $(\lambda x.N)[V/x]=\lambda x.N$ and the claim is trivially true. Otherwise, if $x\not\equiv y$, we have $(\lambda y.N)[V/x]=\lambda y .(N[V/x])$. In this case, by inductive hypothesis we get $\Phi',\Gamma,y:C;Q\vdash N[V/x] : D$ and conclude $\Phi',\Gamma;Q \vdash \lambda y.(N[V/x]) : C\multimap D$ by the \textit{abs} rule.
        
    \item Case of \textit{app}. Suppose $M\equiv NP$. In this case, $(NP)[V/x]=(N[V/x])(P[V/x])$. By Lemma \ref{generation lemma} we know that $\Phi,\Gamma_1;Q_1 \vdash N : C \multimap B$ and $\Phi,\Gamma_2;Q_2 \vdash P : C$, for some $C$ and for $\Gamma_1,\Gamma_2,Q_1,Q_2$ such that $\Gamma=\Gamma_1,\Gamma_2$ and $Q=Q_1,Q_2$. By inductive hypothesis we get $\Phi',\Gamma_1;Q_1 \vdash N[V/x] : C \multimap B$ and $\Phi',\Gamma_2;Q_2 \vdash P[V/x] : C$ and conclude $\Phi',\Gamma;Q\vdash (N[V/x])(P[V/x]) : B$ by the \textit{app} rule.
    
    \item Case of \textit{tuple}. Suppose $M\equiv \tuple{N,P}$. In this case, $\tuple{NP}[V/x]=\tuple{N[V/x],P[V/x]}$. By Lemma \ref{generation lemma} we know that $\Phi,\Gamma_1;Q_1 \vdash N : C$ and $\Phi,\Gamma_2;Q_2 \vdash P : D$, for some $C$ and $D$ such that $B\equiv C\otimes D$ and for $\Gamma_1,\Gamma_2,Q_1,Q_2$ such that $\Gamma=\Gamma_1,\Gamma_2$ and $Q=Q_1,Q_2$. By inductive hypothesis we get $\Phi',\Gamma_1;Q_1 \vdash N[V/x] : C$ and $\Phi',\Gamma_2;Q_2 \vdash P[V/x] : D$ and conclude $\Phi',\Gamma;Q\vdash \tuple{N[V/x],P[V/x]} : C\otimes D$ by the \textit{tuple} rule.
    
    \item Case of \textit{let}. Suppose $M\equiv \letin{\tuple{y,z}}{N}{P}$. By Lemma \ref{generation lemma} we know that $\Phi,\Gamma_1;Q_1 \vdash N : C\otimes D$ and $\Phi,\Gamma_2,y:C,z:D;Q_2 \vdash P : B$, for some $C$ and $D$ and for $\Gamma_1,\Gamma_2,Q_1,Q_2$ such that $\Gamma=\Gamma_1,\Gamma_2$ and $Q=Q_1,Q_2$. If $x\equiv y$ or $x\equiv z$ then by the definition of capture-free substitution we have $(\letin{\tuple{y,z}}{N}{P})[V/x]=\letin{\tuple{y,z}}{N[V/x]}{P}$. In this case, by inductive hypothesis we get $\Phi',\Gamma_1;Q_1\vdash N[V/x]:C\otimes D$ and conclude $\Phi',\Gamma;Q\vdash\letin{\tuple{y,z}}{N[V/x]}{P}:B$ by the \textit{tuple} rule. Otherwise, if $x\not\equiv y$ and $x\not\equiv z$, we have $(\letin{\tuple{y,z}}{N}{P})[V/x]=\letin{\tuple{y,z}}{N[V/x]}{P[V/x]}$. In this case, by inductive hypothesis we get both $\Phi',\Gamma_1;Q_1\vdash N[V/x]:C\otimes D$ and $\Phi',\Gamma_2,y:C,z:D;Q_2 \vdash P[V/x] : B$ and conclude $\Phi',\Gamma;Q\vdash\letin{\tuple{y,z}}{N[V/x]}{P[V/x]}:B$ by the \textit{let} rule.
        
    \item Case of \textit{lift}. Suppose $M\equiv \lift N$. In this case, $(\lift N)[V/x]= \lift (N[V/x])$ and $\Gamma = Q=\emptyset$ by Lemma \ref{generation lemma}. By the same lemma we know that $\Phi;\emptyset \vdash N : C$ for some $C$ such that $B\equiv \bang C$. By inductive hypothesis we get $\Phi';\emptyset \vdash N[V/x] : C$ and conclude $\Phi',\emptyset;Q' \vdash \lift (N[V/x]): \bang C$ by the \textit{lift} rule.
        
    \item Case of \textit{force}. Suppose $M\equiv \force N$. In this case, $(\force N)[V/x]=\force(N[V/x])$. By Lemma \ref{generation lemma} we know that $\Phi,\Gamma;Q \vdash N:\bang B$. By inductive hypothesis we get $\Phi',\Gamma;Q \vdash N[V/x]:\bang B$ and conclude $\Phi',\Gamma;Q \vdash \force (N[V/x]):B$ by the \textit{force} rule.
        
    \item Case of \textit{box}. Suppose $M\equiv \boxt_T N$. In this case, $(\boxt_T N)[V/x]=\boxt_T (N[V/x])$. By Lemma \ref{generation lemma} we know that $\Phi,\Gamma;Q\vdash N:\bang(T\multimap U)$ for some $T,U$ such that $B\equiv \circt(T,U)$. By inductive hypothesis we get $\Phi',\Gamma;Q\vdash N[V/x]:\bang(T\multimap U)$ and conclude $\Phi',\Gamma;Q \vdash \boxt_T (N[V/x]) : \circt(T,U)$ by the \textit{box} rule.
    
    \emergencystretch=10pt
    \item Case of \textit{apply}. Suppose $M\equiv \apply(N,P)$. In this case, we have $\apply(N,P)[V/x]=\apply(N[V/x],P[V/x])$. By Lemma \ref{generation lemma} we know that $\Phi,\Gamma_1;Q_1\vdash N : \circt(T,U)$ and $\Phi,\Gamma_2;Q_2\vdash P:T$, for some $T,U$ such that $B\equiv U$ and for $Q_1,Q_2,\Gamma_1,\Gamma_2$ such that $B\equiv U$ and $\Gamma=\Gamma_1,\Gamma_2$ and $Q=Q_1,Q_2$. By inductive hypothesis we get $\Phi',\Gamma_1;Q_1\vdash N[V/x] : \circt(T,U)$ and $\Phi',\Gamma_2;Q_2\vdash P[V/x]:T$ and conclude $\Phi',\Gamma;Q \vdash \apply(N[V/x],P[V/x]) : U$ by the \textit{apply} rule. 
        
    \item Case of \textit{circ}. Suppose $M\equiv (\vec\ell, D, \vec{\ell'})$. In this case $(\vec\ell, D,\vec{\ell'})[N/x]= (\vec\ell, D,\vec{\ell'})$ and the claim is trivially true.
\end{itemize}
\end{proofEnd}

\begin{theoremEnd}{lem} [Linear Substitution]
\label{linear substitution}
If $\Phi,\Gamma,x:A;Q\vdash M:B$ and $\Phi,\Gamma';Q'\vdash V:A$, where $A$ is a linear type and $V$ is a value, then $\Phi,\Gamma,\Gamma';Q,Q'\vdash M[V/x]:B$.
\end{theoremEnd}
\begin{proofEnd}
By induction on the derivation of $\Phi,\Gamma,x:A;Q\vdash M:B$.
\begin{itemize}
    \item Case of \textit{var}. Suppose $M\equiv y$. Here necessarily $y\equiv x$ (i.e. $x$ occurs free exactly once in $M$), since $\Gamma$ contains exclusively $x$ and therefore cannot possibly assign a type to any $y\not\equiv x$.
    In this case $x[V/x]=V$ and $\Phi,\Gamma';Q'\vdash V:A$ by hypothesis.
    
    \item Case of \textit{labels}. This case is impossible since it would entail $\Gamma = \emptyset$ by Lemma \ref{generation lemma}.
    
    \item Case of \textit{abs}. Suppose $M\equiv \lambda y.N$. In this case, $(\lambda y.N)[V/x]=\lambda y .(L[V/x])$. By Lemma \ref{generation lemma} we know that $\Phi,\Gamma,y:C;Q\vdash N : D$, for some $C,D$ such that $B \equiv C\multimap D$. By inductive hypothesis we get $\Phi,\Gamma,y:C,\Gamma';Q,Q'\vdash N[V/x] : D$ and conclude $\Phi,\Gamma,\Gamma';Q,Q' \vdash \lambda y.(N[V/x]) : C\multimap D$ by the \textit{abs} rule.
        
    \item Case of \textit{app}. Suppose $M\equiv NP$. By Lemma \ref{generation lemma} we know that $\Phi,\Gamma_1;Q_1 \vdash N : C \multimap B$ and $\Phi,\Gamma_2;Q_2 \vdash P : C$, for some $C$ and for $\Gamma_1,\Gamma_2,Q_1,Q_2$ such that $\Gamma,x:A=\Gamma_1,\Gamma_2$ and $Q=Q_1,Q_2$. Because $\Gamma_1$ and $\Gamma_2$ are disjoint, we have either $x\in \Gamma_1$ or $x\in \Gamma_2$. Let us assume, without loss of generality, that $x\in \Gamma_1$ and thus $\Gamma_1= \Gamma_1',x:A$ for some $\Gamma_1'$. In this case, $(NP)[V/x]=(N[V/x])P$. By the results of Lemma \ref{generation lemma} and the inductive hypothesis we get $\Phi,\Gamma_1',\Gamma';Q_1,Q'\vdash N[V/x]:C\multimap B$ and conclude $\Phi,\Gamma,\Gamma';Q,Q'\vdash (N[V/x])P : B$ by the \textit{app} rule.
    
    \item Case of \textit{tuple}. Suppose $M\equiv \tuple{N,P}$. By Lemma \ref{generation lemma} we know that $\Phi,\Gamma_1;Q_1 \vdash N : C$ and $\Phi,\Gamma_2;Q_2 \vdash P : D$, for some $C$ and $D$ such that $B\equiv C\otimes D$ and for $\Gamma_1,\Gamma_2,Q_1,Q_2$ such that $\Gamma=\Gamma_1,\Gamma_2$ and $Q=Q_1,Q_2$. Because $\Gamma_1$ and $\Gamma_2$ are disjoint, we have either $x\in \Gamma_1$ or $x\in \Gamma_2$. Let us assume, without loss of generality, that $x\in \Gamma_1$ and thus $\Gamma_1= \Gamma_1',x:A$ for some $\Gamma_1'$. In this case, $\tuple{NP}[V/x]=\tuple{N[V/x],P}$. By inductive hypothesis we get $\Phi,\Gamma_1',\Gamma';Q_1,Q' \vdash N[V/x] : C$ and conclude $\Phi,\Gamma,\Gamma';Q,Q'\vdash \tuple{N[V/x],P} : C\otimes D$ by the \textit{tuple} rule.
    
    \item Case of \textit{let}. Suppose $M\equiv \letin{\tuple{y,z}}{N}{P}$. By Lemma \ref{generation lemma} we know that $\Phi,\Gamma_1;Q_1 \vdash N : C\otimes D$ and $\Phi,\Gamma_2,y:C,z:D;Q_2 \vdash P : B$, for some $C$ and $D$ and for $\Gamma_1,\Gamma_2,Q_1,Q_2$ such that $\Gamma=\Gamma_1,\Gamma_2$ and $Q=Q_1,Q_2$. Because $\Gamma_1$ and $\Gamma_2$ are disjoint, we have either $x\in \Gamma_1$ or $x\in \Gamma_2$. Let us assume, without loss of generality, that $x\in \Gamma_1$ and thus $\Gamma_1= \Gamma_1',x:A$ for some $\Gamma_1'$. In this case, $(\letin{\tuple{y,z}}{N}{P})[V/x]=\letin{\tuple{y,z}}{N[V/x]}{P}$. By inductive hypothesis we get $\Phi,\Gamma_1',\Gamma';Q_1,Q'\vdash N[V/x]:C\otimes D$ and conclude $\Phi,\Gamma,\Gamma';Q,Q' \vdash \letin{\tuple{y,z}}{N[V/x]}{P} : B$ by the \textit{let} rule.
        
    \item Case of \textit{lift}. This case is impossible since it would entail $\Gamma=\emptyset$ by Lemma \ref{generation lemma}.
        
    \item Case of \textit{force}. Suppose $M\equiv \force N$. In this case, $(\force N)[V/x]=\force(N[V/x])$. By Lemma \ref{generation lemma} we know that $\Phi,\Gamma,x:A;Q \vdash N:\bang B$. By inductive hypothesis we get $\Phi,\Gamma,\Gamma';Q,Q' \vdash N[V/x]:\bang B$ and conclude $\Phi,\Gamma,\Gamma';Q,Q' \vdash \force (N[V/x]):B$
        
    \item Case of \textit{box}. Suppose $M\equiv \boxt_T N$. In this case, $(\boxt_T N)[V/x]=\boxt_T (N[V/x])$. By Lemma \ref{generation lemma} we know that $\Phi,\Gamma,x:A;Q\vdash N:\bang(T\multimap U)$ for some $T,U$ such that $B\equiv \circt(T,U)$. By inductive hypothesis we get $\Phi,\Gamma,\Gamma';Q,Q'\vdash N[V/x]:\bang(T\multimap U)$ and conclude $\Phi,\Gamma,\Gamma';Q,Q' \vdash \boxt_T (N[V/x]) : \circt(T,U)$ by the \textit{box} rule.
        
    \item Case of \textit{apply}. Suppose $M\equiv \apply(N,P)$. By Lemma \ref{generation lemma} we know that $\Phi,\Gamma_1,x:A;Q_1 \vdash N : \circt(T,U)$ and $\Phi,\Gamma_2;Q_2 \vdash P : T$, for some $T,U$ such that $B\equiv U$ and for $\Gamma_1,\Gamma_2,Q_1,Q_2$ such that $\Gamma,x:A=\Gamma_1,\Gamma_2$ and $Q=Q_1,Q_2$. Because $\Gamma_1$ and $\Gamma_2$ are disjoint, we have either $x\in \Gamma_1$ or $x\in \Gamma_2$. Let us assume, without loss of generality, that $x\in \Gamma_1$ and thus $\Gamma_1 = \Gamma_1',x:A$ for some $\Gamma_1'$. In this case, $\apply(N,P)[V/x]=\apply(N[V/x],P)$. By the results of Lemma \ref{generation lemma} and the inductive hypothesis we get $\Phi,\Gamma_1',\Gamma';Q_1,Q'\vdash N[V/x]:\circt(T,U)$ and conclude $\Phi,\Gamma,\Gamma';Q,Q'\vdash \apply(N[V/x],P) : U$ by the \textit{apply} rule.
        
    \item Case of \textit{circ}. This case is impossible since it would entail $\Gamma=\emptyset$ by Lemma \ref{generation lemma}.
\end{itemize}
\end{proofEnd}

\begin{thm}[Substitution]\label{substitution theorem}
If $\Phi,\Gamma,x:A;Q \vdash M:B$ and $\Phi,\Gamma';Q' \vdash V : B$,where $V$ is a value, then
$$\Phi,\Gamma,\Gamma';Q,Q' \vdash M[V/x]:B.$$
\end{thm}
\begin{proof}
The claim follows immediately from Lemma \ref{value type} and lemmata \ref{parameter substitution} and \ref{linear substitution}.
\end{proof}

\subsection{Big-step Operational Semantics}

We now review the operational semantics given by Rios and Selinger for Proto-Quipper-M, which, as the title suggests, are big-step. As we mentioned earlier, the evaluation of a Proto-Quipper-M program is intimately related to the circuit that the program is designed to build. Because of this, the operational semantics of the language is not defined on terms alone, but rather jointly on terms and circuits. To this effect, we give the definition of \textit{configuration}.

\begin{df}[Configuration]
A \emph{configuration} is a pair $(C,M)$, where $C$ is a circuit and $M$ is a term.
\end{df}

\noindent Intuitively, $C$ is the underlying circuit being built as a side-effect of the evaluation of $M$. We now proceed to give the following definitions and functions, which will be essential throughout the rest of our work.

\begin{df}[Equivalent Circuit]
Let $C:Q_1\to Q_1'$ and $D:Q_2\to Q_2'$ be two labelled circuits and let $(\vec\ell, C, \vec{\ell'})$ and $(\vec k, D, \vec{k'})$ be the corresponding boxed circuits. We say that $(\vec\ell, C, \vec{\ell'})$ and $(\vec k, D, \vec{k'})$ are \emph{equivalent} and we write $(\vec\ell, C, \vec{\ell'}) \cong (\vec k, D, \vec{k'})$ when they only differ by a renaming of labels, that is when $C=D$ in $\cat M$.
\end{df}

\begin{df}[$\freshlabels$]
Given a term $M$ and a simple M-type $T$, we define the function $\freshlabels$ as follows:
$$\freshlabels(M,T) = (Q,\vec\ell),$$
such that the labels in $\vec\ell$ do not occur in $M$ and $\emptyset;Q\vdash \vec\ell : T$.
\end{df}

\begin{df}[$\append$]
Let $C:Q_1\to Q_1'$ and $D:Q_2\to Q_2'$ be two labelled circuits and let $(\vec\ell_1, C, \vec{\ell_1'})$ and $(\vec\ell_2, D, \vec{\ell_2'})$ be the corresponding boxed circuits. Let $\vec k$ be a subset of the labels which occur in $\vec{\ell_1'}$. We define the function $\append$ as follows:
$$\append(C,\vec k,\vec\ell_2, D, \vec{\ell_2'})=(C', \vec{k'}),$$
where $C'$ is the circuit obtained by attaching the inputs of $D'$ to the matching outputs of $C$, for $(\vec k, D', \vec{k'})\cong(\vec\ell_2, D, \vec{\ell_2'})$. More formally, assume, without loss of generality, that $Q_1'$ is the concatenation of $Q_{11}'$ and $Q_{12}'$, where $Q_{12}'$ contains all and only the labels in $\vec k$. Then we have
$$C' = (id_{Q_{11}'} \otimes D') \circ C.$$
\end{df}

\begin{figure}[H]
    \centering
    \includegraphics[width=.4\textwidth]{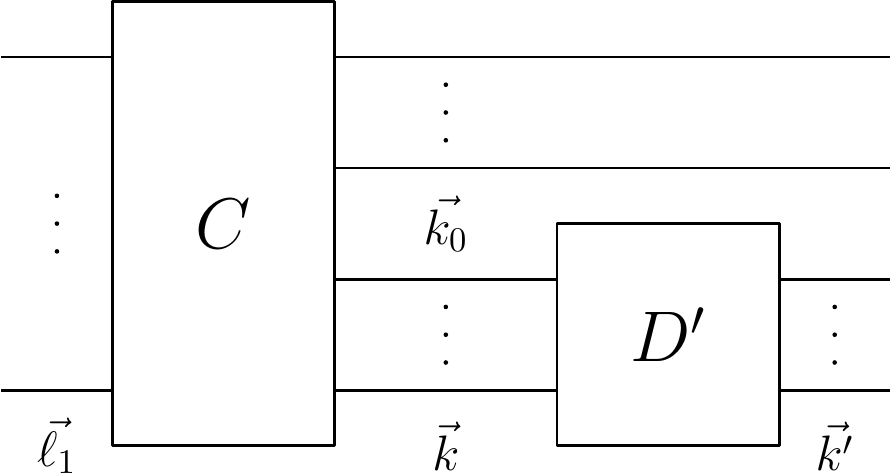}
\end{figure}

\noindent Now we have all the prerequisites for the definition of an operational semantics. We define $\Downarrow$ as a binary relation over configurations. Informally, $(C,M)\Downarrow(D,V)$ means that the evaluation of $M$ with an underlying circuit $C$ eventually results in value $V$ and in the construction of circuit $D$.

$$
\frac{\void}
{(C,x) \Downarrow \serr}
\qquad
\frac{\void}
{(C,\vec\ell) \Downarrow (C,\vec\ell)}
\qquad
\frac{\void}
{(C,\lambda x. M) \Downarrow (C,\lambda x. M)}
$$
\vspace{5pt}
$$
\frac{(C,M)\Downarrow(C_1,\lambda x.P) \quad (C_1, N) \Downarrow (C_2, V) \quad (C_2,P[V/x]) \Downarrow (C_3,W)}
{(C,MN) \Downarrow (C_3,W)}
$$
\vspace{5pt}
$$
\frac{(C,M)\Downarrow \sother}
{(C,MN)\Downarrow \serr}
\qquad
\frac{(C,M)\Downarrow(C_1,\tuple{V_1,V_2}) \quad (C_1, N[V_1/x][V_2/y]) \Downarrow (C_3,W)}
{(C,\letin{\tuple{x,y}}{M}{N}) \Downarrow (C_3,W)}
$$
\vspace{5pt}
$$
\frac{(C,M)\Downarrow \sother}
{(C,\letin{\tuple{x,y}}{M}{N}) \Downarrow \serr}
\qquad
\frac{(C,M)\Downarrow(C_1,V) \quad (C_1,N)\Downarrow (C_2,W)}
{(C,\tuple{M,N}) \Downarrow (C_2,\tuple{V,W})}
$$
\vspace{5pt}
$$
\frac{\void}
{(C,\lift M)\Downarrow(C,\lift M)}
\qquad
\frac{(C,M)\Downarrow(C_1,\lift N) \quad (C_1, N) \Downarrow (C_2,V)}
{(C,\force M) \Downarrow (C_2,V)}
$$
\vspace{5pt}
$$
\frac{(C,M)\Downarrow\sother}
{(C,\force M) \Downarrow \serr}
$$
\vspace{5pt}
$$
\frac{(C,M)\Downarrow(C_1,\lift N) \quad (Q,\vec\ell)=\freshlabels(N,T) \quad (id_Q,N\vec\ell) \Downarrow (D,\vec{\ell'})}
{(C,\boxt_T M) \Downarrow (C_1,(\vec\ell,D,\vec{\ell'}))}
$$
\vspace{5pt}
$$
\frac{(C,M)\Downarrow \sother}
{(C,\boxt_T M) \Downarrow \serr}
\qquad
\frac{(C,M)\Downarrow \sother}
{(C,\apply(M,N)) \Downarrow \serr}
\qquad
\frac{\void}
{(C,(\vec\ell,D,\vec{\ell'})) \Downarrow (C,(\vec\ell,D,\vec{\ell'}))}
$$
\vspace{5pt}
$$
\frac{(C,M)\Downarrow(C_1,(\vec\ell,D,\vec{\ell'})) \quad (C_1, N) \Downarrow (C_2, \vec k) \quad (C_3,\vec{k'}) = \append(C_2,\vec k, \vec\ell, D, \vec{\ell'})}
{(C,\apply(M,N)) \Downarrow (C_3,\vec{k'})}
$$
\vspace{5pt}

\noindent where the notation ``$(C,M)\Downarrow\sother$'' is shorthand for $(C,M)\Downarrow(D,V)$ where $V$ does not match the explicit form required by any other rule that evaluates the same configuration.

\subsection{Small-step Operational Semantics}

It is now time to take the first step towards a machine semantics for Proto-Quipper-M. In this section, we extrapolate an equivalent small-step semantics from the big-step semantics that we just saw, and we examine its properties and its limitations. For the sake of simplicity, from now on we will assume that the terms in the configurations we work with do not contain free variables.

\begin{df}[Small-step Configuration]
A \emph{small-step configuration} is a pair of the form $(C,M)$, where $C$ is a circuit and $M$ is a term with no free variables.
\end{df}

\noindent We then define a binary reduction relation $\red$ on configurations. Informally, $(C,M)\red (D,N)$ means that $M$ evaluates to $N$ in one step, in a way that updates the underlying circuit from $C$ to $D$. In order to mimic the behavior of the big-step semantics, we design two sets of rules. The rules in the first set each resemble one of the main rules of the original big-step semantics and operate directly on redexes.
$$
\frac{\void}
{(C,(\lambda x.M)V) \to (C,M[V/x])}\textit{$\beta$-reduction}
$$
\vspace{5pt}
$$
\frac{\void}
{(C,\letin{\tuple{x,y}}{\tuple{V,W}}{M})\red (C,M[V/x][W/y])}\textit{let}
$$
\vspace{5pt}
$$
\frac{\void}
{(C,\force  (\lift M)) \to (C,M)}\textit{force}
\qquad
\frac{(Q,\vec\ell)=\freshlabels(M,T)\quad (id_Q,M\vec\ell) \red\dots\red (D,\vec{\ell'})}
{(C,\boxt_T(\lift M)) \red (C,(\vec\ell,D,\vec{\ell'}))}
\textit{box}
$$
\vspace{5pt}
$$
\frac{(C',\vec k')= \append(C,\vec k, \vec \ell, D, \vec\ell')}
{(C,\apply ((\vec\ell,D,\vec{\ell'}),\vec k)) \to (C',\vec k')}
\textit{apply}
$$
\vspace{5pt}

\noindent where the notation $(C,M)\red\dots\red(C',M')$ is shorthand for
\begin{align*}
    (C,M) \equiv\; &(C_1,M_1),(C_1,M_1) \red(C_2,M_2), \dots,\\
    &(C_{n-1},M_{n-1})\red(C_n,M_n), (C_n,M_n) \equiv (C',M'),
\end{align*}
for some $n>0$. That is, a reduction sequence of finite length from $(C,M)$ to $(C',M')$. We employ a finite, but arbitrary number of premises instead of a single premise with the transitive and reflexive closure $\red^*$ of the reduction relation $\red$ in order to make proofs by induction easier. We now examine the second set of rules, which we call \textit{contextual rules}. Each one of these rules allows for the intermediate evaluation of an immediate sub-term within a term.
$$
\frac{(C,M)\red (C',M')}
{(C,MN)\red (C',M'N)}\textit{ctx-app-left}
\qquad
\frac{(C,M)\red (C',M')}
{(C,VM)\red (C',VM')}\textit{ctx-app-right}
$$
\vspace{5pt}
$$
\frac{(C,M)\red (C',M')}
{(C,\tuple{M,N})\red (C',\tuple{M',N})}\textit{ctx-tuple-left}
\qquad
\frac{(C,M)\red (C',M')}
{(C,\tuple{V,M})\red (C',\tuple{V,M'})}\textit{ctx-tuple-right}
$$
\vspace{5pt}
$$
\frac{(C,M)\red(C',M')}
{(C,\letin{\tuple{x,y}}{M}{N})\red (C',\letin{\tuple{x,y}}{M'}{N})}\textit{ctx-let}
$$
\vspace{5pt}
$$
\frac{(C,M)\red (C',M')}
{(C,\force M)\red (C',\force M')}\textit{ctx-force}
\qquad
\frac{(C,M)\red (C',M')}
{(C,\boxt_T M)\red (C',\boxt_T M')}\textit{ctx-box}
$$
\vspace{5pt}
$$
\frac{(C,M)\red (C',M')}
{(C,\apply(M,N))\red (C',\apply(M',N))}\textit{ctx-apply-left}
$$
\vspace{5pt}
$$
\frac{(C,M)\red (C',M')}
{(C,\apply(V,M))\red (C',\apply(V,M'))}\textit{ctx-apply-right}
$$

\paragraph{} The way we introduced it, the reduction relation $\red$ is deterministic, as stated in the following results.

\begin{theoremEnd}{lem}\label{small-step rule mutex}
Every small-step configuration $(C,M)$ can be reduced by at most one rule of the small-step operational semantics.
\end{theoremEnd}
\begin{proofEnd}
We proceed by cases on $M$:
\begin{itemize}
    \item Case $M\equiv x$. This case is impossible, since by the definition of small-step configuration $M$ must contain no free variables.
    
    \item Case $M\equiv \vec\ell$. In this case $\vec\ell$ is a value and $(C,\vec\ell)$ cannot be reduced by any rule.
    
    \item Case $M\equiv \lambda x.N$. In this case $\lambda x.N$ is a value and $(C,\lambda x.N)$ cannot be reduced by any rule.
    
    \item Case $M\equiv NP$. In this case $NP$ is matched by the \textit{$\beta$-reduction}, \textit{ctx-app-left} and \textit{ctx-app-right} rules. However, \textit{$\beta$-reduction} requires that both $N$ and $P$ be values to be applied, while \textit{ctx-app-left} requires that $N$ be reducible (and therefore not a value) and \textit{ctx-app-right} requires that $N$ be a value and $P$ be reducible (and therefore not a value). Because these conditions mutually exclude each other, we conclude that $(C,NP)$ can be reduced by at most one rule.
    
    \item Case $M\equiv \tuple{N,P}$. In this case $\tuple{N,P}$ is matched by the \textit{ctx-tuple-left} and \textit{ctx-tuple-right} rules. However, the latter requires that $N$ be a value to be applied, while the former requires that it be reducible (and therefore not a value). Because these conditions mutually exclude each other, we conclude that $(C,\tuple{N,P})$ can be reduced by at most one rule.
    
    \item Case $M\equiv \letin{\tuple{x,y}}{N}{P}$. In this case $\letin{\tuple{x,y}}{N}{P}$ is matched by the \textit{let} and \textit{ctx-let} rules. However, the former requires that $N$ be a value to be applied, while the latter requires that it be reducible (and therefore not a value). Because these conditions mutually exclude each other, we conclude that $(C,\letin{\tuple{x,y}}{N}{P})$ can be reduced by at most one rule.
    
    \item Case $M\equiv \lift N$. In this case $\lift N$ is a value and $(C,\lift N)$ cannot be reduced by any rule.
    
    \item Case $M\equiv \force N$. In this case $\force N$ is matched by the \textit{force} and \textit{ctx-force} rules. However, the former requires that $N$ be a value to be applied, while the latter requires that it be reducible (and therefore not a value). Because these conditions mutually exclude each other, we conclude that $(C,\force N)$ can be reduced by at most one rule.
    
    \item Case $M\equiv \boxt_T N$. In this case $\boxt_T N$ is matched by the \textit{box} and \textit{ctx-box} rules. However, the former requires that $N$ be a value to be applied, while the latter requires that it be reducible (and therefore not a value). Because these conditions mutually exclude each other, we conclude that $(C,\boxt_T N)$ can be reduced by at most one rule.
    
    \item Case $M\equiv \apply(N,P)$. In this case $\apply(N,P)$ is matched by the \textit{apply}, \textit{ctx-apply-left} and \textit{ctx-apply-right} rules. However, \textit{apply} requires that both $N$ and $P$ be values to be applied, while \textit{ctx-apply-left} requires that $N$ be reducible (and therefore not a value) and \textit{ctx-apply-right} requires that $N$ be a value and $P$ be reducible (and therefore not a value). Because these conditions mutually exclude each other, we conclude that $(C,\apply(N,P))$ can be reduced by at most one rule.
    
    \item Case $M\equiv (\vec\ell,D,\vec{\ell'})$. In this case $(\vec\ell,D,\vec{\ell'}$ is a value and $(C,(\vec\ell,D,\vec{\ell'})$ cannot be reduced by any rule.
\end{itemize}
\end{proofEnd}

\begin{prop}[Determinism of Small-step Semantics]\label{small-step determinism}
The reduction relation $\red$ is deterministic. That is, if $(C,M) \red (D,N)$, then for every configuration $(D',N')$ such that $(C,M) \red (D',N')$ we have $D=D'$ and $N\equiv N'$.
\end{prop}
\begin{proof}
We already known by Lemma \ref{small-step rule mutex} that at most one rule can be applied to reduce any given configuration. What is left to do is prove that each rule is deterministic by itself, which is done trivially by induction on the derivation of $(C,M) \red (D,N)$.
\end{proof}

\subsubsection{Evaluation Contexts}

The contextual rules are defined in a recursive fashion, which means that redexes can be reduced at an arbitrary depth within a term through multiple rule applications. In order to be able to reason about all the valid positions where a reduction may occur within a term, we introduce the notion of \textit{evaluation context}. 

\begin{df}[Evaluation Context]
An evaluation context defines where, within a term, we can reduce a sub-term. Formally, an evaluation context is a function defined by the following grammar:
\begin{align*}
E,F::=\;&[\cdot] \mid EM \mid VE \mid \tuple{E,M} \mid \tuple{V,E} \mid \letin{\tuple{x,y}}{E}{N} \\
&\mid \force E \mid \boxt_T E \mid \apply(E, M) \mid \apply(V, E),
\end{align*}

and the following semantics:
\begin{align*}
    [\cdot][M]     &= M\\
    (EN)[M]        &= (E[M])N\\
    (VE)[M]        &= V(E[M])\\
    \tuple{E,N}[M] &= \tuple{E[M],N}\\
    \tuple{V,E}[M] &= \tuple{V,E[M]}\\
    (\letin{\tuple{x,y}}{E}{N})[M]&=\letin{\tuple{x,y}}{E[M]}{N}\\
    (\force E)[M]  &= \force E[M]\\
    (\boxt_T E)[M] &= \boxt_T E[M]\\
    \apply(E,N)[M] &= \apply(E[M],N)\\
    \apply(V,E)[M] &= \apply(V,E[M]).
\end{align*}
\end{df}

\noindent An evaluation context is in fact a function from terms to terms. However, it is perhaps more intuitively to understand an evaluation context as an incomplete term with a hole in which we can stick different sub-terms. The way the grammar for evaluation contexts is designed guarantees that whenever a term is reducible on its own, then it is also reducible when we stick it in an evaluation context, and vice-versa.

\begin{theoremEnd}{thm}[Fundamental Theorem of Evaluation Contexts]\label{context reduction lemma}
For every evaluation context $E$, we have that $(C,M)\red(C',M')$ if and only if $(C,E[M])\red(C',E[M'])$.
\end{theoremEnd}
\begin{proofEnd}
We first prove that if $(C,M)\red(C',M')$, then $(C,E[M])\red(C',E[M'])$. We proceed by induction on the form of $E$:
\begin{itemize}
    \item Case $E\equiv [\cdot]$. In this case $E[M]\equiv M$ and the claim is trivially true.
    \item Case $E\equiv FN$. In this case $E[M] \equiv (F[M])N$. By inductive hypothesis we know that $(C,F[M]) \red (C',F[M'])$, so we conclude $(C,(F[M])N) \red (C',(F[M'])N)$ by the \textit{ctx-app-left} rule.
    
    \item Case $E\equiv VF$. In this case $E[M] \equiv V(F[M])$. By inductive hypothesis we know that $(C,F[M]) \red (C',F[M'])$, so we conclude $(C,V(F[M])) \red (C',V(F[M']))$ by the \textit{ctx-app-right} rule.
    
    \item Case $E\equiv \tuple{F,N}$. In this case $E[M] \equiv \tuple{F[M],N}$. By inductive hypothesis we know that $(C,F[M]) \red (C',F[M'])$, so we conclude $(C,\tuple{F[M],N}) \red (C',\tuple{F[M'],N})$ by the \textit{ctx-tuple-left} rule.
    
    \item Case $E\equiv \tuple{V,F}$. In this case $E[M] \equiv \tuple{V,F[M]}$. By inductive hypothesis we know that $(C,F[M]) \red (C',F[M'])$, so we conclude $(C,\tuple{V,F[M]}) \red (C',\tuple{V,F[M]})$ by the \textit{ctx-tuple-right} rule.
    
    \item Case $E\equiv \letin{\tuple{x,y}}{F}{N}$. In this case $E[M]= \letin{\tuple{x,y}}{F[M]}{N}$. By inductive hypothesis we know that $(C,F[M]) \red (C',F[M'])$, so we conclude $(C,\letin{\tuple{x,y}}{F[M]}{N}) \red (C',\letin{\tuple{x,y}}{F[M']}{N})$ by the \textit{ctx-let} rule.
    
    \item Case $E\equiv \force F$. In this case $E[M]= \force F[M]$. By inductive hypothesis we know that $(C,F[M]) \red (C',F[M'])$, so we conclude $(C, \force F[M]) \red (C', \force F[M'])$ by the \textit{ctx-force} rule.
    
    \item Case $E\equiv \boxt_T F$. In this case $E[M]= \boxt_T F[M]$. By inductive hypothesis we know that $(C,F[M]) \red (C',F[M'])$, so we conclude $(C, \boxt_T F[M]) \red (C', \boxt_T F[M'])$ by the \textit{ctx-box} rule.
    
    \emergencystretch=10pt
    \item Case $E\equiv \apply(F,N)$. In this case $E[M] \equiv \apply(F[M],N)$. By inductive hypothesis we know that $(C,F[M]) \red (C',F[M'])$, so we conclude that $(C,\apply(F[M],N)) \red (C',\apply(F[M'],N))$ by the \textit{ctx-apply-left} rule.
    
    \item Case $E\equiv \apply(V,F)$. In this case $E[M] \equiv \apply(V, F[M])$. By inductive hypothesis we know that $(C,F[M]) \red (C',F[M'])$, so we conclude that $(C,\apply(V, F[M])) \red (C',\apply(V, F[M']))$ by the \textit{ctx-apply-right} rule
\end{itemize}

\noindent Now we prove that if $(C,E[M])\red(C',E[M'])$, then $(C,M)\red(C',M')$. We again proceed by induction on the form of $E$:

\begin{itemize}
    \item Case $E \equiv [\cdot]$. In this case $E[M]\equiv M$ and the claim is trivially true.
    
    \item Case $E\equiv FN$. In this case we have
    $$
    \frac{(C,F[M]) \red (C',F[M'])}
    {(C,(F[M])N) \red (C',(F[M'])N)
    }\textit{ctx-app-left}
    $$
    and by inductive hypothesis on $F$ we immediately conclude $(C,M) \red (C',M')$.
    
    \item Case $E\equiv VF$. In this case we have
    $$
    \frac{(C,F[M]) \red (C',F[M'])}
    {(C,V(F[M])) \red (C',V(F[M']))
    }\textit{ctx-app-right}
    $$
    and by inductive hypothesis on $F$ we immediately conclude $(C,M) \red (C',M')$.
    
    \item Case $E\equiv \tuple{F,N}$. In this case we have
    $$
    \frac{(C,F[M]) \red (C',F[M'])}
    {(C,\tuple{F[M],N}) \red (C',\tuple{F[M'],N})
    }\textit{ctx-tuple-left}
    $$
    and by inductive hypothesis on $F$ we immediately conclude $(C,M) \red (C',M')$.
    
    \item Case $E\equiv \tuple{V,F}$. In this case we have 
    $$
    \frac{(C,F[M]) \red (C',F[M'])}
    {(C,\tuple{V,F[M]}) \red (C',\tuple{V,F[M']})
    }\textit{ctx-tuple-right}
    $$
    and by inductive hypothesis on $F$ we immediately conclude $(C,M) \red (C',M')$.
    
    \item Case $E\equiv \letin{\tuple{x,y}}{F}{N}$. In this case  we have
    $$
    \frac{(C,F[M]) \red (C',F[M'])}
    {(C,\letin{\tuple{x,y}}{F[M]}{N}) \red (C',\letin{\tuple{x,y}}{F[M']}{N})
    }\textit{ctx-let}
    $$
    and by inductive hypothesis on $F$ we immediately conclude $(C,M) \red (C',M')$.
    
    \item Case $E\equiv \force F$. In this case we have
    $$
    \frac{(C,F[M]) \red (C',F[M'])}
    {(C,\force F[M]) \red (C',\force F[M'])
    }\textit{ctx-force}
    $$
    and by inductive hypothesis on $F$ we immediately conclude $(C,M) \red (C',M')$.
    
    \item Case $E\equiv \boxt_T F$. In this case we have
    $$
    \frac{(C,F[M]) \red (C',F[M'])}
    {(C,\boxt_T F[M]) \red (C',\boxt_T F[M'])
    }\textit{ctx-box}
    $$
    and by inductive hypothesis on $F$ we immediately conclude $(C,M) \red (C',M')$.
    
    \item Case $E\equiv \apply(F,N)$. In this case we have
    $$
    \frac{(C,F[M]) \red (C',F[M'])}
    {(C,\apply(F[M],N)) \red (C',\apply(F[M'],N))
    }\textit{ctx-apply-left}
    $$
    and by inductive hypothesis on $F$ we immediately conclude $(C,M) \red (C',M')$.
    
    \item Case $E\equiv \apply(V,F)$. In this case we have
    $$
    \frac{(C,F[M]) \red (C',F[M'])}
    {(C,\apply(V,F[M])) \red (C',\apply(V,F[M']))
    }\textit{ctx-apply-left}
    $$
    and by inductive hypothesis on $F$ we immediately conclude $(C,M) \red (C',M')$.
    
\end{itemize}
\end{proofEnd}

\begin{cor}\label{subterm reducibility}
Suppose we have a term $M$ of the form $E[N]$ for some $E$. Then $(C,M)$ is reducible if and only if $(C,N)$ is reducible.
\end{cor}

\begin{cor}\label{subterm irreducibility}
Suppose we have a term $M$ of the form $E[N]$ for some $E$. Then $(C,M)$ is irreducible if and only if $(C,N)$ is irreducible.
\end{cor}

In the light of the Fundamental Theorem of Evaluation Contexts, we can easily understand that if we have a term of the form $E[M]$, where $M$ is a redex, then $E[M]\not\equiv F[N]$ for any other evaluation context $F$ and redex $N$, otherwise we could choose to reduce $M$ or $N$, leading to different configurations and breaking determinism. In fact, we prove that this result does not only hold for redexes, but more generally for \textit{proto-redexes}. Intuitively, a proto-redex is a redex without constraints on the explicit form of the values that make it up. More formally, we give the following definition.

\begin{df}[Proto-redex]
A term is said to be a \emph{proto-redex} when it can be generated by the following grammar:
$$
P::= VW \mid \letin{\tuple{x,y}}{V}{N} \mid \force V \mid \boxt_T V \mid \apply(V,W).
$$
\end{df}

\begin{theoremEnd}{lem}\label{context identity lemma}
Suppose $M$ and $N$ are proto-redexes. If $E[M]\equiv E'[N]$, then $E\equiv E'$.
\end{theoremEnd}
\begin{proofEnd}
By induction on the form of $E$:
\begin{itemize}
    \item Case $E \equiv [\cdot]$. In this case $E[M]\equiv M \equiv E'[N]$. Suppose $E' \not\equiv [\cdot]$. We proceed by cases on $M$:
    \begin{itemize}
        \item $M\equiv VW$. In this case the only possibilities are $E'\equiv [\cdot]W$ and $E'\equiv V[\cdot]$. However, both would imply that $N$ is a value, which is impossible since $N$ is a proto-redex.
        
        \item $M\equiv \letin{\tuple{x,y}}{V}{P}$. In this case the only possibility is $E'\equiv \letin{\tuple{x,y}}{[\cdot]}{P}$. However, this would imply that $N$ is a value, which is impossible since $N$ is a proto-redex.
        
        \item $M\equiv \force(V)$. In this case the only possibility is $E'\equiv \force \,[\cdot]$. However, this would imply that $N$ is a value, which is impossible since $N$ is a proto-redex.
        
        \item $M\equiv \boxt_T(V)$. In this case the only possibility is $E'\equiv \boxt_T \,[\cdot]$. However, this would imply that $N$ is a value, which is impossible since $N$ is a proto-redex.
        
        \item $M\equiv \apply(V,W)$. In this case the only possibilities are $E'\equiv \apply([\cdot],W)$ and $E'\equiv \apply(V,[\cdot])$. However, both would imply that $N$ is a value, which is impossible since $N$ is a proto-redex.
    \end{itemize}
    Since all possibilities lead to a contradiction, we conclude that $E'\equiv [\cdot]\equiv E$.
    
    \item $E \equiv FP$. In this case $E[M]\equiv (F[M])P \equiv E'[N]$ and the two possibilities are $E'\equiv E''P$ and $E' \equiv (F[M])E''$. In the former case we have $E''[N]\equiv F[M]$, so by inductive hypothesis we get $E''\equiv F$ and conclude $E'\equiv FP \equiv E$. The latter case would require that $F[M]$ be a value, which is impossible since $M$ is a proto-redex.
    
    \item $E \equiv VF$. In this case $E[M]\equiv V(F[M]) \equiv E'[N]$ and the two possibilities are $E'\equiv E''(F[M])$ and $E' \equiv VE''$. The former would require that $E''[N]$ be a value, which is impossible since $N$ is a proto-redex. In the latter case we have $E''[N]\equiv F[M]$, so by inductive hypothesis we get $E''\equiv F$ and conclude $E'\equiv VF \equiv E$.
    
    \item $E \equiv \tuple{F,P}$. In this case $E[M]\equiv \tuple{F[M],P} \equiv E'[N]$ and the two possibilities are $E'\equiv \tuple{E'',P}$ and $E' \equiv \tuple{F[M],E''}$. In the former case we have $E''[N]\equiv F[M]$, so by inductive hypothesis we get $E''\equiv F$ and conclude $E'\equiv \tuple{F,P} \equiv E$. The latter case would require that $F[M]$ be a value, which is impossible since $M$ is a proto-redex.
    
    \item $E \equiv \tuple{V,F}$. In this case $E[M]\equiv \tuple{V,F[M]} \equiv E'[N]$ and the two possibilities are $E'\equiv \tuple{E'',F[M]}$ and $E' \equiv \tuple{V,E''}$. The former would require that $E''[N]$ be a value, which is impossible since $N$ is a proto-redex. In the latter case we have $E''[N]\equiv F[M]$, so by inductive hypothesis we get $E''\equiv F$ and conclude $E'\equiv \tuple{V,F} \equiv E$.
    
    \item $E \equiv \letin{\tuple{x,y}}{F}{P}$. In this case $E[M] \equiv \letin{\tuple{x,y}}{F[M]}{P}\equiv E'[N]$ and the only possibility is $E'\equiv \letin{\tuple{x,y}}{E''}{P}$. This implies $E''[N]\equiv F[M]$, so by inductive hypothesis we get $E'' \equiv F$ and conclude $E' \equiv \letin{\tuple{x,y}}{F}{P}\equiv E$.
    
    \item $E \equiv \force F$. In this case $E[M] \equiv \force F[M]\equiv E'[N]$ and the only possibility is $E'\equiv \force E''$. This implies $E''[N]\equiv F[M]$, so by inductive hypothesis we get $E'' \equiv F$ and conclude $E' \equiv \force F \equiv E$.
    
    \item $E \equiv \boxt_T F$. In this case $E[M] \equiv \boxt_T F[M]\equiv E'[N]$ and the only possibility is $E'\equiv \boxt_T E''$. This implies $E''[N]\equiv F[M]$, so by inductive hypothesis we get $E'' \equiv F$ and conclude $E' \equiv \boxt_T F \equiv E$.
    
    \item $E \equiv \apply(F,P)$. In this case $E[M]\equiv \apply(F[M],P) \equiv E'[N]$ and the two possibilities are $E'\equiv \apply(E'',P)$ and $E' \equiv \apply(F[M],E'')$. In the former case we have $E''[N]\equiv F[M]$, so by inductive hypothesis we get $E''\equiv F$ and conclude $E'\equiv \apply(F,P) \equiv E$. The latter case would require that $F[M]$ be a value, which is impossible since $M$ is a proto-redex.
    
    \item $E \equiv \apply(V,F)$. In this case $E[M]\equiv \apply(V,F[M]) \equiv E'[N]$ and the two possibilities are $E'\equiv \apply(E'',F[M])$ and $E' \equiv \apply(V,E'')$. The former would require that $E''[N]$ be a value, which is impossible since $N$ is a proto-redex. In the latter case we have $E''[N]\equiv F[M]$, so by inductive hypothesis we get $E''\equiv F$ and conclude $E'\equiv \apply(V,F) \equiv E$.
    
\end{itemize}
\end{proofEnd}

\begin{lem}\label{term identity lemma}
If $E[M]\equiv E[N]$, then $M\equiv N$
\end{lem}
\begin{proof}
The claim follows naturally from the definition of evaluation context. The proof is trivial by induction on $E$.
\end{proof}

\begin{prop}[Context Exclusivity]
\label{context exclusion theorem proto}
Suppose $M$ and $N$ are proto-redexes. If $M\not\equiv N$, then $E[M] \not\equiv F[N]$, for any $E,F$.
\end{prop}
\begin{proof}
Suppose $M$ and $N$ are proto-redexes. If $E[M] \equiv F[N]$ for some $E,F$, then by Lemma \ref{context identity lemma} we get $E\equiv F$. Because $E[M] \equiv F[N] \equiv E[N]$, by Lemma \ref{term identity lemma} we also get $M\equiv N$, so we know that $E[M]\equiv F[N]$ entails $M\equiv N$. From this we conclude that if $M \not\equiv N$, then $E[M] \not\equiv F[N]$ for any $E,F$.
\end{proof}

\begin{cor}\label{context exclusion lemma redex}
Suppose $M$ and $N$ are redexes. If $M\not\equiv N$, then $E[M]\not\equiv F[N]$, for any $E,F$.
\end{cor}
\begin{proof}
The claim follows naturally from the fact that every redex is also a proto-redex. This is obvious, as every redex can be obtained from a production of the proto-redex grammar by instantiating the generic values occurring in the latter with the explicit forms required by the former.
\end{proof}

\paragraph{} We give one last result about evaluation contexts: if we ``inject'' an evaluation context into a second evaluation context, the result is still an evaluation context.

\begin{theoremEnd}{prop}[Context Injection]
\label{context propagation}
Suppose we have a term $M$ of the form $E[N]$ for some $E$. If $N$ is of the form $E'[L]$ for some $E'$, then $M$ is of the form $E''[L]$ for some $E''$.
\end{theoremEnd}
\begin{proofEnd}
By induction on the form of $E$:
\begin{itemize}
    \item Case $E\equiv[\cdot]$. In this case $M\equiv N$ and the claim is trivially true.
    
    \item Case $E\equiv FP$. In this case $M\equiv (F[N])P$. By inductive hypothesis we know that $F[N]$ is of the form $F'[L]$ for some $F'$, so $M$ is of the form $E''[L]$ for $E''\equiv F'P$.
    
    \item Case $E\equiv VF$. In this case $M\equiv V(F[N])$. By inductive hypothesis we know that $F[N]$ is of the form $F'[L]$ for some $F'$, so $M$ is of the form $E''[L]$ for $E''\equiv VF'$.
    
    \item Case $E\equiv \tuple{F,P}$. In this case $M\equiv \tuple{F[N],P}$. By inductive hypothesis we know that $F[N]$ is of the form $F'[L]$ for some $F'$, so $M$ is of the form $E''[L]$ for $E''\equiv \tuple{F',P}$.
    
    \item Case $E\equiv \tuple{V,F}$. In this case $M\equiv \tuple{V,F[N]}$. By inductive hypothesis we know that $F[N]$ is of the form $F'[L]$ for some $F'$, so $M$ is of the form $E''[L]$ for $E''\equiv \tuple{V,F'}$.
    
    \item Case $E\equiv \letin{\tuple{x,y}}{F}{P}$. In this case $M\equiv \letin{\tuple{x,y}}{F[N]}{P}$. By inductive hypothesis we know that $F[N]$ is of the form $F'[L]$ for some $F'$, so $M$ is of the form $E''[L]$ for $E''\equiv \letin{\tuple{x,y}}{F'}{P}$.
    
    \item Case $E\equiv \force F$. In this case $M\equiv \force F[N]$. By inductive hypothesis we know that $F[N]$ is of the form $F'[L]$ for some $F'$, so $M$ is of the form $E''[L]$ for $E''\equiv \force F'$.
    
    \item Case $E\equiv \boxt_T F$. In this case $M\equiv \boxt_T F[N]$. By inductive hypothesis we know that $F[N]$ is of the form $F'[L]$ for some $F'$, so $M$ is of the form $E''[L]$ for $E''\equiv \boxt_T F'$.
    
    \item Case $E\equiv \apply(F,P)$. In this case $M\equiv \apply(F[N],P)$. By inductive hypothesis we know that $F[N]$ is of the form $F'[L]$ for some $F'$, so $M$ is of the form $E''[L]$ for $E''\equiv \apply(F',P)$.
    
    \item Case $E\equiv \apply(V,F)$. In this case $M\equiv \apply(V,F[N])$. By inductive hypothesis we know that $F[N]$ is of the form $F'[L]$ for some $F'$, so $M$ is of the form $E''[L]$ for $E''\equiv \apply(V,F')$.
    
\end{itemize}
\end{proofEnd}

\begin{cor}\label{context propagation corollary}
Suppose we have a term $M$ of the form $E[N]$ for some $E$. If $M$ is not of the form $E'[L]$ for any $E'$, then $N$ is not of the form $E''[L]$ for any $E''$.
\end{cor}

\subsubsection{Convergence, Deadlock and Divergence}
\label{small-step convergence, deadlock and divergence}

Now that we have (deterministic) small-step semantics and evaluation contexts, we can start distinguishing between \textit{converging}, \textit{deadlocking} and \textit{diverging} configurations. Informally, a configuration converges if its evaluation terminates successfully returning a value, it goes into deadlock if it gets stuck without returning a value, and it diverges if its evaluation does not terminate at all. The first two definitions are standard inductive definitions.

\begin{df}[Converging Small-step Configuration]
Let $\converges$ be the smallest unary relation over small-step configurations such that:
\begin{enumerate}
    \item For every circuit $C$ and value $V$, $(C,V)\converges$,
    \item If $(C,M)\red (D,N)$ and $(D,N)\converges$, then $(C,M)\converges$.
\end{enumerate}
We say that a configuration $(C,M)$ is \emph{converging} when $(C,M)\converges$.
\end{df}

\begin{df}[Deadlocking Small-step Configuration]
\label{small-step deadlocking}
Let $\deadlocks$ be the smallest unary relation over small-step configurations such that:
\begin{enumerate}
    \item If $(C,M)$ is irreducible and $M$ is neither of the form $E[\boxt_T(\lift N)]$, nor a value, then $(C,M)\deadlocks$.
    
    \item If $(C,M)\red (D,N)$ and $(D,N)\deadlocks$, then $(C,M)\deadlocks$.
    
    \emergencystretch=15pt
    \item If $(C,M)$ is of the form $(C,E[\boxt_T(\lift N)])$ and $(id_{Q},N\vec\ell)\deadlocks$, where $(Q,\vec\ell)=\freshlabels(N,T)$, then $(C,M) \deadlocks$.
    
    \item If $(C,M)$ is of the form $(C,E[\boxt_T(\lift N)])$ and $(id_{Q},N\vec\ell) \red^* (D,V)$, where $(Q,\vec\ell)=\freshlabels(N,T)$, and $V$ is not a label tuple, then $(C,M) \deadlocks$.
\end{enumerate}
We say that a configuration $(C,M)$ \emph{goes into deadlock} when $(C,M)\deadlocks$.
\end{df}

\noindent The last definition, the one for diverging configurations, is co-inductive. The intuitive difference between the two kinds of definition is that whereas an element belongs in an inductive set if there is a good reason for it to do so, an element belongs in a co-inductive set if there is no good reason for it not to. More practically, an inductive definition starts with the empty set and states the properties that an element must satisfy in order to get into the set, whereas a co-inductive definition starts with the universal set and states the properties that an element which is \textit{already} in the set must satisfy in order not to get kicked out.

\begin{df}[Diverging Small-step Configuration]
Let $\diverges$ be the largest unary relation over small-step configurations such that whenever $(C,M)\diverges$ either one of the following is true:
\begin{enumerate}
    \item $(C,M)\red(D,N)$ and $(D,N)\diverges$,
    
    \item $(C,M)$ is a configuration of the form $(C,E[\boxt_T(\lift N)])$ and $(id_Q,N\vec\ell)\diverges$, where $(Q,\vec\ell)=\freshlabels(N,T).$
\end{enumerate}
We say that a configuration $(C,M)$ is \emph{diverging} when $(C,M)\diverges$.
\end{df}

\noindent What this definition says is that every configuration is diverging unless it is irreducible (i.e. a normal form or a deadlocked form) or reduces to a configuration that does not diverge (directly or as part of a sub-derivation introduced by box). Naturally, we expect convergence, deadlock and divergence to be mutually exclusive. This expectation is formalized by the following proposition.

\begin{theoremEnd}{prop}\label{small-step mutex}
The relations $\converges, \deadlocks$ and $\diverges$ are mutually exclusive over small-step configurations. That is, for every small-step configuration $(C,M)$, the following are true:
\begin{enumerate}
    \item If $(C,M)\converges$, then $(C,M)\ndeadlocks$,
    \item If $(C,M)\converges$, then $(C,M)\ndiverges$,
    \item If $(C,M)\deadlocks$, then $(C,M)\ndiverges$.
\end{enumerate}
\end{theoremEnd}
\begin{proofEnd}
We prove each claim separately:
\begin{enumerate}
    \item We proceed by induction on $(C,M)\converges$:
    \begin{itemize}
        \item Case of $M \equiv V$. Since $(C,V)$ is irreducible, but $V$ is a value and (as a consequence) cannot be of the form $E[\boxt_T(\lift N)]$, there is no way for $(C,V)$ to go into deadlock, so we conclude $(C,V)\ndeadlocks$.
        
        \item Case of $(C,M) \red (D,N)$ and $(D,N)\converges$. Since $(C,M)$ is reducible, it must be that $(C,M) \red (D,N)$ ($\red$ is deterministic) and $(D,N)\deadlocks$ in order for $(C,M)$ to go into deadlock. However, by inductive hypothesis we know that $(D,N)\ndeadlocks$, so this is impossible and we conclude $(C,M)\ndeadlocks$.
    \end{itemize}
    
    \item We proceed by induction on $(C,M)\converges$:
    \begin{itemize}
        \item Case of $M \equiv V$. Since $(C,V)$ is irreducible, $V$ is a value and (as a consequence) cannot be of the form $E[\boxt_T(\lift N)]$, there is no way for $(C,V)$ to diverge, so we conclude $(C,V)\ndiverges$.
        
        \item Case of $(C,M) \red (D,N)$ and $(D,N)\converges$. Since $(C,M)$ is reducible, it must be that $(C,M) \red (D,N)$ ($\red$ is deterministic) and $(D,N)\diverges$ in order for $(C,M)$ to diverge. However, by inductive hypothesis we know that $(D,N)\ndiverges$, so this is impossible and we conclude $(C,M)\ndiverges$.
    \end{itemize}
    
    \item We proceed by induction on $(C,M)\deadlocks$:
    \begin{itemize}
        \item Case in which $(C,M)$ is irreducible, $M\not\equiv V$ and $M \not\equiv E[\boxt_T(\lift N)]$. Since $(C,M)$ is irreducible and $M$ is not of the form $E[\boxt_T(\lift N)]$, there is no way for $(C,M)$ to diverge, so we conclude $(C,M)\ndiverges$.
        
        \item Case of $(C,M) \red (D,N)$ and $(D,N)\deadlocks$. Since $(C,M)$ is reducible, it must be that $(C,M) \red (D,N)$ and $(D,N)\diverges$ in order for $(C,M)$ to diverge. However, by inductive hypothesis we know that $(D,N)\ndiverges$, so this is impossible and we conclude $(C,M)\ndiverges$.
        
        \item Case of $M \equiv E[\boxt_T(\lift N)]$ and $(id_Q,N\vec\ell)\deadlocks$. Since $(C,M)$ is irreducible, it must be that $(id_Q,N\vec\ell)\diverges$ in order for $(C,M)$ to diverge. However, by inductive hypothesis we know that $(id_Q,N\vec\ell)\ndiverges$, so this is impossible and we conclude $(C,M)\ndiverges$.

        \item Case of $M \equiv E[\boxt_T(\lift N)]$, $(id_Q,N\vec\ell) \red^* (D,V)$ and $V$ is not a label tuple, where $(Q,\vec\ell)=\freshlabels(N,T)$. Since $(C,M)$ is irreducible, it must be that $(id_Q,N\vec\ell)\diverges$ in order for $(C,M)$ to diverge. However, because $(id_Q,N\vec\ell) \red^* (D,V)$ (that is, $(id_Q,N\vec\ell)\converges$), we know by claim 2 that $(id_Q,N\vec\ell)\ndiverges$, so this is impossible and we conclude $(C,M)\ndiverges$.
    \end{itemize}
\end{enumerate}
\end{proofEnd}

\paragraph{}Just as naturally, we expect the same relations to saturate the space of small-step configurations. That is, the three relations are defined in such a way that every configuration either converges or goes into deadlock or diverges.

\begin{theoremEnd}{prop}\label{small-step totality}
Every small-step configuration $(C,M)$ either converges, goes into deadlock or diverges, that is:
$$(C,M)\converges \vee (C,M)\deadlocks \vee (C,M)\diverges.$$
\end{theoremEnd}
\begin{proofEnd}
Let $\clen$ be a function that, given a small-step configuration, returns the number of reduction steps that can be taken starting from that configuration, either at the top-level or in any of the sub-reductions introduced by a boxing operation. The $\clen$ function is defined as the least fixed point of the following equation on functions from small-step configurations to $\mathbb{N}^\infty$:
\begin{align*}
    \clen_h(C,M) &= \begin{cases}
        0 & \textnormal{if $(C,M)$ is irreducible,}\\
        \clen(D,N) + 1 & \textnormal{if $(C,M) \red (D,N)$.}
    \end{cases}\\
    \clen_v(C,M) &= \begin{cases}
        0 & \textnormal{if $M \not \equiv E[\boxt_T(\lift N)],$}\\
        \clen(id_Q,N\vec\ell) + 1 & \textnormal{if $M \equiv E[\boxt_T(\lift N)].$}
    \end{cases}\\
    \clen(C,M) &= \clen_h(C,M) + \clen_v(C,M).
\end{align*}
Where $(Q,\vec\ell)=\freshlabels(N,T)$. If $\clen(C,M)=\infty$, that means that either $(C,M) \red (D,N)$ and $\clen(D,N)=\infty$ or $M\equiv E[\boxt_T(\lift N)]$ and $\clen(id_Q,N\vec\ell)=\infty$. Because $\diverges$ is defined as the largest relation such that $(C,M)\diverges$ implies either $(C,M) \red (D,N)$ and $(D,N)\diverges$ or $M\equiv E[\boxt_T(\lift N)]$ and $(id_Q,N\vec\ell)\diverges$, we conclude that $(C,M)\diverges$. On the other hand, if $\clen(C,M)\in\mathbb{N}$, we proceed by induction on $\clen(C,M)$:
        \begin{itemize}
            \item Case $\clen(C,M) = 0$. In this case $(C,M)$ is irreducible and $M\not\equiv E[\boxt_T(\lift N)]$. If $M$ is a value, then we trivially conclude $(C,M)\converges$. Otherwise, if $M$ is not a value, we trivially conclude $(C,M)\deadlocks$.
            
            \item Case $\clen(C,M) = n+1$. We distinguish two cases:
            \begin{itemize}
            \emergencystretch=10pt
                \item If $\clen_h(C,M)=0$, then $M\equiv E[\boxt_T(\lift N)]$ and $\clen(id_Q,N\vec\ell)=n$, where $(Q,\vec\ell)=\freshlabels(N,T)$. By inductive hypothesis we know that either $(id_Q,N\vec\ell)\converges$ or $(id_Q,N\vec\ell)\deadlocks$ (we exclude $(id_Q,N\vec\ell)\diverges$ as it would entail $\clen(C,E[\boxt_T(\lift N)])=\infty\not\in\mathbb{N}$). If $(id_Q,N\vec\ell)\deadlocks$ then we immediately conclude $(C,E[\boxt_T(\lift N)])\deadlocks$. Otherwise, if $(id_Q,N\vec\ell)\converges$ we know that $(id_Q,N\vec\ell)\red^*(D,V)$. If $V$ were a label tuple $\vec{\ell'}$ we would have $(C,\boxt_T(\lift N))\red(C,(\vec\ell,D,\vec\ell'))$ and, by Theorem \ref{context reduction lemma}, $(C,E[\boxt_T(\lift N)])\red(C,E[(\vec\ell,D,\vec\ell')])$, which would contradict the hypothesis that $\clen_h(C,E[\boxt_T(\lift N)])=0$. As a result, $V$ is not a label tuple and we conclude $(E[\boxt_T(\lift N)])\deadlocks$.
                
                \item If $\clen_h(C,M)>0$, then $(C,M)\red(D,N)$ and $\clen(D,N)\leq n$. By inductive hypothesis we know that either $(D,N)\converges$ or $(D,N)\deadlocks$ (we exclude $(D,N)\diverges$ as it would entail $\clen(C,M)=\infty\not\in\mathbb{N}$). If $(D,N)\converges$ then we immediately conclude $(C,M)\converges$. Otherwise, if $(D,N)\deadlocks$ we immediately conclude $(C,M)\deadlocks$.
            \end{itemize}
        \end{itemize}
\end{proofEnd}

\subsubsection{Equivalence with the Big-Step Semantics}

In order for our small-step semantics to be of any use, we must prove that it behaves like the original big-step semantics. The following results guarantee that our small-step semantics is ultimately equivalent to the big step semantics as far as converging computations are concerned. That is, that $(C,M)\Downarrow (D,N)$ implies $(C,M)\red^*(D,N)$ and vice-versa.

\begin{theoremEnd}{lem}\label{big-step to small-step correspondence}
If $(C,M)\Downarrow (D,V)$, then $(C,M) \red^* (D,V)$.
\end{theoremEnd}
\begin{proofEnd}
We proceed by induction on the derivation of $(C,M)\Downarrow (D,V)$:
\begin{itemize}
    \item Case $M \equiv \vec\ell$ and
    $$
    \frac{\void}
    {(C,\vec\ell) \Downarrow (C,\vec\ell)}
    $$
    In this case the claim is trivially true by the reflexivity of $\red^*$.
    
    \item Case $M \equiv \lambda x.N$ and
    $$
    \frac{\void}
    {(C,\lambda x.N) \Downarrow (C,\lambda x.N)}
    $$
    In this case the claim is trivially true by the reflexivity of $\red^*$.
    
    \item Case $M \equiv NP$ and
    $$
    \frac{(C,N)\Downarrow(C_1,\lambda x.L) \quad (C_1, P) \Downarrow (C_2, W) \quad (C_2,L[W/x]) \Downarrow (D,V)}
    {(C,NP) \Downarrow (D,V)}
    $$
    In this case by inductive hypothesis we know $(C,N)\red^* (C_1,\lambda x.L)$ and $(C_1, P) \red^* (C_2, W)$, so by a finite number of applications of the \textit{ctx-app-left} and \textit{ctx-app-right} rules we get $(C,NP)\red^* (C_1,(\lambda x.L)P)$ and $(C_1,(\lambda x.L)P) \red^* (C_2,(\lambda x.L)W)$, respectively. Next, by the \textit{$\beta$-reduction} rule we get $(C_2,(\lambda x.L)W) \red (C_2, L[W/x])$. By inductive hypothesis we also know $(C_2,L[W/x]) \red^* (D,V)$, so by the transitivity of $\red^*$ we conclude $(C,NP) \red^* (D,V)$.
    
    \item Case $M \equiv \letin{\tuple{x,y}}{N}{P}$ and
    $$
    \frac{(C,N)\Downarrow(C_1,\tuple{W_1,W_2}) \quad (C_1, P[W_1/x][W_2/y]) \Downarrow (D,V)}
    {(C,\letin{\tuple{x,y}}{N}{P}) \Downarrow (D,V)}
    $$
    \emergencystretch=15pt
    In this case by inductive hypothesis we know $(C,N)\red^* (C_1,\tuple{W_1,W_2})$, so by a finite number of applications of the \textit{ctx-let} rule we get $(C,\letin{\tuple{x,y}}{N}{P})\red^* (C_1,\letin{\tuple{x,y}}{\tuple{W_1,W_2}}{P})$. Next, by the \textit{let} rule we get $(C_1,\letin{\tuple{x,y}}{\tuple{W_1,W_2}}{P}) \red (C_1, P[W_1/x][W_2/y])$. By inductive hypothesis we also know $(C_1,P[W_1/x][W_2/y]) \red^* (D,V)$, so by the transitivity of $\red^*$ we conclude $(C,\letin{\tuple{x,y}}{N}{P}) \red^* (D,V)$.
    
    \item Case $M \equiv \tuple{N,P}$ and
    $$
    \frac{(C,N)\Downarrow(C_1,W_1) \quad (C_1,P) \Downarrow (D,W_2)}
    {(C,\tuple{N,P}) \Downarrow (D,\tuple{W_1,W_2})}
    $$
    In this case by inductive hypothesis we know $(C,N)\red^* (C_1,W_1)$ and $(C_1,P)\red^*(D,W_2)$, so by a finite number of applications of the \textit{ctx-tuple-left} and \textit{ctx-tuple-right} rules we get $(C,\tuple{N,P})\red^* (C_1,\tuple{W_1,P})$ and $(C_1,\tuple{W_1,P}) \red^* (D,\tuple{W_1,W_2})$. By the transitivity of $\red^*$ we conclude $(C,\tuple{N,P}) \red^* (D,\tuple{W_1,W_2})$.
    
    \item Case $M \equiv \lift N$ and
    $$
    \frac{\void}
    {(C,\lift N) \Downarrow (C,\lift N)}
    $$
    In this case the claim is trivially true by the reflexivity of $\red^*$.
    
    \item Case $M \equiv \force N$ and
    $$
    \frac{(C,N)\Downarrow(C_1,\lift P) \quad (C_1, P) \Downarrow (D,V)}
    {(C,\force N) \Downarrow (D,V)}
    $$
    In this case by inductive hypothesis we know $(C,N)\red^* (C_1,\lift P)$, so by a finite number of applications of the \textit{ctx-force} rule we get $(C,\force N)\red^* (C_1,\force (\lift P))$. Next, by the \textit{force} rule we get $(C_1,\force(\lift P)) \red (C_1, P)$. By inductive hypothesis we also know $(C_1,P) \red^* (D,V)$, so by the transitivity of $\red^*$ we conclude $(C,\force N) \red^* (D,V)$.
    
    \item Case $M \equiv \boxt_T N$ and
    $$
    \frac{(C,N)\Downarrow(D,\lift P) \quad (Q,\vec\ell)=\freshlabels(P,T) \quad
    (id_Q,P\vec\ell) \Downarrow (D',\vec{\ell'})}
    {(C,\boxt_T N) \Downarrow (D,(\vec\ell,D',\vec{\ell'}))}
    $$
    In this case by inductive hypothesis we know $(C,N)\red^* (D,\lift P)$, so by a finite number of applications of the \textit{ctx-box} rule we get $(C,\boxt_T N)\red^* (D,\boxt_T (\lift P))$. By inductive hypothesis we also know $(id_Q,P\vec\ell)\red^*(D',\vec{\ell'})$, so by the \textit{box} rule we get $(D,\boxt_T(\lift P)) \red (D,(\vec\ell,D',\vec{\ell'}))$ and conclude $(C,\boxt_T N) \red^* (D,(\vec\ell,D',\vec{\ell'}))$ by the transitivity of $\red^*$.
    
    \item Case $M \equiv \apply(N,P)$ and
    $$
    \frac{(C,N)\Downarrow(C_1,(\vec\ell,D',\vec{\ell'})) \quad (C_1, P) \Downarrow (C_2, \vec k) \quad (D,\vec{k'}) = \append(C_2,\vec k, \vec\ell, D', \vec{\ell'})}
    {(C,\apply(N,P)) \Downarrow (D,\vec{k'})}
    $$
    \emergencystretch=20pt
    In this case by inductive hypothesis we know $(C,N)\red^* (C_1,(\vec\ell,D',\vec{\ell'}))$ and $(C_1, P) \red^* (C_2, \vec k)$, so by a finite number of applications of the \textit{ctx-apply-left} and \textit{ctx-apply-right} rules we get $(C,\apply(N,P))\red^* (C_1,\apply((\vec\ell,D',\vec{\ell'}),P))$ and $(C_1,\apply((\vec\ell,D',\vec{\ell'}),P)) \red^* (C_2,\apply((\vec\ell,D',\vec{\ell'}),\vec k))$, respectively. Next, we get $(C_2,\apply((\vec\ell,D',\vec{\ell'}),\vec k)) \red (D, \vec{k'})$ by the \textit{apply} rule and conclude $(C,\apply(N,P)) \red^* (D,\vec{k'})$ by the transitivity of $\red^*$.
    
    \item Case $M \equiv (\vec\ell,D,\vec{\ell'})$ and
    $$
    \frac{\void}
    {(C,(\vec\ell,D,\vec{\ell'})) \Downarrow (C,(\vec\ell,D,\vec{\ell'}))}
    $$
    In this case the claim is trivially true by the reflexivity of $\red^*$.
\end{itemize}
\end{proofEnd}

\begin{theoremEnd}[all end]{lem}\label{decomposition of converging computations}
The following are all true:
\begin{enumerate}
    \emergencystretch=20pt
    \item If $(C,MN) \red^+ (D,V)$, then $(C,MN)\red^*(C_1,(\lambda x.P)N) \red^* (C_2,(\lambda x.P)W) \red (C_2,P[W/x])\red^* (D,V)$.
    
    \item If $(C,\tuple{M,N}) \red^+ (D,V)$, then $V\equiv \tuple{V_1,V_2}$ and $(C,\tuple{M,N})\red^*(C_1,\tuple{V1,N}) \red^* (D,\tuple{V_1,V_2})$.
    
    \item If $(C,\letin{\tuple{x,y}}{M}{N}) \red^+ (D,V)$, then $(C,\letin{\tuple{x,y}}{M}{N})\red^*(C_1,\letin{\tuple{x,y}}{\tuple{V,W}}{N}) \red (C_1,N[V/x][W/y])\red^* (D,V)$.
    
    \item If $(C,\force M) \red^+ (D,V)$, then $(C,\force M)\red^*(C_1,\force (\lift N)) \red (C_1,N)\red^* (D,V)$.
    
    \item If $(C,\boxt_T(M)) \red^+ (D,V)$, then  $(C,\boxt_T(M))\red^*(D,\boxt_T(\lift N)) \red (D,(\vec\ell,D',\vec{\ell'}))$ and $(id_Q,N\vec\ell) \red^* (D',\vec{\ell'})$, for $(Q,\vec\ell)=\freshlabels(N,T)$.
    
    \item If $(C,\apply(M,N)) \red^+ (D,V)$, then $(C,\apply(M,N))\red^*(C_1,\apply((\vec\ell,D,\vec{\ell'}),N)) \red^* (C_2,\apply((\vec\ell,D,\vec{\ell'}),\vec k)) \red  (D,\vec{k'})$ for $(D,\vec{k'})=\append(C_2,\vec k, \vec\ell, D', \vec{\ell'})$.
\end{enumerate}
\end{theoremEnd}
\begin{proofEnd}
We prove each claim separately:
\begin{enumerate}

    \item We proceed by induction on the length of the derivation $(C,MN)\red^+(D,V)$:
    \begin{itemize}
        \item Case $(C,MN)\red(D,V)$. In this case we necessarily have $M\equiv \lambda x.x$ and $N\equiv V$ and $C=D$. We get $(C,MN)=(C,(\lambda x.x)N)=(C,(\lambda x.x)V)\red(C,V)$ by the \textit{app} rule, which trivially proves the claim by the reflexivity of $\red^*$.
        
        \item Case $(C,MN)\red^+(D,V)$ in two or more steps. We distinguish three cases
        \begin{itemize}
            \item $M\equiv \lambda x.P$ and $N\equiv W$. In this case we have $(C,(\lambda x.P)W) \red (C,P[W/x])$ by the \textit{$\beta$-reduction} rule and $(C,P[W/x]) \red^+ (D,V)$, so we immediately conclude $(C,MN)=(C,(\lambda x.P)N)=(C,(\lambda x.P)W)\red(C,P[W/x])\red^*(D,V)$.
            
            \item $M\equiv \lambda x.P$ and $N\not\equiv W$. In this case we have $(C,(\lambda x.P)N) \red (C_1,(\lambda x.P)N')$ by the \textit{ctx-app-right} rule and $(C_1,(\lambda x.P)N') \red^+ (D,V)$. By inductive hypothesis we get $(C_1,(\lambda x.P)N') \red^* (C_2,(\lambda x.P)W) \red (C_2,P[W/x]) \red^* (D,V)$ and conclude $(C,MN)=(C,(\lambda x.P)N)\red^*(C_2,(\lambda x.P)W)\red(C_2,P[W/x])\red^*(D,V)$ by the transitivity of $\red^*$.
            
            \emergencystretch=14pt
            \item $M\not\equiv \lambda x.P$ and $N\not\equiv W$. In this case we have $(C,MN) \red (C_1,M'N)$ by the \textit{ctx-app-left} rule and $(C_1,M'N) \red^+ (D,V)$. By inductive hypothesis we get $(C_1,M'N) \red^* (C_2,(\lambda x.P)N) \red^* (C_3,(\lambda x.P)W) \red (C_3,P[W/x]) \red^* (D,V)$ and conclude $(C,MN)\red^*(C_2,(\lambda x.P)N)\red^*(C_3,(\lambda x.P)W)\red(C_3,P[W/x])\red^*(D,V)$ by the transitivity of $\red^*$.
        \end{itemize}
        
    \end{itemize}
    
    \item We proceed by induction on the length of the derivation $(C,\tuple{M,N})\red^+(D,V)$:
    \begin{itemize}
        \item Case $(C,\tuple{M,N})\red(D,V)$. In this case we necessarily have $M\equiv V_1$ and $(C,N)\red(D,V_2)$. We get $(C,\tuple{M,N})=(C,\tuple{V_1,N})\red(D,\tuple{{V_1,V_2}})$ by the \textit{ctx-tuple-right} rule, which trivially proves the claim by the reflexivity of $\red^*$.
        
        \item Case $(C,\tuple{M,N})\red^+(D,V)$ in two or more steps. We distinguish two cases
        \begin{itemize}
            \emergencystretch=15pt
            \item $M\equiv V_1$. In this case we have $(C,\tuple{V_1,N}) \red (C_1,\tuple{V_1,N'})$ by the \textit{ctx-tuple-right} rule and $(C_1,\tuple{V_1,N'}) \red^+ (D,V)$. By inductive hypothesis we get $(C_1,\tuple{V_1,N'}) \red^* (D,\tuple{V_1,V_2})$ and conclude $(C,\tuple{M,N})=(C,\tuple{V_1,N})\red^*(D,\tuple{V_1,V_2)}$ by the transitivity of $\red^*$.
            
            \emergencystretch=15pt
            \item $M\not\equiv V_1$. In this case we have $(C,\tuple{M,N}) \red (C_1,\tuple{M',N})$ by the \textit{ctx-tuple-left} rule and $(C_1,\tuple{M',N}) \red^+ (D,V)$. By inductive hypothesis we get $(C_1,\tuple{M',N}) \red^* (C_2,\tuple{V_1,N}) \red^* (D,\tuple{V_1,V_2})$ and conclude $(C,\tuple{M,N})\red^*(C_2,\tuple{V_1,N})\red^*(D,\tuple{V_1,V_2)}$ by the transitivity of $\red^*$.
        \end{itemize}
        
    \end{itemize}
    
    \item We proceed by induction on the length of the derivation $(C,\letin{\tuple{x,y}}{M}{N})\red^+(D,V)$:
    \begin{itemize}
        \item Case $(C,\letin{\tuple{x,y}}{M}{N}) \red (D,V)$. In this case we necessarily have $M\equiv \tuple{W_1,W_2}$ and $N\equiv V$ and $C=D$. By the \textit{let} rule we get $(C,\letin{\tuple{x,y}}{M}{V}) =(C,\letin{\tuple{x,y}}{\tuple{W_1,W_2}}{V}) \red (C,V[W_1/x][W_2/y])=(C,V)$, which trivially proves the claim by the reflexivity of $\red^*$.
        
        \item Case $(C,\letin{\tuple{x,y}}{M}{N})\red^+(D,V)$ in two or more steps. We distinguish two cases:
        \begin{itemize}
            \emergencystretch=25pt
            \item Case $M\equiv \tuple{W_1,W_2}$. In this case we have $(C,\letin{\tuple{x,y}}{\tuple{W_1,W_2}}{N}) \red (C,N[W_1/x][W_2/y])$ by the \textit{let} rule and $(C,N[W_1/x][W_2/y])\red^+(D,V)$, so we immediately conclude $(C,\letin{\tuple{x,y}}{M}{N})=(C,\letin{\tuple{x,y}}{\tuple{W_1,W_2}M}{N}) \red (C,N[W_1/x][W_2/y])\red^* (D,V)$.
            
            \emergencystretch=25pt
            \item Case $M\not\equiv \tuple{W_1,W_2}$. In this case we have $(C,\letin{\tuple{x,y}}{M}{N}) \red (C_1, \letin{\tuple{x,y}}{M'}{N})$ by the \textit{ctx-let} rule and $(C_1, \letin{\tuple{x,y}}{M'}{N}) \red^+ (D,V)$. By inductive hypothesis we get $(C_1,\letin{\tuple{x,y}}{M'}{N}) \red^* (C_2,\letin{\tuple{x,y}}{\tuple{W_1,W_2}}{N}) \red (C_2,N[W_1/x][W_2/y]) \red^* (D,V)$ and conclude $(C,\letin{\tuple{x,y}}{M}{N}) \red^* (C_2,\letin{\tuple{x,y}}{\tuple{W_1,W_2}}{N}) \red (C_2,N[W_1/x][W_2/y]) \red^* (D,V)$ by the transitivity of $\red^*$.
        \end{itemize}
    \end{itemize}
    
    \item We proceed by induction on the length of the derivation $(C,\force M)\red^+(D,V)$:
    \begin{itemize}
        \item Case $(C,\force M) \red (D,V)$. In this case we necessarily have $M\equiv \lift V$ and $C=D$. We get $(C,\force M) = (C,\force(\lift V)) \red (C,V)$ by the \textit{force} rule, which trivially proves the claim by the reflexivity of $\red^*$.
        
        \item Case $(C,\force M)\red^+(D,V)$ in two or more steps. We distinguish two cases:
        \begin{itemize}
            \emergencystretch=10pt
            \item Case $M\equiv \lift N$. In this case we have $(C,\force(\lift N)) \red (C,N)$ by the \textit{force} rule and $(C,N)\red^+(D,V)$, so we immediately conclude $(C,\force M)=(C,\force (\lift N)) \red (C,N)\red^* (D,V)$.
            
            \emergencystretch=15pt
            \item Case $M\not\equiv \lift N$. In this case we have $(C,\force M) \red (C_1, \force M')$ by the \textit{ctx-force} rule and $(C_1, \force M') \red^+ (D,V)$. By inductive hypothesis we get $(C_1,\force M') \red^* (C_2,\force (\lift N)) \red (C_2,N) \red^* (D,V)$ and conclude $(C,\force M) \red^* (C_2,\force (\lift N)) \red (C_2,N) \red^* (D,V)$ by the transitivity of $\red^*$.
        \end{itemize}
    \end{itemize}
    
    \item We proceed by induction on the length of the derivation $(C,\boxt_T M)\red^+(D,V)$:
    \begin{itemize}
        \emergencystretch=10pt
        \item Case $(C,\boxt_T M) \red (D,V)$. In this case we necessarily have $M\equiv \lift N$ and $(id_Q,N\vec\ell)\red^* (D',\vec{\ell'})$ for $(Q,\vec\ell)=\freshlabels(N,T)$ and $C=D$. We get $(C,\boxt_T M) = (C,\boxt_T(\lift N)) \red (C,(\vec\ell,D',\vec{\ell'}))$ by the \textit{box} rule, which trivially proves the claim by the reflexivity of $\red^*$.
        
        \item Case $(C,\boxt_T M)\red^+(D,V)$ in two or more steps. We distinguish two cases:
        \begin{itemize}
            \item Case $M\equiv \lift N$. This case is impossible since it would entail $(C,\boxt_T(\lift N)) \red (C,(\vec\ell,D',\vec{\ell'}))$ in one step.
            
            \item Case $M\not\equiv \lift N$. In this case we have $(C,\boxt_T M) \red (C_1, \boxt_T M')$ by the \textit{ctx-box} rule and $(C_1, \boxt_T M') \red^+ (D,V)$. By inductive hypothesis we get $(C_1,\boxt_T M') \red^* (D,\boxt_T (\lift N)) \red (D,(\vec\ell,D',\vec{\ell'}))$ and $(id_Q,N\vec\ell) \red^* (D',\vec{\ell'})$ for $(Q,\vec\ell)=\freshlabels(N,T)$. We conclude $(C,\boxt_T M) \red^* (D,\boxt_T (\lift N)) \red (D,(\vec\ell,D',\vec{\ell'}))$ by the transitivity of $\red^*$.
        \end{itemize}
    \end{itemize}
    
    \item We proceed by induction on the length of the derivation $(C,\apply(M,N))\red^+(D,V)$:
    \begin{itemize}
        \emergencystretch=30pt
        \item Case $(C,\apply(M,N))\red(D,V)$. In this case we necessarily have $M\equiv (\vec\ell,D',\vec{\ell'})$ and $N\equiv \vec k$. As a consequence, we know that $(C,\apply(M,N))=(C,\apply((\vec\ell,D',\vec{\ell'}),N))=(C,\apply((\vec\ell,D',\vec{\ell'}),\vec k))\red(D,\vec{k'})$ by the \textit{apply} rule, where $(D,\vec{k'})=\append(C,\vec k, \vec\ell, D',\vec{\ell'})$. This trivially proves the claim by the reflexivity of $\red^*$.
        
        \item Case $(C,\apply(M,N))\red^+(D,V)$ in two or more steps. We distinguish three cases:
        \begin{itemize}
        
            \item $M\equiv (\vec\ell,D',\vec{\ell'})$ and $N\equiv \vec k$. This case is impossible since it would entail $(C,\apply(M,N))\red(D,V)$ in one step.
            
            \emergencystretch=50pt
            \item $M\equiv (\vec\ell,D',\vec{\ell'})$ and $N\not\equiv \vec k$. In this case we have $(C,\apply((\vec\ell,D',\vec{\ell'}),N)) \red (C_1,\apply((\vec\ell,D',\vec{\ell'}),N'))$ by the \textit{ctx-apply-right} rule, and we also have $(C_1,\apply((\vec\ell,D',\vec{\ell'}),N')) \red^+ (D,V)$. By inductive hypothesis we know that $(C_1,\apply((\vec\ell,D',\vec{\ell'}),N')) \red^* (C_2,\apply((\vec\ell,D',\vec{\ell'}),\vec k)) \red (D,\vec{k'})$, where $(D,\vec{k'})=\append(C_2,\vec k, \vec\ell, D',\vec{\ell'})$, and conclude $(C,\apply(M,N))=(C,\apply((\vec\ell,D',\vec{\ell'}),N))\red^*(C_2,\apply((\vec\ell,D',\vec{\ell'}),\vec k))\red(D,\vec{k'})$ by the transitivity of $\red^*$.
            
            \emergencystretch=30pt
            \item $M\not\equiv (\vec\ell,D',\vec{\ell'})$ and $N\not\equiv \vec k$. In this case we have $(C,\apply(M,N)) \red (C_1,\apply(M',N))$ by the \textit{ctx-apply-left} rule and $(C_1,M'N) \red^+ (D,V)$. By inductive hypothesis we get $(C_1,\apply(M',N)) \red^* (C_2,\apply((\vec\ell,D',\vec{\ell'}),N)) \red^* (C_3,\apply((\vec\ell,D',\vec{\ell'}),\vec k)) \red (D,\vec{k'})$, where $(D,\vec{k'})=\append(C_3,\vec k, \vec\ell, D',\vec{\ell'})$ and therefore conclude $(C,\apply(M,N))\red^*(C_2,\apply((\vec\ell,D',\vec{\ell'}),N))\red^*(C_3,\apply((\vec\ell,D',\vec{\ell'}),\vec k))\red(D,\vec{k'})$ by the transitivity of $\red^*$.
        \end{itemize}
        
    \end{itemize}
\end{enumerate}
\end{proofEnd}

\begin{theoremEnd}{lem}\label{small-step to big-step correspondence}
If $(C,M)\red^* (D,V)$, then $(C,M) \Downarrow (D,V)$.
\end{theoremEnd}
\begin{proofEnd}
We reason by induction on the size of the proof for $(C,M)\red^* (D,V)$, proceeding by cases on $M$:
\begin{itemize}
    \item Case $M\equiv x$. This case is impossible, since by the definition of small-step configuration $M$ must contain no free variables.
        
    \item Case $M\equiv \vec\ell$. In this case $\vec\ell$ is a value, so we trivially have $C=D$ and
    $$
    \frac{\void}
    {(C,\vec\ell) \Downarrow (C,\vec\ell)}
    $$
        
    \item Case $M\equiv \lambda x.N$. In this case $\lambda x.N$ is a value, so we trivially have $C=D$ and
    $$
    \frac{\void}
    {(C,\lambda x.N) \Downarrow (C,\lambda x.N)}
    $$
        
    \item Case $M\equiv NP$. In this case we have $(C,NP)\red^+(D,V)$ and by Lemma \ref{decomposition of converging computations} we know that $(C,NP)\red^*(C_1,(\lambda x.L)P) \red^* (C_2,(\lambda x.L)W) \red (C_2,L[W/x]) \red^* (D,V)$. By repeated inversion on the \textit{ctx-app-left} and \textit{ctx-app-right} rules we get $(C,N)\red^*(C_1,\lambda x.L)$ and $(C1,P) \red^* (C_2,W)$, respectively. By inductive hypothesis we therefore know $(C,N) \Downarrow (C_1,\lambda x.L)$ and $(C1,P) \Downarrow (C_2,W)$ and $(C_2,L[W/x]) \Downarrow (D,V)$, by which we conclude
        $$
        \frac{(C,N) \Downarrow (C_1,\lambda x.L) \quad (C1,P) \Downarrow (C_2,W) \quad (C_2,L[W/x]) \Downarrow (D,V)}
        {(C,NP) \Downarrow (D,V)}
        $$
        
    \item Case $M\equiv \tuple{N,P}$. In this case we distinguish two cases. If $\tuple{N,P}\equiv\tuple{V_1,V_2}$, then $C=D$ and we trivially conclude
    $$
    \frac{\void}
    {(C,\tuple{V_1,V_2})\Downarrow(C,\tuple{V_1,V_2})}
    $$
    
    Otherwise, if $(C,\tuple{N,P})\red^+(D,V)$, then by Lemma \ref{decomposition of converging computations} we know that $(C,\tuple{N,P})\red^*(C_1,\tuple{V_1,P}) \red^* (D,\tuple{V_1,V_2})$. By repeated inversion on the \textit{ctx-tuple-left} and \textit{ctx-tuple-right} rules we get $(C,N)\red^*(C_1,V_1)$ and $(C_1,P) \red^* (D,V_2)$, respectively. By inductive hypothesis we therefore know $(C,N)\Downarrow(C_1,V_1)$ and $(C_1,P) \Downarrow (D,V_2)$, by which we conclude
    $$
    \frac{(C,N)\Downarrow(C_1,V_1) \quad (C_1,P) \Downarrow (D,V_2)}
    {(C,\tuple{N,P}) \Downarrow (D,\tuple{V_1,V_2})}
    $$
        
    \item Case $M\equiv \letin{\tuple{x,y}}{N}{P}$. In this case we have $(C,\letin{\tuple{x,y}}{N}{P})\red^+(D,V)$ and by Lemma \ref{decomposition of converging computations} we know that $(C,\letin{\tuple{x,y}}{N}{P}) \red^* (C_1,\letin{\tuple{x,y}}{\tuple{W_1,W_2}}{P}) \red (C_1,P[W_1/x][W_2/y]) \red^* (D,V)$.
    By repeated inversion on the \textit{ctx-let} we get $(C,N) \red^* (C_1,\tuple{W_1,W_2})$.
    By inductive hypothesis we therefore know $(C,N) \Downarrow (C_1,\tuple{W_1,W_2})$ and $(C_1,P[W_1/x][W_2/y]) \Downarrow (D,V)$, by which we conclude
    $$
    \frac{(C,N) \Downarrow (C_1,\lift P) \quad (C_1,P[W_1/x][W_2/y]) \Downarrow (D,V)}
    {(C,\letin{\tuple{x,y}}{N}{P}) \Downarrow (D,V)}
    $$
        
    \item Case $M\equiv \lift N$. In this case $\lift N$ is a value, so we trivially have $C=D$ and
    $$
    \frac{\void}
    {(C,\lift N) \Downarrow (C,\lift N)}
    $$
        
    \item Case $M\equiv \force N$. In this case we have $(C,\force N)\red^+(D,V)$ and by Lemma \ref{decomposition of converging computations} we know that $(C,\force N) \red^* (C_1,\force (\lift P)) \red (C_1,P) \red^* (D,V)$ and by repeated inversion on the \textit{ctx-force} rule we get $(C,N) \red^* (C_1,\lift P)$.
    By inductive hypothesis we therefore know $(C,N) \Downarrow (C_1,\lift P)$ and $(C_1,P) \Downarrow (D,V)$, by which we conclude
    $$
    \frac{(C,N) \Downarrow (C_1,\lift P) \quad (C_1,P) \Downarrow (D,V)}
    {(C,\force(\lift P)) \Downarrow (D,V)}
    $$
        
    \item Case $M\equiv \boxt_T N$. In this case we have $(C,\boxt_T N)\red^+(D,V)$ and by Lemma \ref{decomposition of converging computations} we know that $(C,\boxt_T N) \red^* (D,\boxt_T (\lift P)) \red (D,(\vec\ell,D',\vec{\ell'}))$ and $(id_Q,P\vec\ell)\red^*(D',\vec{\ell'})$, for $(Q,\vec\ell)=\freshlabels(P,T)$ and by repeated inversion on the \textit{ctx-box} rule we get $(C,N) \red^* (D,\lift P)$.
    By inductive hypothesis we therefore know $(C,N) \Downarrow (D,\lift P)$ and $(id_Q,P\vec\ell) \Downarrow (D',\vec{\ell'})$, by which we conclude
    $$
    \frac{(C,N) \Downarrow (D,\lift P) \quad (Q,\vec\ell)=\freshlabels(P,T) \quad (id_Q,P\vec\ell) \Downarrow (D',\vec{\ell'})}
    {(C,\boxt_T M) \Downarrow (D,(\vec\ell,D',\vec{\ell'}))}
    $$
    
    \item Case $M\equiv \apply(N,P)$. In this case we have $(C,\apply(N,P))\red^+(D,V)$ and by Lemma \ref{decomposition of converging computations} we know that $(C,\apply(N,P))\red^*(C_1,\apply((\vec\ell,D',\vec{\ell'}),P)) \red^* (C_2,\apply((\vec\ell,D',\vec{\ell'}),\vec k)) \red (D,\vec{k'})$.
    By repeated inversion on the \textit{ctx-apply-left} and \textit{ctx-apply-right} rules we get $(C,N)\red^*(C_1,(\vec\ell,D',\vec{\ell'}))$ and $(C1,P) \red^* (C_2,\vec k)$, respectively.
    By inductive hypothesis we therefore know $(C,N) \Downarrow (C_1,(\vec\ell,D',\vec{\ell'}))$ and $(C1,P) \Downarrow (C_2,\vec k)$, by which we conclude
    $$
    \frac{(C,N) \Downarrow (C_1,(\vec\ell,D',\vec{\ell'})) \quad (C1,P) \Downarrow (C_2,\vec k) \quad (D,\vec{k'})=\append(C_2,\vec k, \vec\ell, D', \vec{\ell'})}
    {(C,\apply(N,P)) \Downarrow (D,\vec{k'})}
    $$
        
    \item Case $M\equiv (\vec\ell,D',\vec{\ell'})$. In this case $(\vec\ell,D',\vec{\ell'})$ is a value, so we trivially have $C=D$ and
    $$
    \frac{\void}
    {(C,(\vec\ell,D',\vec{\ell'})) \Downarrow (C,(\vec\ell,D',\vec{\ell'}))}
    $$
\end{itemize}
\end{proofEnd}

\begin{thm}\label{small-step big-step equivalence}
Suppose $(C,M)$ and $(D,V)$ are small-step configurations. We have that $(C,M)\Downarrow(D,V)$ if and only if $(D,M)\red^*(D,V)$.
\end{thm}
\begin{proof}
The claim follows immediately from lemmata \ref{big-step to small-step correspondence} and \ref{small-step to big-step correspondence}.
\end{proof}

Note that although some form of equivalence could also be expected between the error relation of the big-step semantics ($\Downarrow\serr$) and the deadlocking relation of the small-step semantics ($\deadlocks$), this is not the case. This is mainly due to the fact that the big-step semantics interrupts the evaluation of a term as soon as an error is encountered, whereas by the definition of deadlocking configuration we reduce a term as much as possible before declaring that it is indeed stuck. If we had, for example, a configuration $(C,MN)$ where the evaluation of $M$ leads to something that is not an abstraction and the evaluation of $N$ diverges, in the big-step semantics the configuration would evaluate to an error as soon as $M$ is evaluated, whereas in the small-step semantics it would diverge. This means that the line between deadlocking and diverging configurations is not the same in the big-step and small-step semantics. This discrepancy, however, is not relevant, as by Theorem \ref{small-step big-step equivalence} we know that there are no cases in which the same configuration evaluates correctly in one semantics and raises an error (or diverges) in the other.

\subsubsection{Safety Results}\label{safety results}

The operational semantics is not the only aspect under which circuits and terms are related. We mentioned earlier that all the occurrences of labels within a term are free, and in Section \ref{proto-quipper-m:types} we saw that labels are given a type by label contexts. In order for a configuration $(C,M)$ to be considered \textit{well-typed}, we must be able to derive a type judgement for $M$ using exclusively labels coming from the outputs of $C$.

\begin{df}[Well-typedness]
Given label contexts $Q$ and $Q'$, a type $A$ and a configuration $(C,M)$, we say that the latter is \emph{well-typed} with input labels $Q$, output labels $Q'$ and type $A$, and we write
$$Q\vdash (C,M):A;Q',$$
when there exists $Q''$ disjoint from $Q'$ such that
$$ C:Q \to Q' \cup Q'',\quad \emptyset;Q''\vdash M:A.$$
\end{df}

\noindent We now prove two fundamental safety properties of the small-step semantics. The \textit{subject reduction} result tells us that reducing a configuration in the small-step semantics does not alter its type.

\begin{theoremEnd}[all end]{lem}\label{bridging}
Suppose we have circuits $C:Q_1\to Q_1'\cup Q_1'',D:Q_2\to Q_2'\cup Q_2''$, terms $M,N$ and types $A,B$. If $Q_1\vdash(C,M):A;Q_1' \implies Q_2\vdash(D,N):B;Q_2'$, then
$$\emptyset;Q_1''\vdash M:A \implies \emptyset;Q_2''\vdash N:B.$$
\end{theoremEnd}
\begin{proofEnd}
The proof follows immediately from the definition of well-typedness.
\end{proofEnd}

\begin{theoremEnd}{thm}[Subject Reduction]\label{subject reduction}
If $Q\vdash (C,M):A;Q'$ and $(C,M)\red(C',M')$, then $Q\vdash (C',M'):A;Q'$
\end{theoremEnd}
\begin{proofEnd}
By the definition of well-typedness we know that there exists $Q''$ disjoint from $Q'$ such that $C\to Q'\cup Q''$ and $\emptyset;Q''\vdash M:A$. We proceed by induction on the derivation of $(C,M)\red(C',M')$:
\begin{itemize}
    \item Case of \textit{$\beta$-reduction}. Suppose $M\equiv(\lambda x.N)V$. We know $Q\vdash (C,(\lambda x.N)V):A;Q'$ and we must prove $Q\vdash (C,N[V/x]):A;Q'$. We know $\emptyset;Q''\vdash (\lambda x.N)V:A$, so by Lemma \ref{generation lemma} we get
    $$\emptyset;Q''_1\vdash(\lambda x.N):B\multimap A,\quad \emptyset;Q''_2\vdash V:B,$$
    for some $B$ and $Q''_1,Q''_2$ such that $Q''=Q''_1,Q''_2$. By the same lemma, we get
    $$x:B;Q''_1 \vdash N:A,$$
    By Theorem \ref{substitution theorem} we get $\emptyset;Q''_1,Q''_2\vdash N[V/x]:A$ and we conclude
    $$Q\vdash (C, N[V/x]):A;Q'.$$
    
    \item Case of \textit{let}. Suppose $M\equiv \letin{\tuple{x,y}}{\tuple{V,W}}{N}$. We know $Q\vdash(C,\letin{\tuple{x,y}}{\tuple{V,W}}{N}):A;Q'$ and we must prove $Q\vdash(C,N[V/x][W/y]):A;Q'$. We know $\emptyset;Q''\vdash \letin{\tuple{x,y}}{\tuple{V,W}}{N} : A$, so by applying Lemma \ref{generation lemma} twice consecutively we get
    $$  \emptyset;Q''_1\vdash V:B, \quad
        \emptyset;Q''_2\vdash W:C, \quad
        x:B,y:C;Q''_3\vdash V:B,
    $$
    for some $B,C$ and $Q''_1,Q''_2,Q''_3$ such that $Q''_1,Q''_2,Q''_3=Q''$. By Theorem \ref{substitution theorem} we get first $y:C;Q_3'',Q_1''\vdash N[V/x]:A$ and then $\emptyset;Q_3'',Q_1'',Q_2''\vdash N[V/x][W/y]:A$ and we conclude
    $$Q\vdash (C,N[V/x][W/y]):A;Q'.$$
    
    \item Case of \textit{force}. Suppose $M \equiv \force(\lift N)$. We know $Q\vdash (C,\force(\lift N)):A;Q'$ and we must prove $Q\vdash (C,N):A;Q'$. We know $\emptyset;Q''\vdash \force(\lift N):A$, so by applying Lemma \ref{generation lemma} twice consecutively we get that $Q''=\emptyset$ and
    $$\emptyset;\emptyset\vdash \lift N: \bang A, \quad \emptyset; \emptyset \vdash N : A,$$
    from which we conclude
    $$Q\vdash (C,N):A;Q'.$$
    
    \item Case of \textit{box}. Suppose $M \equiv \boxt_T(\lift N)$. We know $Q\vdash \boxt_T(\lift N):\circt(T,U);Q'$ and we must prove $Q\vdash (C, (\vec\ell, D, \vec{\ell'})):\circt(T,U);Q'$, where $(Q_0,\vec\ell)=\freshlabels(N,T)$ and $(id_{Q_0},N\vec\ell)\red\dots\red (D,\vec{\ell'})$. We know $\emptyset;Q''\vdash\boxt_T(\lift N):\circt(T,U)$ and, by the definition of $\freshlabels$, $\emptyset;Q_0\vdash \vec\ell : T$. By applying Lemma \ref{generation lemma} twice we get that $Q''=\emptyset$ and
    $$\emptyset;\emptyset\vdash \lift N: \bang(T\multimap U), \quad \emptyset;\emptyset\vdash N: T\multimap U,$$
    and we therefore know by the \textit{app} rule that $$\emptyset;Q_0\vdash N\vec\ell:U,$$
    which entails that the configuration $(id_{Q_0},N\vec\ell)$ is well-typed, with $Q_0\vdash (id_{Q_0},N\vec\ell):U;\emptyset$. By applying the inductive hypothesis $n$ times, we know that all the configurations in the sequence $(id_{Q_0},N\vec\ell)\red\dots\red (D,\vec{\ell'})$ are well typed, with $Q_0\vdash (C_i,M_i):U;\emptyset$ for $i=1\dots n$. This includes $(C_n,M_n)\equiv(D,\vec{\ell'})$, so we get
    $$Q_0 \vdash (D,\vec{\ell'}):U;\emptyset.$$
    This in turn entails that there exists $Q_0''$ such that $\emptyset;Q_0''\vdash \vec{\ell'}:U$ and $D$ is a circuit of the form $D:Q_0\to Q_0''$. Thus $D\in \functorfam(Q_0,Q_0'')$ and by the \textit{circ} rule we get $\emptyset;\emptyset\vdash (\vec\ell, D, \vec{\ell'}):\circt(T,U)$, from which we conclude
    $$Q\vdash (C, (\vec\ell, D, \vec{\ell'})):\circt(T,U),Q'.$$
    
    \item Case of \textit{apply}. Suppose $M\equiv\apply((\vec\ell,D,\vec{\ell'}),\vec k)$. We know $Q\vdash(C,\apply((\vec\ell,D,\vec{\ell'}),\vec k)):U;Q'$ and we must prove $Q\vdash (C',\vec{k'}):U;Q'$, where $(C',\vec{k'})=\append(C,\vec k,\vec\ell,D,\vec{\ell'})$. We know $\emptyset;Q''\vdash \apply((\vec\ell,D,\vec{\ell'}),\vec k):U$, so by Lemma \ref{generation lemma} we get
    $$\emptyset;\emptyset\vdash (\vec\ell,D,\vec{\ell'}):\circt(T,U), \quad \emptyset;Q''\vdash \vec k:T,$$
    for some $T$. We also know that $\append$ finds $(\vec k,D',\vec{k'})\cong(\vec\ell,D,\vec{\ell'})$ to apply to the outputs of $C$. By Lemma \ref{generation lemma} we get $\emptyset;Q''\vdash \vec k:T$, which we already know, as well as
    $$\emptyset;Q'''\vdash \vec{k'} :U, \quad D'\in \functorfam(Q'',Q'''),$$
    for some $Q'''$. By appending $D':Q''\to Q'''$ to $C:Q\to Q'\cup Q''$ we thus obtain $C':Q \to Q' \cup Q'''$ and because $\emptyset; Q''' \vdash \vec{k'}:U$ we conclude
    $$Q\vdash (C',\vec{k'}):U;Q'.$$
    
    \item Case of \textit{ctx-app-left}. Suppose $M\equiv NP$, where $N$ is not a value. We know $Q\vdash (C,NP):A;Q'$ and we must prove $Q\vdash (C',N'P):A;Q'$, where $(C,N)\red(C',N')$. We know $\emptyset;Q''\vdash NP:A$, so by Lemma \ref{generation lemma} we get
    $$\emptyset;Q_1''\vdash N:B\multimap A,\quad \emptyset;Q_2''\vdash P:B,$$
    for some $B$ and $Q_1'',Q_2''$ such that $Q_1'',Q_2''=Q''$. This entails $Q\vdash (C,N):B\multimap A; Q'\cup Q_2''$ and by inductive hypothesis and Lemma \ref{bridging} we get $\emptyset;Q'''\vdash N':B\multimap A$, with $C':Q\to Q' \cup Q_2'' \cup Q'''$. Finally, by the \textit{app} rule we get $\emptyset;Q''',Q_2''\vdash N'P:A$ and conclude
    $$Q\vdash (C',N'P):A;Q'.$$
    
    \item Case of \textit{ctx-app-right}. Suppose $M\equiv VN$, where $N$ is not a value. We know $Q\vdash (C,VN):A;Q'$ and we must prove $Q\vdash (C',VN'):A;Q'$, where $(C,N)\red(C',N')$. We know $\emptyset;Q''\vdash VN:A$, so by Lemma \ref{generation lemma} we get
    $$\emptyset;Q_1''\vdash V:B\multimap A,\quad \emptyset;Q_2''\vdash N:B,$$
    for some $B$ and $Q_1'',Q_2''$ such that $Q_1'',Q_2''=Q''$. This entails $Q\vdash (C,N):B; Q'\cup Q_1''$ and by inductive hypothesis and Lemma \ref{bridging} we get $\emptyset;Q'''\vdash N':B$, with $C':Q\to Q' \cup Q_1'' \cup Q'''$. Finally, by the \textit{app} rule we get $\emptyset;Q_1'',Q'''\vdash VN':A$ and conclude
    $$Q\vdash (C',VN'):A;Q'.$$
    
    \item Case of \textit{ctx-tuple-left}. Suppose $M\equiv \tuple{N,P}$, where $N$ is not a value. We know $Q\vdash (C, \tuple{N,P}):B\otimes C;Q'$ and we must prove $Q\vdash (C', \tuple{N',P}):B\otimes C;Q'$, where $(C,N)\red(C',N')$. We know $\emptyset;Q''\vdash \tuple{N,P}:B\otimes C$, so by Lemma \ref{generation lemma} we get
    $$\emptyset;Q_1''\vdash N:B, \quad \emptyset;Q_1''\vdash P:C,$$
    for some $B,C$ and $Q_1'',Q_2''$ such that $Q_1'',Q_2''=Q''$. This entails $Q\vdash(C,N):B;Q'\cup Q_2''$ and by inductive hypothesis and Lemma \ref{bridging} we get $\emptyset;Q'''\vdash N':B$ with $C':Q\to Q'\cup Q_2'' \cup Q'''$. Finally, by the \textit{tuple} rule we get $\emptyset;Q_1'',Q_2''\vdash \tuple{N',P}:B\otimes C$ and conclude
    $$Q\vdash (C',\tuple{N',P}):B\otimes C;Q'.$$
    
    \item Case of \textit{ctx-tuple-right}. Suppose $M\equiv \tuple{V,N}$, where $N$ is not a value. We know $Q\vdash (C, \tuple{V,N}):B\otimes C;Q'$ and we must prove $Q\vdash (C', \tuple{V,N'}):B\otimes C;Q'$, where $(C,N)\red(C',N')$. We know $\emptyset;Q''\vdash \tuple{V,N}:B\otimes C$, so by Lemma \ref{generation lemma} we get
    $$\emptyset;Q_1''\vdash V:B, \quad \emptyset;Q_1''\vdash N:C,$$
    for some $B,C$ and $Q_1'',Q_2''$ such that $Q_1'',Q_2''=Q''$. This entails $Q\vdash(C,N):B;Q'\cup Q_1''$ and by inductive hypothesis and Lemma \ref{bridging} we get $\emptyset;Q'''\vdash N':B$ with $C':Q\to Q'\cup Q_1'' \cup Q'''$. Finally, by the \textit{tuple} rule we get $\emptyset;Q_1'',Q_2''\vdash \tuple{V,N'}:B\otimes C$ and conclude
    $$Q\vdash (C',\tuple{V,N'}):B\otimes C;Q'.$$
    
    \item Case of \textit{ctx-let}. Suppose $M\equiv \letin{\tuple{x,y}}{N}{P}$, where $N$ is not a value. We know $Q\vdash (C, \letin{\tuple{x,y}}{N}{P}):A;Q'$ and we must prove $Q\vdash (C', \letin{\tuple{x,y}}{N'}{P}):A;Q'$, where $(C,N)\red(C',N')$. We know $\emptyset;Q''\vdash \letin{\tuple{x,y}}{N}{P}:A$, so by Lemma \ref{generation lemma} we get
    $$\emptyset;Q_1''\vdash N:B\otimes C,\quad x:B,y:C;Q_2''\vdash P:A,$$
    for some $B,C$ and $Q_1'',Q_2''$ such that $Q_1'',Q_2''=Q''$. This entails $Q\vdash(C,N):B\otimes C;Q'\cup Q_2''$ and by inductive hypothesis and Lemma \ref{bridging} we get $\emptyset;Q'''\vdash N':B\otimes C$ with $C':Q\to Q'\cup Q_2''\cup Q'''$. Finally, by the \textit{let} rule we get $\emptyset;Q_1'',Q_2''\vdash \letin{\tuple{x,y}}{N'}{P}:A$ and conclude
    $$Q\vdash (C', \letin{\tuple{x,y}}{N'}{P}):A;Q'.$$
    
    \item Case of \textit{ctx-force}. Suppose $M\equiv \force N$, where $N$ is not a value. We know $Q\vdash (C,\force N):A;Q'$ and we must prove $Q\vdash (C',\force N'):A;Q'$, where $(C,N)\red(C',N')$. We know $\emptyset;Q''\vdash\force N:A$, so by Lemma \ref{generation lemma} we get
    $$\emptyset;Q''\vdash N:\bang A,$$
    which entails $Q\vdash (C,N):\bang A;Q'$. By inductive hypothesis and Lemma \ref{bridging} we get $\emptyset;Q''':N':\bang A$, with $C':Q\to Q'\cup Q'''$. Finally, by the \textit{force} rule we get $\emptyset;Q'''\vdash\force N':A$ and conclude
    $$Q\vdash (C',\force N'):A;Q'.$$
    
    \item Case of \textit{ctx-box}. Suppose $M\equiv \boxt N$, where $N$ is not a value. We know $Q\vdash (C,\boxt_T N):\circt(T,U);Q'$ and we must prove $Q\vdash (C',\boxt_T N'):\circt(T,U);Q'$, where $(C,N)\red(C',N')$. We know $\emptyset;Q''\vdash \boxt_T N : \circt(T,U)$, so by Lemma \ref{generation lemma} we get
    $$\emptyset;Q''\vdash N: \bang(T\multimap U),$$
    which entails $Q\vdash(C,N):\bang(T\multimap U);Q'$. By inductive hypothesis and Lemma \ref{bridging} we get $\emptyset;Q'''\vdash N': \bang(T\multimap U)$ with $C':Q\to Q' \cup Q'''$. Finally, by the \textit{box} rule we get $\emptyset;Q'''\vdash \boxt_T N' : \circt(T,U)$ and conclude
    $$Q\vdash (C', \boxt_T N'):\circt(T,U):Q'.$$
    
    \emergencystretch=15pt
    \item Case of \textit{ctx-apply-left}. Suppose $M\equiv \apply(N,P)$, where $N$ is not a value. We know $Q\vdash (C,\apply(N,P)):U;Q'$ and we must prove $Q\vdash (C',\apply(N',P)):U;Q'$, where $(C,N)\red(C',N')$. We know $\emptyset;Q''\vdash \apply(N,P):U$, so by Lemma \ref{generation lemma} we get
    $$\emptyset;Q_1''\vdash N:\circt(T,U),\quad \emptyset;Q_2''\vdash P:T,$$
    for some $T$ and $Q_1'',Q_2''$ such that $Q_1'',Q_2''=Q''$. This entails $Q\vdash (C,N):\circt(T,U); Q'\cup Q_2''$ and by inductive hypothesis and Lemma \ref{bridging} we get $\emptyset;Q'''\vdash N':\circt(T,U)$, with $C':Q\to Q' \cup Q_2'' \cup Q'''$. Finally, by the \textit{app} rule we get $\emptyset;Q''',Q_2''\vdash \apply(N',P):U$ and conclude
    $$Q\vdash (C',\apply(N',P)):U;Q'.$$
    
    \item Case of \textit{ctx-apply-right}. Suppose $M\equiv \apply(V,N)$, where $N$ is not a value. We know $Q\vdash (C,\apply(V,N)):U;Q'$ and we must prove $Q\vdash (C',\apply(V,N')):U;Q'$, where $(C,N)\red(C',N')$. We know $\emptyset;Q''\vdash \apply(V,N):U$, so by Lemma \ref{generation lemma} we get
    $$\emptyset;Q_1''\vdash V:\circt(T,U),\quad \emptyset;Q_2''\vdash N:T,$$
    for some $T$ and $Q_1'',Q_2''$ such that $Q_1'',Q_2''=Q''$. This entails $Q\vdash (C,N):T; Q'\cup Q_1''$ and by inductive hypothesis and Lemma \ref{bridging} we get $\emptyset;Q'''\vdash N':T$, with $C':Q\to Q' \cup Q_1'' \cup Q'''$. Finally, by the \textit{app} rule we get $\emptyset;Q_1'',Q'''\vdash \apply(V,N'):U$ and conclude
    $$Q\vdash (C',\apply(V,N')):U;Q'.$$
\end{itemize}
\end{proofEnd}

\paragraph{} On the other hand, the \textit{progress} result tells us that if a configuration is well-typed, then its evaluation can always be carried on by taking a further step either in the main reduction or in a sub-reduction introduced by a boxing operation. In simpler words, it tells us that well-typed configurations are never \textit{deadlocked}. We already know (from Definition \ref{small-step deadlocking}) what it means for a configuration to \textit{go} into deadlock, and now we need to know what it means for a configuration to \textit{be} in a deadlock already. Intuitively, this definition corresponds to cases 1, 3 and 4 of Definition \ref{small-step deadlocking}, or more formally:

\begin{df}[Deadlock]
Let $\deadlocked$ be the smallest set of configurations such that:
\begin{enumerate}
    \item If $(C,M)$ is irreducible and $M$ is neither of the form $E[\boxt_T(\lift N)]$, nor a value, then $(C,M)\in \deadlocked$.
    
    \emergencystretch=15pt
    \item If $M$ is of the form $E[\boxt_T(\lift N)]$ and $(id_Q,N\vec\ell)\red^*(D,N)$, where $(Q,\vec\ell) = \freshlabels(N,T)$, and $(D,N)\in \deadlocked$, then $(C,M)\in \deadlocked$.
    
    \emergencystretch=15pt
    \item If $M$ is of the form $E[\boxt_T(\lift N)]$ and $(id_{Q},N\vec\ell)\red^*(D,V)$, where $(Q,\vec\ell) = \freshlabels(N,T)$, and $V$ is not a label tuple, then $(C,M)\in \deadlocked$.
\end{enumerate}
We say that a configuration $(C,M)$ is \emph{deadlocked} when $(C,M)\in \deadlocked$.
\end{df}

\noindent In other words, a configuration is deadlocked when it cannot be further reduced to a value or when it contains a boxing operation and the sub-reduction introduced by it ends up going into deadlock too. Note that this definition does not include the case in which such a sub-reduction diverges. In a sense, the following progress result asserts that well-typed configurations are safe from deadlock, but not from \textit{livelock}.

\begin{theoremEnd}[all end]{lem}\label{no label tuple entails no mtype}
Suppose $(C,V)$ is a configuration. If $V$ is not a label tuple, then there exists no $Q''$ such that $C: Q\to Q'\cup Q''$ and $\emptyset;Q'' \vdash M:T$, for any simple M-type $T$.
\end{theoremEnd}
\begin{proofEnd}
The proof is trivial and stems directly from the fact that simple M-types are just products of wire types and that labels are the only value that can be assigned a wire type.
\end{proofEnd}

\begin{theoremEnd}{thm}[Progress]
If $Q\vdash (C,M):A;Q'$ then $(C,M)\not\in \deadlocked$. That is, $(C,M)$ is not deadlocked.
\end{theoremEnd}
\begin{proofEnd}
By the definition of well-typedness we know that there exists $Q''$ disjoint from $Q'$ such that $C:Q\to Q'\cup Q''$ and $\emptyset;Q''\vdash M:A$. We therefore prove the contrapositive: if $(C,M)\in \deadlocked$, then there exists no $Q''$ such that $C:Q\to Q'\cup Q''$ and $\emptyset;Q''\vdash M :A$. We proceed by induction on the structure of the proof that $(C,M)\in \deadlocked$, distinguishing the case in which $M\not\equiv E[\boxt_T(\lift L)]$ from the case in which $M\equiv E[\boxt_T(\lift L)]$.

\begin{itemize}
\item Case $M\not\equiv E[\boxt_T(\lift L)]$. In this case, $M$ is necessarily not a value and there exist no $C',M'$ such that $(C,M)\red(C',M')$. We must prove that there exists no $Q''$ such that $C:Q\to Q'\cup Q''$ and $\emptyset;Q''\vdash M :A$ We proceed by induction on the form of $M$:
\begin{itemize}
    \item Case $M\equiv x$. This case is impossible, since by the definition of small-step configuration $M$ must contain no free variables.
    
    \item Case $M \equiv \ell$. This case is impossible since $\ell$ is a value.
    
    \item Case $M \equiv \lambda x. N$. This case is impossible since $\lambda x. N$ is a value.
    
    \item Case $M \equiv NP$. We know that $(C,NP)$ is irreducible and $NP$ is not of the form $E[\boxt_T(\lift L)]$ and we must prove that there exists no $Q''$ such that $C:Q\to Q'\cup Q''$ and $\emptyset;Q''\vdash NP:A$. This would amount to finding $Q''_1,Q''_2$ such that $Q''_1,Q''_2 = Q''$ and
    $$\emptyset;Q_1''\vdash N:B\multimap A,\quad \emptyset;Q_2''\vdash P:B,$$
    for some $B$. By corollaries \ref{subterm irreducibility} and \ref{context propagation corollary} we know that $(C,N)$ is irreducible and $N$ is not of the form $E'[\boxt_T(\lift L)]$, so we distinguish two cases:
    \begin{itemize}
        \item $N$ is a value. In this case $N\equiv \lambda x.O$ for some $x$ and $O$, by Lemma \ref{value generation}, and we consider $P$. By corollaries \ref{subterm irreducibility} and \ref{context propagation corollary} we know that $(C,P)$ is irreducible and $P$ is not of the form $E'[\boxt_T(\lift L)]$. Also, $P$ is not a value, since if that were the case then $(C,(\lambda x.O)P)$ would be reducible by the \textit{$\beta$-reduction} rule. Therefore, by inductive hypothesis we get that there exists no $Q''_2$ such that $C:Q\to Q'\cup Q''_1\cup Q''_2$ and $\emptyset;Q''_2\vdash P : B$ and conclude that there exists no $Q''$ such that $C:Q\to Q'\cup Q''$ and $\emptyset;Q''\vdash NP:A$.
        \item $N$ is not a value. In this case, by inductive hypothesis we get that there exists no $Q''_1$ such that $C:Q\to Q'\cup Q''_1\cup Q''_2$ and $\emptyset;Q''_1\vdash N : B\multimap A$ and conclude that there exists no $Q''$ such that $C:Q\to Q'\cup Q''$ and $\emptyset;Q''\vdash NP:A$.
    \end{itemize}
    
    \item Case $M\equiv \tuple{N,P}$. We know that $(C,\tuple{N,P})$ is irreducible and $\tuple{N,P}$ is neither a value, nor of the form $E[\boxt_T(\lift L)]$ and we must prove that there exists no $Q''$ such that $C:Q\to Q'\cup Q''$ and $\emptyset;Q''\vdash \tuple{N,P}:B\otimes C$. This would amount to finding $Q''_1,Q''_2$ such that $Q''_1,Q''_2 = Q''$ and
    $$\emptyset;Q_1''\vdash N:B,\quad \emptyset;Q_2''\vdash P:C,$$
    for some $B,C$. By corollaries \ref{subterm irreducibility} and \ref{context propagation corollary} we know that $(C,N)$ is irreducible and $N$ is not of the form $E'[\boxt_T(\lift L)]$, so we distinguish two cases:
    \begin{itemize}
        \item $N$ is a value. In this case we consider $P$. By corollaries \ref{subterm irreducibility} and \ref{context propagation corollary} we know that $(C,P)$ is irreducible and $P$ is not of the form $E'[\boxt_T(\lift L)]$. Also, $P$ is not a value, since if that were the case then $\tuple{N,P}$ would be a value. Therefore, by inductive hypothesis we get that there exists no $Q''_2$ such that $C:Q\to Q'\cup Q''_1\cup Q''_2$ and $\emptyset;Q''_2\vdash P : C$ and conclude that there exists no $Q''$ such that $C:Q\to Q'\cup Q''$ and $\emptyset;Q''\vdash \tuple{N,P}:B\otimes C$.
        \item $N$ is not a value. In this case, by inductive hypothesis we get that there exists no $Q''_1$ such that $C:Q\to Q'\cup Q''_1\cup Q''_2$ and $\emptyset;Q''_1\vdash N : B$ and conclude that there exists no $Q''$ such that $C:Q\to Q'\cup Q''$ and $\emptyset;Q''\vdash \tuple{N,P}:B\otimes C$.
    \end{itemize}
    
    \item Case $M\equiv \letin{\tuple{x,y}}{N}{P}$. We know that $(C,\letin{\tuple{x,y}}{N}{P})$ is irreducible and $\letin{\tuple{x,y}}{N}{P}$ is not of the form $E[\boxt_T(\lift L)]$ and we must prove that there exists no $Q''$ such that $C:Q\to Q' \cup Q''$ and $\emptyset;Q''\vdash \letin{\tuple{x,y}}{N}{P} : A$. This would amount to finding $Q_1'',Q_2''$ such that $Q_1'',Q_2''=Q''$ and
    $$  \emptyset;Q_1''\vdash N : B\otimes C,\quad
        x:B,y:C;Q_2''\vdash P: A,    
    $$
    
    for some $B,C$. By corollaries \ref{subterm irreducibility} and \ref{context propagation corollary} we know that $(C,N)$ is irreducible and $N$ is not of the form $E'[\boxt_T(\lift L)]$. Also, $N$ is not a value, since if that were the case then $(C,\letin{\tuple{x,y}}{N}{P})$ would be reducible by the \textit{let} rule. Therefore, by inductive hypothesis we get that there exists no $Q_1''$ such that $C:Q\to Q' \cup Q_1'' \cup Q_2''$ and $\emptyset;Q_1'' \vdash N:B\otimes C$ and conclude that there exists no $Q''$ such that $C: Q \red Q' \cup Q''$ and $\emptyset;Q''\vdash \letin{\tuple{x,y}}{N}{P} : A$. 
    
    \item Case $M \equiv \lift N$. This case is impossible since $\lift N$ is a value.
    
    \item Case $M \equiv \force N$. We know that $(C,\force N)$ is irreducible and $\force N$ is not of the form $E[\boxt_T(\lift L)]$ and we must prove that there exists no $Q''$ such that $C:Q\to Q'\cup Q''$ and $\emptyset;Q''\vdash \force N:A$. This would amount to finding $Q''$ such that
    $$\emptyset;Q''\vdash N:\bang A.$$
    By corollaries \ref{subterm irreducibility} and \ref{context propagation corollary} we know that $(C,N)$ is irreducible and $N$ is not of the form $E'[\boxt_T(\lift L)]$. Also, $N$ is not a value, since if that were the case then we would have $N\equiv \lift L$ for some $L$, by Lemma \ref{value generation}, and $(C,\force (\lift L))\red (C,L)$ by the \textit{force} rule. Therefore, by inductive hypothesis we get that there exists no $Q''$ such that $C:Q\to Q' \cup Q''$ and $\emptyset;Q''\vdash N:\bang A$ and conclude that there exists no $Q''$ such that $C:Q\to Q' \cup Q''$ and $\emptyset;Q''\vdash \force N: A$
    
    \item Case $M \equiv \boxt_T N$. We know that $(C,\boxt_T N)$ is irreducible and $\boxt_T N$ is not of the form $E[\boxt_T(\lift L)]$ and we must prove that there exists no $Q''$ such that $C:Q\to Q'\cup Q''$ and $\emptyset;Q''\vdash\boxt_T N:\circt(T,U)$. This would amount to finding $Q''$ such that
    $$\emptyset;Q''\vdash N:\bang(T\multimap U).$$
   By corollaries \ref{subterm irreducibility} and \ref{context propagation corollary} we know that $(C,N)$ is irreducible and $N$ is not of the form $E'[\boxt_T(\lift L)]$. Also, $N$ is not a value, since if that were the case then we would have $N\equiv\lift L$ for some $L$, by Lemma \ref{value generation}, and $M\equiv E[\boxt_T(\lift L)]$ for $E\equiv [\cdot]$, which would contradict the hypothesis. Therefore, by inductive hypothesis we get that there exists no $Q''$ such that $C:Q\to Q' \cup Q''$ and $\emptyset;Q''\vdash N:\bang (T\multimap U)$ and we conclude that there exists no $Q''$ such that $C:Q\to Q' \cup Q''$ and $\emptyset;Q''\vdash \boxt_T N: \circt(T,U)$. 
    
    \item Case $M \equiv \apply(N,P)$. We know that $(C,\apply(N,P))$ is irreducible and also that $\apply(N,P)$ is not of the form $E[\boxt_T(\lift L)]$ and we must prove that there exists no $Q''$ such that $C:Q\to Q'\cup Q''$ and $\emptyset;Q''\vdash \apply(N,P):U$. This would amount to finding $Q''_1,Q''_2$ such that $Q''_1,Q''_2 = Q''$ and
    $$\emptyset;Q_1''\vdash N:\circt(T,U),\quad \emptyset;Q_2''\vdash P:T,$$
    for some $T$. By corollaries \ref{subterm irreducibility} and \ref{context propagation corollary} we know that $(C,N)$ is irreducible and $N$ is not of the form $E'[\boxt_T(\lift L)]$, so we distinguish two cases:
    \begin{itemize}
        \item $N$ is a value. In this case $N\equiv (\vec\ell,D,\vec{\ell'})$ for some $\vec\ell,D$ and $\vec{\ell'}$, by Lemma \ref{value generation}, and we consider $P$. By corollaries \ref{subterm irreducibility} and \ref{context propagation corollary} we know that $(C,P)$ is irreducible and $P$ is not of the form $E'[\boxt_T(\lift L)]$. Also, $P$ is not a value, since if that were the case then we would have $P\equiv \vec k$ for some $\vec k$, by Lemma \ref{value generation}, and $(C,(\vec\ell,D,\vec{\ell'}),\vec k)\red(C',\vec {k'})$ for $(C',\vec{k'})=\append(C,\vec k, \vec\ell,D,\vec{\ell'})$, by the \textit{apply} rule. Therefore, by inductive hypothesis we get that there exists no $Q''_2$ such that $C:Q\to Q'\cup Q''_1\cup Q''_2$ and $\emptyset;Q''_2\vdash P : T$ and conclude that there exists no $Q''$ such that $C:Q\to Q'\cup Q''$ and $\emptyset;Q''\vdash \apply(N,P):U$.
        \item $N$ is not a value. In this case, by inductive hypothesis we get that there exists no $Q''_1$ such that $C:Q\to Q'\cup Q''_1\cup Q''_2$ and $\emptyset;Q''_1\vdash N :\circt(T,U)$ and we conclude that there exists no $Q''$ such that $C:Q\to Q'\cup Q''$ and $\emptyset;Q''\vdash \apply(N,P):U$.
    \end{itemize}
    
    \item Case $M \equiv (\vec\ell, D, \vec{\ell'})$. This case is impossible since $(\vec\ell, D, \vec{\ell'})$ is a value.
    
\end{itemize}

\item Case $M\equiv E[\boxt_T(\lift L)]$. In this case, we necessarily have $(id_{Q_0},L\vec\ell)\red^*(C_n,M_n)$, for $(Q_0,\vec\ell)=\freshlabels(L,T)$, and either $(C_n,M_n)\in \mathcal{D}$ or $M_n$ is a value other than a label tuple. We proceed by induction on the form of $E$:
\begin{itemize}
    \item Case $E\equiv [\cdot]$. In this case $M\equiv \boxt_T(\lift L)$. Suppose there exists $Q''$ such that $C:Q\to Q'\cup Q''$ and $\emptyset;Q''\vdash \boxt_T(\lift L):\circt(T,U)$. By applying Lemma \ref{generation lemma} twice we get that $Q''=\emptyset$ and 
    $$\emptyset;\emptyset\vdash \lift L :\bang(T\multimap U), \quad \emptyset;\emptyset \vdash L : T \multimap U.$$
    This entails $\emptyset;Q_0\vdash L\vec\ell:U$, where $U$ is a simple M-type, by the definition of $\freshlabels$ and the \textit{app} rule. We know that $(id_{Q_0},L\vec\ell)\red^*(C_n,M_n)$, that is, $(id_{Q_0},L\vec\ell)$ reduces to $(C_n,M_n)$ in zero or more steps, so we distinguish two possibilities:
    \begin{itemize}
        \item $(id_{Q_0},L\vec\ell) = (C_n,M_n)$. In this case we necessarily have $(id_{Q_0},L\vec\ell)\in\deadlocked$, since $L\vec\ell$ is not a value, so by the outer inductive hypothesis we know that there exists no $Q''_0$ such that $id_{Q_0}:Q_0\to Q_0'\cup Q_0''$ and $\emptyset;Q_0''\vdash L\vec\ell:B$, for any $B$, which contradicts $\emptyset;Q_0\vdash L\vec\ell:U$.
        
        \item $(id_{Q_0},L\vec\ell)\red^+(C_n,M_n)$. In this case, we know that $Q_0\vdash(id_{Q_0},L\vec\ell):U;\emptyset$ and by a finite number of consecutive applications of Theorem \ref{subject reduction} we get $Q_0\vdash(C_n,M_n):U;\emptyset$. This entails $C_n:Q_0\to Q_0 $ and $\emptyset;Q_0\vdash M_n:U$. At the same time, we know that either $(C_n,M_n)\in \deadlocked$, and by the outer inductive hypothesis there exists no $Q_0''$ such that $C_n:Q_0\to Q'_0 \cup Q''_0$ and $\emptyset;Q_0''\vdash M_n:B$, for any $B$, or $M_n$ is a value other than a label tuple, and by Lemma \ref{no label tuple entails no mtype} there exists no $Q''_0$ such that $C':Q_0\to Q_0'\cup Q_0''$ and $\emptyset;Q_0''\vdash M_n:U$, for any simple M-type $U$. Both cases contradict $\emptyset;Q_0\vdash M_n:U$.
    \end{itemize}
    
    In either case, we reach a contradiction and conclude that there exists no $Q''$ such that $C:Q\to Q'\cup Q''$ and $\emptyset;Q''\vdash \boxt_T(\lift L):\circt(T,U)$.
    
    \item Case $E\equiv E'P$. In this case $M\equiv NP$ and $N$ is of the form $E'[\boxt_T(\lift L)]$. By inductive hypothesis we get that there exists no $Q_1''$ such that $C:Q\to Q'\cup Q''_1\cup Q''_2$ and $\emptyset;Q''_1\vdash N:B\multimap A$ and conclude that there exists no $Q''$ such that $C:Q\to Q'\cup Q''$ and $\emptyset;Q''\vdash NP:A$.
    
    \item Case $E\equiv VE'$. In this case $M\equiv VP$ and $P$ is of the form $E'[\boxt_T(\lift L)]$. By inductive hypothesis we get that there exists no $Q_2''$ such that $C:Q\to Q'\cup Q''_1\cup Q''_2$ and $\emptyset;Q''_2\vdash P:B$ and conclude that there exists no $Q''$ such that $C:Q\to Q'\cup Q''$ and $\emptyset;Q''\vdash VP:A$.
    
    \item Case $E\equiv \tuple{E',P}$. In this case $M\equiv \tuple{N,P}$ and $N$ is of the form $E'[\boxt_T(\lift L)]$. By inductive hypothesis we get that there exists no $Q_1''$ such that $C:Q\to Q'\cup Q''_1\cup Q''_2$ and $\emptyset;Q''_1\vdash N:B$ and conclude that there exists no $Q''$ such that $C:Q\to Q'\cup Q''$ and $\emptyset;Q''\vdash \tuple{N,P}:B\otimes C$.
    
    \item Case $E\equiv \tuple{V,E'}$. In this case $M\equiv \tuple{V,P}$ and $P$ is of the form $E'[\boxt_T(\lift L)]$. By inductive hypothesis we get that there exists no $Q_2''$ such that $C:Q\to Q'\cup Q''_1\cup Q''_2$ and $\emptyset;Q''_2\vdash P:C$ and conclude that there exists no $Q''$ such that $C:Q\to Q'\cup Q''$ and $\emptyset;Q''\vdash \tuple{V,P}:B\otimes C$.
    
    \item Case $E\equiv \letin{\tuple{x,y}}{E'}{P}$. In this case $M\equiv \letin{\tuple{x,y}}{N}{P}$ and $N$ is of the form $E'[\boxt_T(\lift L)]$. By inductive hypothesis we get that there exists no $Q_1''$ such that $C:Q\to Q'\cup Q_1''\cup Q_2''$ and $\emptyset;Q_1''\vdash N:B\otimes C$ and conclude that there exists no $Q''$ such that $C:Q\to Q'\cup Q''$ and $\emptyset;Q''\vdash \letin{\tuple{x,y}}{N}{P}:A$.
    
    \item Case $E\equiv \force E'$. In this case $M\equiv \force N$ and $N$ is of the form $E'[\boxt_T(\lift L)]$. By inductive hypothesis we get that there exists no $Q''$ such that $C:Q\to Q'\cup Q''$ and $\emptyset;Q''\vdash N:\bang A$ and conclude that there exists no $Q''$ such that $C:Q\to Q'\cup Q''$ and $\emptyset;Q''\vdash \force N:A$.
    
    \item Case $E\equiv \boxt_T E'$. In this case $M\equiv \boxt_T N$ and $N$ is of the form $E'[\boxt_T(\lift L)]$. By inductive hypothesis we get that there exists no $Q''$ such that $C:Q\to Q'\cup Q''$ and $\emptyset;Q''\vdash N:\bang (T\multimap U)$ and conclude that there exists no $Q''$ such that $C:Q\to Q'\cup Q''$ and $\emptyset;Q''\vdash \boxt_T N:\circt(T,U)$.
    
    \item Case $E\equiv \apply(E',P)$. In this case $M\equiv \apply(N,P)$ and $N$ is of the form $E'[\boxt_T(\lift L)]$. By inductive hypothesis we get that there exists no $Q_1''$ such that $C:Q\to Q'\cup Q''_1\cup Q''_2$ and $\emptyset;Q''_1\vdash N:\circt(T,U)$ and conclude that there exists no $Q''$ such that $C:Q\to Q'\cup Q''$ and $\emptyset;Q''\vdash \apply(N,P):U$.
    
    \item Case $E\equiv \apply(V,E')$. In this case $M\equiv \apply(V,P)$ and $P$ is of the form $E'[\boxt_T(\lift L)]$. By inductive hypothesis we get that there exists no $Q_2''$ such that $C:Q\to Q'\cup Q''_1\cup Q''_2$ and $\emptyset;Q''_2\vdash P:T$ and conclude that there exists no $Q''$ such that $C:Q\to Q'\cup Q''$ and $\emptyset;Q''\vdash \apply(V,P):U$.
\end{itemize}
\end{itemize}
\end{proofEnd}

\subsubsection{Limitations of the Current Semantics}

In this section we introduced Proto-Quipper-M and we gave a small-step semantics for its evaluation. We showed that this semantics is equivalent to the original big-step semantics by Rios and Selinger and we gave the relevant subject reduction and progress results. The proposed small-step semantics constitutes a first step towards our goal, which is devising a machine semantics for Proto-Quipper-M, but it falls short on one crucial detail: it is not truly small-step. This is due to the fact that the \textit{box} reduction rule actually requires that a full evaluation $(id_Q,N\vec\ell)\red\dots\red(D,\vec{\ell'})$, of arbitrary length, take place in its premises in order to compute the individual step $(C,\boxt_T(\lift N))\red (C,(\vec\ell,D,\vec{\ell'}))$. Furthermore, unlike the other recursive rules of the semantics (the contextual ones), the \textit{box} rule recurs on the term $N\vec\ell$, which is not a sub-term of $\boxt_T(\lift N)$. As a consequence, this small-step semantics is harder-than-usual to reason about and on multiple occasions we have had to consider the \textit{box} case separately when giving definitions or proving results about the semantics (take, for example, Definition \ref{small-step deadlocking}).

\section{Towards a Machine Semantics}
\label{towards machine semantics}

In this section we first introduce an intermediate semantics that solves the problems with the boxing operator that we just mentioned in the previous section, and that consequently is truly small-step. Then, we proceed to give the actual machine semantics for Proto-Quipper-M which is the objective of this paper. We also give a number of definitions and results about the individual semantics, which will be useful in the next section, when we explore the relationship between the different semantics.

\paragraph{}In his PhD thesis \cite{proto-quipper-s}, Ross gives small-step semantics for Proto-Quipper (specifically, Proto-Quipper-S) and avoids our pitfall with the \textit{box} rule by introducing a term of the form $(\vec\ell,D,M)$ and a contextual rule that allows to reduce $(C,(\vec\ell,D,M))\red (C,(\vec\ell,D',M'))$ whenever $(D,M)\red(D',M')$. This approach introduces an implicit evaluation stack into the language, with every term of the form $(\vec\ell,D,M)$ conceptually representing an individual stack frame. This effectively avoids the problems we encountered with \textit{box}, but it also comes with its own set of complications. For instance, despite the fact that terms of the form $(\vec\ell,D,M)$ are ``intermediate forms'' which are never meant to be written by the users of the language, every result that holds for terms in general must hold for $(\vec\ell,D,M)$ too.

\paragraph{}It is mainly for this reason that we decide to take a different route. Specifically, instead of implicitly modelling a stack through the structure of the terms inside a configuration, we explicitly add a stack to the configurations themselves. At first, we only do it for the sub-reductions introduced by the \textit{box} rule, in order to get a fully small-step semantics which we call a \textit{stacked semantics}. Later, we extend this approach to all the contextual rules to obtain what we call a \textit{machine semantics}.

\subsection{Stacked Semantics}
\label{stacked semantics section}

As the name suggests, the stacked semantics operates on ``stacks of configurations''. Every time a term of the form $\boxt_T(\lift N)$ is ready to be evaluated, a new configuration $(id_Q,N\vec\ell)$ is pushed on the stack and marked with the labels that are locally available to its evaluation (in this case, $\vec\ell$). When $(id_Q,N\vec\ell)$ eventually evaluates to $(D,\vec{\ell'})$, the configuration is popped from the stack and $\boxt_T(\lift N)$ is replaced with $(\vec\ell,D,\vec{\ell'})$ in the previous stack frame.

\begin{df}[Stacked Configuration]
A \emph{stacked configuration} is given by the following grammar:
$$X,Y ::= \epsilon \mid (C,M)^{\vec\ell}.X ,$$
where $C$ is a circuit, $M$ is a term with no free variables and $\vec\ell$ is a label tuple, which can possibly be empty ($\vec\ell = \emptyset$).
\end{df}

\begin{df}[Well-formed Stacked Configuration]
A stacked configuration is said to be \emph{well-formed} when it is of the form $(C,M)^{\vec\ell}.X$ for some $C,M,\vec\ell$ and $X$ and either one of the following conditions is met:
\begin{enumerate}
    \item $\vec\ell = \emptyset$ and $X\equiv \epsilon$,
    \item $\vec\ell \neq \emptyset$ and $X$ is a well-formed stacked configuration.
\end{enumerate}
\end{df}

\noindent In this case a configuration of the form $(C,M)^\emptyset.\epsilon$ represents a situation in which no sub-reductions are being evaluated, and all the labels occurring in $M$ are global (i.e. they were not introduced by a boxing operation). From this point onward we will assume that every stacked configuration we work with is well-formed. We can define a reduction relation $\stred$ on stacked configurations, with the following rules:

$$
\frac{(C,M)\red(D,N) \quad M \not\equiv E[\boxt_T(\lift P)]}
{(C,M)^{\vec\ell}.X \stred (D,N)^{\vec\ell}.X}\textit{head}
$$
\vspace{5pt}
$$
\frac{(Q,\vec\ell) = \freshlabels(M,T)}
{(C,E[\boxt_T(\lift M)])^{\vec k}.X \stred (id_Q,M\vec\ell)^{\vec\ell}.(C,E[\boxt_T(\lift M)])^{\vec k}.X}\textit{step-in}
$$
\vspace{5pt}
$$
\frac{\void}
{(D,\vec{\ell'})^{\vec\ell}.(C,E[\boxt_T(\lift M)])^{\vec k}.X \stred (C,E[(\vec\ell,D,\vec{\ell'})])^{\vec k}.X}\textit{step-out}
$$
\vspace{5pt}

\noindent where can clearly see that if a term does not contain a sub-term of the form $\boxt_T(\lift M)$ which is ready to be evaluated, then by definition the stacked semantics behaves exactly like the small-step semantics, reducing the head configuration (the active stack frame). It is when a term of the form $\boxt_T(\lift M)$ is ready to be evaluated that we part ways with the small-step semantics and we start taking advantage of the stack structure of these new configurations. Naturally, the reduction relation $\stred$ is deterministic.

\begin{lem}\label{stacked rule mutex}
Every stacked configuration $(C,M)^{\vec\ell}.X$ can be reduced by at most one rule of the stacked semantics.
\end{lem}
\begin{proof}
If a configuration $(C,M)^{\vec\ell}.X$ can be reduced by the \textit{step-out} rule, it means that $M$ is a value. Therefore, $(C,M)^{\vec\ell}.X$ cannot be reduced by either \textit{head} (because $(C,V)$ is irreducible) or \textit{step-in} (because $V$ cannot be of the form $E[\boxt_T(\lift N)]$, for any $E,N$). At the same time, if $(C,M)^{\vec\ell}.X$ can be reduced by the \textit{step-in} rule, it means that $M\equiv E[\boxt_T(\lift N)]$ and $(C,E[\boxt_T(\lift N)])^{\vec\ell}.X$ cannot be reduced by \textit{head}. This is sufficient to conclude that at most one rule can be applied to $(C,M)^{\vec\ell}.X$.
\end{proof}

\begin{prop}[Determinism of Stacked Semantics]\label{stacked determinism}
The reduction relation $\stred$ is deterministic. That is, if $(C,M)^{\vec\ell}.X \stred (D,N)^{\vec k}.Y$, then for every stacked configuration $(D',N')^{\vec{k'}}.Y'$ such that $(C,M)^{\vec\ell}.X \stred (D',N')^{\vec{k'}}.Y'$ we have $D=D',N\equiv N',k=k'$ and $Y=Y'$.
\end{prop}
\begin{proof}
We already know by Lemma \ref{stacked rule mutex} that at most one rule can be applied to reduce any given stacked configuration. What is left to do is prove that each rule is deterministic by itself, which is straightforward: the \textit{head} rule is deterministic thanks to Proposition \ref{small-step determinism}, while the \textit{step-out} rule is deterministic because $\freshlabels$ is a function and \textit{step-out} is trivially deterministic. We therefore conclude that $\stred$ is deterministic.
\end{proof}

\subsubsection{Initiality and Reachability}
\label{stacked initiality and reachability}

When reasoning about stacked configurations we must operate a necessary distinction between ``starting'' and ``intermediate'' configurations that we did not have to make with small-step configurations. Whereas in the small-step semantics we could expect a computation to start from any configuration $(C,M)$, in the stacked semantics we are only interested in starting a computation from a configuration of the form $(C,M)^\emptyset.\epsilon$, where the stack is empty and all of the labels in $M$ are global. Intuitively, a configuration of this form corresponds precisely to the small-step configuration $(C,M)$.

\begin{df}[Initial Stacked Configuration]
A stacked configuration is said to be \emph{initial} when it is of the form $(C,M)^\emptyset.\epsilon$. The set of initial stacked configurations is denoted by $\initialstacked$.
\end{df}

\noindent We also distinguish between the stacked configurations which can be reached by a computation starting from an initial configuration and those which cannot. For example, a configuration of the form $(C,\lambda x. x)^{\vec\ell}.(D,\lambda x. x)^\emptyset.\epsilon$, although well-formed, is clearly impossible to obtain during the evaluation of an initial configuration, since new stack frames are only introduced when a term containing a boxing operator is encountered. As a result, configurations such as this one are ill-natured in their own way. We therefore give the definition of \textit{reachable stacked configuration}.

\begin{df}[Reachable Stacked Configuration]
A stacked configuration of the form $(C,M)^{\vec \ell}.X$ is said to be \emph{reachable} when either of the following is true:
\emergencystretch=10pt
\begin{enumerate}
    \item $(C,M)^{\vec \ell}.X \in \initialstacked$,
    \item There exists a stacked configuration $(D,N)^{\vec{k}}.Y$ such that $(D,N)^{\vec{k}}.Y$ is reachable and $(D,N)^{\vec{k}}.Y \stred (C,M)^{\vec \ell}.X$.
\end{enumerate}
\end{df}

\subsubsection{Convergence, Deadlock and Divergence}
\label{stacked convergence, deadlock and divergence}

Like we did with small-step configurations, we define what it means for a stacked configuration to converge, go into deadlock or diverge.

\begin{df}[Converging Stacked Configuration]
Let $\converges$ be the smallest unary relation over stacked configurations such that:
\begin{enumerate}
    \item For every circuit $C$ and value $V$, $(C,V)^\emptyset.\epsilon \converges$,
    
    \item If $(C,M)^{\vec\ell}.X \stred (D,N)^{\vec{k}}.Y$ and $(D,N)^{\vec k}.Y \converges$, then $(C,M)^{\vec\ell}.X \converges$.
\end{enumerate}
We say that a configuration $(C,M)^{\vec\ell}.X$ is \emph{converging} when $(C,M)^{\vec\ell}.X\converges$.
\end{df}

\begin{df}[Deadlocking Stacked Configuration]
Let $\deadlocks$ be the smallest unary relation over stacked configurations such that:
\begin{enumerate}
    \item If there exists no $(D,N)^{\vec{k}}.Y$ such that $(C,M)^{\vec\ell}.X \stred (D,N)^{\vec{k}}.Y$ and either $M$ is not a value or $\vec\ell \neq \emptyset,X\neq \epsilon$, then $(C,M)^{\vec\ell}.X\deadlocks$,
    
    \item If $(C,M)^{\vec\ell}.X \stred (D,N)^{\vec{k}}.Y$ and $(D,N)^{\vec k}.Y \deadlocks$, then $(C,M)^{\vec\ell}.X \deadlocks$.
\end{enumerate}
We say that a configuration $(C,M)^{\vec\ell}.X$ \emph{goes into deadlock} when $(C,M)^{\vec\ell}.X\deadlocks$.
\end{df}

\begin{df}[Diverging Stacked Configuration]
Let $\diverges$ be the largest unary relation over stacked configurations such that whenever $(C,M)^{\vec\ell}.X\diverges$ there exists $(D,N)^{\vec k}.Y$ such that $(C,M)^{\vec\ell}.X\stred(D,N)^{\vec{k}}.Y$ and $(D,N)^{\vec{k}}.Y\diverges$. We say that a configuration $(C,M)^{\vec\ell}.X$ is \emph{diverging} when $(C,M)^{\vec\ell}.X\diverges$.
\end{df}

\noindent Note that although the intuition behind the concepts of convergence, deadlock and divergence is (obviously) still the same, the respective definitions for stacked configurations are much simpler than their small-step counterparts. This is an effect of the ``inlining'' of the box sub-reductions that we operated when defining the stacked semantics, which allows us to treat the boxing rules (\textit{step-in} and \textit{step-out}) homogeneously with the rest of the rules. Naturally, the three relations are still mutually exclusive.

\begin{theoremEnd}{prop}\label{stacked mutex}
The relations $\converges, \deadlocks$ and $\diverges$ are mutually exclusive over stacked configurations. That is, for every stacked configuration $(C,M)^{\vec\ell}.X$, the following are true:
\begin{enumerate}
    \item If $(C,M)^{\vec\ell}.X\converges$, then $(C,M)^{\vec\ell}.X\ndeadlocks$,
    \item If $(C,M)^{\vec\ell}.X\converges$, then $(C,M)^{\vec\ell}.X\ndiverges$,
    \item If $(C,M)^{\vec\ell}.X\deadlocks$, then $(C,M)^{\vec\ell}.X\ndiverges$.
\end{enumerate}
\end{theoremEnd}
\begin{proofEnd}
We prove each claim separately:
\begin{enumerate}
    \item We proceed by induction on $(C,M)^{\vec\ell}.X\converges$:
    \begin{itemize}
        \item Case of $M \equiv V, \vec\ell=\emptyset$ and $X=\epsilon$. Since $(C,V)^\emptyset.\epsilon$ is irreducible, but $V$ is a value, there is no way for $(C,V)^\emptyset.\epsilon$ to go into deadlock, so we conclude $(C,V)^\emptyset.\epsilon\ndeadlocks$.
        
        \item Case of $(C,M)^{\vec\ell}.X \stred (D,N)^{\vec k}.X'$ and $(D,N)^{\vec k}.X'\converges$. Since $(C,M)^{\vec\ell}.X$ is reducible, it must be that $(C,M)^{\vec\ell}.X \stred (D,N)^{\vec k}.X'$ ($\stred$ is deterministic) and $(D,N)^{\vec k}.X'\deadlocks$ in order for $(C,M)^{\vec\ell}.X$ to go into deadlock. However, by inductive hypothesis we know that $(D,N)^{\vec k}.X'\ndeadlocks$, so this is impossible and we conclude $(C,M)^{\vec\ell}.X\ndeadlocks$.
    \end{itemize}
    
    \item We proceed by induction on $(C,M)^{\vec\ell}.X\converges$:
    \begin{itemize}
        \item Case of $M \equiv V, \vec\ell=\emptyset$ and $X=\epsilon$. Since $(C,V)^\emptyset.\epsilon$ is irreducible, there is no way for it to diverge, so we conclude $(C,V)^\emptyset.\epsilon\ndiverges$.
        
        \item Case of $(C,M)^{\vec\ell}.X \stred (D,N)^{\vec k}.X'$ and $(D,N)^{\vec k}.X'\converges$. Since $(D,N)^{\vec k}.X'\ndiverges$ by inductive hypothesis and $(C,M)^{\vec\ell}.X$ cannot reduce to any other configuration ($\stred$ is deterministic), there is no way for $(C,M)^{\vec\ell}.X$ to diverge, so we conclude $(C,M)^{\vec\ell}.X\ndiverges$.
    \end{itemize}
    
    \item We proceed by induction on $(C,M)^{\vec\ell}.X\deadlocks$:
    \begin{itemize}
        \item Case in which $(C,M)^{\vec\ell}.X$ is irreducible and either $M\not\equiv V$ or $\vec\ell \neq \emptyset$ or $X \neq \epsilon$. Since $(C,M)^{\vec\ell}.X$ is irreducible, there is no way for it to diverge, so we conclude $(C,M)^{\vec\ell}.X\ndiverges$.
        
        \item Case of $(C,M)^{\vec\ell}.X \stred (D,N)^{\vec k}.X'$ and $(D,N)^{\vec k}.X'\deadlocks$. Since $(D,N)^{\vec k}.X'\ndiverges$ by inductive hypothesis and $(C,M)^{\vec\ell}.X$ cannot reduce to any other configuration ($\stred$ is deterministic), there is no way for $(C,M)^{\vec\ell}.X$ to diverge, so we conclude $(C,M)^{\vec\ell}.X\ndiverges$.
    \end{itemize}
\end{enumerate}
\end{proofEnd}

\begin{theoremEnd}[all end]{lem}\label{stack tail convergence}
If $(C,M)^{\vec\ell}.X$ is reachable, $(C,M)^{\vec\ell}.X \converges$ and $X\neq \epsilon$, then $X\converges$.
\end{theoremEnd}
\begin{proofEnd}
$(C,M)^{\vec\ell}.X$ is a stacked configuration of length $n \geq 2$. The base case for reachable stacked configurations has length one. This implies that in order for $(C,M)^{\vec\ell}.X$ to be reachable there must be at least one lengthening reduction somewhere between some $(D,N)^\emptyset.\epsilon \in \initialstacked$ and $(C,M)^{\vec\ell}.X$. Because lengthening reductions are only introduced by \textit{step-in}, we must have $X=(D',E[\boxt_T(\lift P)])^{\vec k}.X'$, where $(Q,\vec\ell)=\freshlabels(T,P)$, and
\begin{align*}
    (D,N)^\emptyset.\epsilon &\stred ^* (D',E[\boxt_T(\lift P)])^{\vec k}.X'\\
    &\stred (id_Q,P\vec\ell)^{\vec\ell}.(D',E[\boxt_T(\lift P)])^{\vec k}.X'\\
    &\stred^* (C,M)^{\vec\ell}.(D',E[\boxt_T(\lift P)])^{\vec k}.X'.
\end{align*}
\emergencystretch=20pt
This entails $(D',E[\boxt_T(\lift P)])^{\vec k}.X' \stred^+ (C,M)^{\vec\ell}.(D',E[\boxt_T(\lift P)])^{\vec k}.X'$ and since we already know $(C,M)^{\vec\ell}.(D',E[\boxt_T(\lift P)])^{\vec k}.X' \converges$ by hypothesis, we conclude $(D',E[\boxt_T(\lift P)])^{\vec k}.X' \converges$.
\end{proofEnd}

\begin{theoremEnd}[all end]{lem}\label{adding a stack does not solve deadlock}
Let $+_{\vec k}$ be a binary function which represents the concatenation of stacked configurations, defined as follows:
\begin{align*}
    (C,M)^\emptyset.\epsilon +_{\vec k} Y &= (C,M)^{\vec k}.Y, \\
    (C,M)^{\vec\ell}.X +_{\vec k} Y &= (C,M)^{\vec\ell}.(X +_{\vec k} Y).
\end{align*}
If $(C,M)^{\vec\ell}.X \deadlocks$, then $(C,M)^{\vec\ell}.X+_{\vec k}Y \deadlocks$ for every $\vec k$ and $Y \neq \epsilon$.
\end{theoremEnd}
\begin{proofEnd}
We proceed by induction on $(C,M)^{\vec\ell}.X \deadlocks$:
\begin{itemize}
    \item Case in which $(C,M)^{\vec\ell}.X$ is irreducible and either $M\not\equiv V$ or $\vec\ell \neq \emptyset,X \neq \epsilon$. If $(C,M)^{\vec\ell}.X+_{\vec k}Y$ were reducible by either the \textit{head} or \textit{step-in} rule, then so would $(C,M)^{\vec\ell}.X$, since both \textit{head} and \textit{step-in} can be applied regardless of the local labels or the rest of the stack. This contradicts the hypothesis, so we know that $(C,M)^{\vec\ell}.X+_{\vec k}Y$ cannot be reduced by these rules. Furthermore, if $(C,M)^{\vec\ell}.X+_{\vec k}Y$ were reducible via the \textit{step-out} rule, we would have $M\equiv \vec{\ell'}$ for some $\vec{\ell'}$ and as a consequence $\vec\ell \neq \emptyset$ and $X \neq \epsilon$ (otherwise we would contradict the hypothesis that either $M\not\equiv V$ or $\vec\ell \neq \emptyset, X \neq \epsilon$). We would therefore have $X=(D,E[\boxt_T(\lift N)])^{\vec k}.X'$ for some $D,E,N,\vec k$ and $X'$. This would entail the reducibility of $(C,\vec{\ell'})^{\vec\ell}.(D,E[\boxt_T(\lift N)])^{\vec k}.X'$, which contradicts the hypothesis, so $(C,M)^{\vec\ell}.X+_{\vec k}Y$ is ultimately irreducible and because $\vec\ell\neq\emptyset, X \neq \epsilon$ and $Y \neq \epsilon$ we trivially conclude $(C,M)^{\vec\ell}.X+_{\vec k}Y \deadlocks$.
    
    \item Case of $(C,M)^{\vec\ell}.X \stred (D,N)^{\vec{\ell'}}.X'$ and $(D,N)^{\vec{\ell'}}.X' \deadlocks$. We proceed by cases on the introduction of $(C,M)^{\vec\ell}.X \stred (D,N)^{\vec{\ell'}}.X'$:
    \begin{itemize}
        \item Case of \textit{head}. If $(C,M)^{\vec\ell}.X \stred (D,N)^{\vec{\ell}}.X$, then because the \textit{head} rule can be applied regardless of the local labels and the remaining stack we also have $(C,M)^{\vec\ell}.X +_{\vec k} Y \stred (D,N)^{\vec{\ell}}.X  +_{\vec k} Y$ by the same rule. By inductive hypothesis we get that $(D,N)^{\vec{\ell}}.X  +_{\vec k} Y \deadlocks$ and conclude $(C,M)^{\vec\ell}.X +_{\vec k} Y \deadlocks$.
        
        \item Case of \textit{step-in}. If $M\equiv E[\boxt_T(\lift P)]$ for some $E,P$ and $(C,E[\boxt_T(\lift P)])^{\vec\ell}.X \stred (id_Q,P\vec{\ell'})^{\vec{\ell'}}.(C,E[\boxt_T(\lift P)])^{\vec\ell}.X$ for $(Q,\vec\ell')=\freshlabels(P,T)$, then because the \textit{step-in} rule can be applied regardless of the local labels and the remaining stack we also have $(C,E[\boxt_T(\lift P)])^{\vec\ell}.X +_{\vec k} Y \stred (id_Q,P\vec{\ell'})^{\vec{\ell'}}.(C,E[\boxt_T(\lift P)])^{\vec\ell}.X +_{\vec k} Y$ by the same rule. By inductive hypothesis we get $(id_Q,P\vec{\ell'})^{\vec{\ell'}}.(C,E[\boxt_T(\lift P)])^{\vec\ell}.X +_{\vec k} Y \deadlocks$ and conclude $(C,E[\boxt_T(\lift P)])^{\vec\ell}.X +_{\vec k} Y \deadlocks$.
        
        \item Case of \textit{step-out}. In this case we necessarily have $\vec\ell\neq\emptyset, X\neq\epsilon$. If $M\equiv \vec{k'}$ for some $\vec{k'}$, $X=(D,E[\boxt_T(\lift P)])^{\vec {\ell'}}.X'$ for some $E,P$ and $(C,\vec{k'})^{\vec\ell}.(D,E[\boxt_T(\lift P)])^{\vec {\ell'}}.X' \stred (D,E[(\vec\ell,C,\vec{k'})])^{\vec {\ell'}}.X'$, then because $+_{\vec k}$ can only alter $\vec{\ell'}$ and $X'$ and the \textit{step-out} rule does not depend on them, we also have $(C,\vec{k'})^{\vec\ell}.(D,E[\boxt_T(\lift P)])^{\vec {\ell'}}.X' +_{\vec k} Y \stred (D,E[(\vec\ell,C,\vec{k'})])^{\vec {\ell'}}.X' +_{\vec k} Y$ by the same rule. By inductive hypothesis we get that $(D,E[(\vec\ell,C,\vec{k'})])^{\vec {\ell'}}.X' +_{\vec k} Y \ndeadlocks$ and conclude $(C,\vec{k'})^{\vec\ell}.(D,E[\boxt_T(\lift P)])^{\vec {\ell'}}.X' +_{\vec k} Y \deadlocks$.
    \end{itemize}
\end{itemize}
\end{proofEnd}

\subsection{Machine Semantics}\label{machine semantics subsection}

In this section we finally introduce a machine operational semantics for Proto-Quipper-M, which is the focus of our work. We call this a \textit{machine} semantics because it is heavily inspired by the concept of \textit{abstract machine}, which is something we ought to touch upon.

\subsubsection{Abstract Machines}

In computer science, an abstract machine is simply a theoretical model of a computer. To this effect, \textit{computer} is to be intended in the most abstract way possible, that is, something that computes. In fact, it is irrelevant whether the theoretical model is actually implementable in hardware, as long as it describes a computation the way a realistic mechanical computer would carry it out (i.e. algorithmically). Usually, this is done by means of a state transition system. Because an abstract machine allows us to know not only \textit{what} a program written in a given programming language evaluates to, but also \textit{how} exactly it is evaluated, abstract machines are often employed to define the semantics of programming languages.

\paragraph{}In the context of lambda-calculi, an abstract machine specification defines two important aspects. The first one is the evaluation strategy. For instance, whereas the basic semantics for the lambda-calculus leave the door open for both call-by-value and call-by-name strategies, any individual abstract machine must commit to \textit{either} a call-by-value \textit{or} a call-by-name strategy. The second aspect is the concrete algorithm used to carry out some of the operations which are otherwise assumed to be elementary, such as the substitution of values for variables within a term, or the exploration of a term in search of a redex. A number of abstract machines already exist that formalize the evaluation of a lambda-term. Notable examples include the SECD machine \cite{secd}, which implements a call-by-value semantics and is based on multiple evaluation stacks, the Krivine machine \cite{krivine}, which implements a call-by-name semantics, and the CEK machine \cite{cek}, which also implements a call-by-value semantics, but using continuations rather than a stack. We describe briefly this last machine, since it is the one that most inspired our machine semantics for Proto-Quipper-M. The CEK machine takes its name from the shape of its states, which are triples of the form
$$
\tuple{C,E,K},
$$
where $C$ is called \textit{control} and corresponds to the term currently being evaluated, $E$ is the \textit{environment}, that is, an associative array from variable names to values, and $K$ is the \textit{continuation} and represents the next action to perform once $C$ has been fully evaluated. For the sake of this presentation, we use the terms of the basic untyped lambda-calculus, that is:
$$
M,N ::= x \mid \lambda x. M \mid MN.
$$
The only values of this language are closures, that is, pairs $\{\lambda x.M, E\}$ of abstractions together with their definition environment. Therefore, the environment can be seen as nothing more than a list of bindings of the form
$$
x \mapsto \{\lambda y.M,E\}.
$$
When the control is a single variable name, we look that variable up in the environment to obtain the corresponding closure, whose abstraction becomes the new control and whose environment becomes the new environment. This is formalized by the following rule:
$$
\frac{\{\lambda x.M, E'\} = \lookup(x,E)}
{\tuple{x,E,K} \to \tuple{\lambda x. M, E', K}}\textit{var}
$$
where $\lookup(x,E)$ finds the first occurrence of $x$ in $E$ and returns the corresponding closure. When the control is an application $MN$, we start by evaluating $M$ to an abstraction. While we do so, we must remember that after we are done we must proceed to evaluate $N$. This is where continuations come into play. A continuation of the form $\contfarg(N,E,K)$ represents a reminder that after we are done evaluating the current control (whatever it might be), we should start evaluating $N$ in the environment $E$, and then proceed in a similar fashion with continuation $K$. The rule for evaluating applications is thus the following:
$$
\frac{\void}
{\tuple{MN,E,K} \to \tuple{M,E,\contfarg(N,E,K)}}\textit{split}
$$
Once we are done evaluating $M$ to a term of the form $\lambda x.P$, we can start evaluating $N$. Here we find ourselves in a symmetric situation: as we evaluate $N$ we must remember that after we are done we must apply $\lambda x.P$ to the result. Since functions are represented as closures, this reminder is represented by a continuation of the form $\contfapp(\lambda x.P,E,K)$, where $E$ is the environment in which $\lambda x.P$ was defined. We therefore introduce a third rule to our CEK machine:
$$
\frac{\void}
{\tuple{\lambda x.P,E,\contfarg(N,E',K)} \to \tuple{N,E',\contfapp(\lambda x.P, E, K)}}\textit{shift}
$$
Lastly, once $N$ has also been evaluated to an abstraction $\lambda y.L$, we can apply $\lambda x.P$ to $\lambda y.L$. Concretely, this means that we start evaluating $P$ under the environment $E$, with the additional binding of $x$ to $\lambda y.L$. With the following rule, the CEK machine is complete:
$$
\frac{\void}
{\tuple{\lambda y.L,E',\contfapp(\lambda x.P,E,K)} \to \tuple{P,(x\mapsto\{\lambda y. L, E'\})::E,K}}\textit{join}
$$
where $::$ denotes the concatenation of environments. To see more clearly how these rules interact with each other in order to reduce a term, consider the evaluation of $(\lambda x. xx)(\lambda y. y)$ to $\lambda y.y$.

\begin{align*}
    \tuple{(\lambda x.xx)(\lambda y. y),[],\mathit{Done}}
    &\to \tuple{\lambda x .xx, [], \contfarg(\lambda y.y,[],\mathit{Done})} &&\textit{split}\\
    &\to \tuple{\lambda y.y, [], \contfapp(\lambda x.xx,[],\mathit{Done})}  &&\textit{shift}\\
    &\to \tuple{xx,[x\mapsto\{\lambda y.y,[]\}],\mathit{Done}}  &&\textit{join}\\
    &\to \tuple{x,[x\mapsto\{\lambda y.y,[]\}],\contfarg(x,[x\mapsto\{\lambda y.y,[]\}],\mathit{Done})}  &&\textit{split}\\
    &\to \tuple{\lambda y.y,[],\contfarg(x,[x\mapsto\{\lambda y.y,[]\}],\mathit{Done})}  &&\textit{var}\\
    &\to \tuple{x,[x\mapsto\{\lambda y.y,[]\}],\contfapp(\lambda y.y,[],\mathit{Done})}  &&\textit{shift}\\
    &\to \tuple{\lambda y.y,[],\contfapp(\lambda y.y,[],\mathit{Done})}  &&\textit{var}\\
    &\to \tuple{y,[y\mapsto\{\lambda y.y,[]\}],\mathit{Done}}  &&\textit{join}\\
    &\to \tuple{\lambda y.y, [], \mathit{Done}}. &&\textit{var}
\end{align*}

\subsubsection{An Abstract Machine for Proto-Quipper-M}
\label{an abstract machine for proto-quipper-m}

As we anticipated earlier, our machine semantics for Proto-Quipper-M is largely influenced by the CEK machine. In particular, we retain the use of continuations as a means to schedule the various phases of the evaluation of a term, although we organize them in a stack rather than one within the other. Note that this change is purely syntactic, since in practice the continuations of the CEK machine already recursively define a stack. Another difference between the CEK machine and our semantics is that for the sake of simplicity we keep relying on the substitution function $M[N/x]$ to evaluate function applications instead of employing environments to the same effect. This choice does not fundamentally detract from the the results that we prove in this paper, which are expected to hold for any sensible explicit implementation of substitutions. We just leave such an implementation as future work. In conclusion, our \textit{machine configurations} are composed of a circuit, a term which builds the circuit, and a stack. 

\begin{df}[Machine Configuration]
A \emph{machine configuration} is a triple
$$
(C,M,S),
$$
where $C$ is a circuit, $M$ is a term with no free variables and $S$ is a \emph{stack}, which is defined by the following grammar:
\begin{align*}
    \textnormal{Stack elements}\quad &  H & &::= &&\,\contfarg(M)\mid\contfapp(V)\\
    &&&&&\mid\contalabel(M)\mid\contacirc(V)\\
    &&&&&\mid\conttright(M)\mid\conttleft(V)\\
    &&&&&\mid\contbox(Q,\vec\ell)\mid\contsub(C,M,\vec\ell,T)\\
    &&&&&\mid \contlet(x,y,M) \mid\contforce,\\
\textnormal{Stacks}\quad & S,R & &::= &&\,\epsilon \mid H.S.
\end{align*}
\end{df}

It is worth noting that although stacked configuration and machine configurations may appear very different at first glance, they are based on the common intuition that a computation is easily modelled by a stack. The only difference between the two is the extent to which we apply this intuition: whereas with stacked configurations we only push a frame onto the stack when evaluating an entirely new configuration as part of a boxing operation (while still reducing everything else ``in place''), with machine configurations we push a frame onto the stack every time we encounter a composite term. Let $\amred$ be a reduction relation for machine configurations. We give the following rules:
$$
\frac{\void}
{(C,MN,S) \amred (C,M,\contfarg(N).S)
}\textit{app-split}
$$
\vspace{5pt}
$$
\frac{\void}
{(C,V,\contfarg(N).S) \amred (C,N,\contfapp(V).S)
}\textit{app-shift}
$$
\vspace{5pt}
$$
\frac{\void}
{(C,V,\contfapp(\lambda x.M).S) \amred (C,M[V/x],S)
}\textit{app-join}
$$
\vspace{5pt}
$$
\frac{\void}
{(C,\apply(M,N),S) \amred (C,M,\contalabel(N).S)
}\textit{apply-split}
$$
\vspace{5pt}
$$
\frac{\void}
{(C,V,\contalabel(N).S) \amred (C,N,\contacirc(V).S)
}\textit{apply-shift}
$$
\vspace{5pt}
$$
\frac{(C',\vec{k'})=\append(C,\vec k,\vec\ell,D,\vec{\ell'})}
{(C,\vec k,\contacirc(\vec\ell,D,\vec{\ell'}).S) \amred (C',\vec{k'},S)
}\textit{apply-join}
$$
\vspace{5pt}
$$
\frac{\tuple{M,N} \textnormal{ is not a value}}
{(C,\tuple{M,N},S) \amred (C,M,\conttright(N).S)
}\textit{tuple-split}
$$
\vspace{5pt}
$$
\frac{\void}
{(C,V,\conttright(N).S) \amred (C,N,\conttleft(V).S)
}\textit{tuple-shift}
$$
\vspace{5pt}
$$
\frac{\void}
{(C,W,\conttleft(V).S) \amred (C,\tuple{V,W},S)
}\textit{tuple-join}
$$
\vspace{5pt}
$$
\frac{(Q,\vec\ell) = \freshlabels(M,T)}
{(C,\boxt_T M,S) \amred (C,M, \contbox(Q,\vec\ell).S)
}\textit{box-open}
$$
\vspace{5pt}
$$
\frac{\emptyset;Q\vdash \vec\ell : T}
{(C,\lift M, \contbox(Q,\vec\ell).S) \amred (id_Q,M\vec\ell,\contsub(C,M,\vec\ell,T).S)
}\textit{box-sub}
$$
\vspace{5pt}
$$
\frac{\void}
{(D,\vec{\ell'},\contsub(C,M,\vec\ell,T).S) \amred (C,(\vec\ell,D,\vec{\ell'}),S)
}\textit{box-close}
$$
\vspace{5pt}
$$
\frac{\void}
{(C,\letin{\tuple{x,y}}{M}{N},S) \amred (C,M,\contlet(x,y,N).S)
}\textit{let-split}
$$
\vspace{5pt}
$$
\frac{\void}
{(C,\tuple{V,W},\contlet(x,y,M).S) \amred (C,M[V/x][W/y],S)
}\textit{let-join}
$$
\vspace{5pt}
$$
\frac{\void}
{(C,\force M, S) \amred (C,M,\contforce.S)
}\textit{force-open}
\qquad
\frac{\void}
{(C,\lift M,\contforce.S) \amred (C,M,S)
}\textit{force-close}
$$
\vspace{5pt}

\noindent In light of the previous exposition of the rules for the CEK machine, the rules for the machine semantics should be self-explanatory. Generally, every binary term constructor (such as the application or the tuple) has an associated \textit{split} rule, which defines how a term is split into smaller sub-terms and which is evaluated first, a \textit{shift rule}, which defines how and when we switch to evaluating the second term, and a \textit{join} rule, which defines the way the results of the two sub-terms are put back together. To keep track of what to do next, these rules employ two kinds of continuations, one that keeps track of the right sub-term while the left one is being evaluated (like $\contfarg$) and one that does the opposite (like $\contfapp$).
The case of unary constructors (such as $\boxt_T$ or $\force$) is similar, although the kind of information that is stored on the stack in this case is more varied. For example, the $\contsub$ continuation, which roughly corresponds to a stack frame of the stacked semantics, has to store the entire circuit and term whose evaluation was temporarily interrupted by the boxing operation, as well as the new local labels introduced by it and their associated type (the latter for bookkeeping reasons which will be clear in the next section). On the other hand, the $\contforce$ continuation does not need to store any additional information, since all the $\force$ operator does is ``undo'' the lifting of a term.

\paragraph{} Notice how the definition of the machine semantics effectively allows us to relinquish the notion of evaluation context. Whereas in the small-step and stacked semantics we sometimes had to reason about \textit{where} a reduction occurred within a term (see, as an example, the \textit{step-in} rule of the stacked semantics), in the machine semantics we always reduce a redex whose components are immediately available in the term component and on top of the stack. If a redex is not immediately available, the term being evaluated is broken down into smaller pieces, and this decomposition operation is an integral part of the semantics. As a result, the ``descent'' into a term in search of a redex is no longer implicit in the derivation of an individual step of the reduction relation, but rather it is explicit in the reduction sequence itself. Naturally, the machine semantics is deterministic, like the small-step and the stacked semantics.

\begin{theoremEnd}{lem}\label{machine rule mutex}
Every machine configuration $(C,M,S)$ can be reduced by at most one rule of the machine semantics.
\end{theoremEnd}
\begin{proofEnd}
We first partition the rules into two sets: one containing the rules that require $M$ to be a value (\textit{app-shift, app-join, apply-shift, apply-join, tuple-shift, tuple-join, box-sub, box-close, let-join} and \textit{force-close}) and the other containing the rules that require $M$ not to be a value (\textit{app-split, apply-split, tuple-split, box-open, let-split, force-open}). Naturally, the applicability of a rule from the first set to any given configuration excludes the possibility of applying any rule from the second set to the same configuration, and vice-versa. Consider now the first set. Each of the rules contained in this set requires a different stack head in order to be applied, so at most one rule from the first set can be applied to any given configuration. Consider now the second set. Each of the rules contained in this set requires a different shape of $M$ in order to be applied, so at most one rule from the second set can be applied to any given configuration. We conclude that every machine configuration $(C,M,S)$ can be reduced by at most one rule of the machine semantics.
\end{proofEnd}

\begin{prop}[Determinism of Machine Semantics]\label{machine determinism}
\emergencystretch=15pt
The reduction relation $\amred$ is deterministic. That is, if $(C,M,S)\amred (D,N,R)$, then for every stacked configuration $(D',N',R')$ such that $(C,M,S)\amred (D',N',R')$ we have $D=D',N\equiv N'$ and $R=R'$.
\end{prop}
\begin{proof}
We already known by Lemma \ref{machine rule mutex} that at most one rule can be applied to reduce any given machine configuration. What is left to do is prove that each rule is deterministic by itself. The proof is trivial, since substitution, $\append$, $\freshlabels$ and the typing judgement are all functions.
\end{proof}

\subsubsection{Initiality and Reachability}

The same reasoning about what kind of configuration we can start a computation from that we made for the stacked semantics can (and must) be made for the machine semantics. The following definitions are not fundamentally different from the corresponding definitions that we gave for stacked configurations.

\begin{df}[Initial Machine Configuration]
A machine configuration is said to be \emph{initial} when it is of the form $(C,M,\epsilon)$. The set of initial machine configurations is denoted by $\initialmachine$.
\end{df}

\begin{df}[Reachable Machine Configuration]
\emergencystretch=10pt
A machine configuration of the form $(C,M,S)$ is said to be \emph{reachable} when either of the following is true:
\begin{enumerate}
    \item $(C,M,S) \in \initialmachine$,
    \item There exists a machine configuration $(D,N,R)$ such that $(D,N,R)$ is reachable and $(D,N,R) \amred (C,M,S)$.
\end{enumerate}
\end{df}

\subsubsection{Convergence, Deadlock and Divergence}

We also give the usual definitions of convergence, deadlock and divergence. Notice how similar the following definitions are to the corresponding definitions given in Section \ref{stacked convergence, deadlock and divergence} for stacked configurations, and how different they are from the corresponding definitions given in Section \ref{small-step convergence, deadlock and divergence} for small-step configurations. This is further proof of what we briefly mentioned earlier, that is, that the stacked semantics and the machine semantics are built on the same intuition.

\begin{df}[Converging Machine Configuration]
Let $\converges$ be the smallest unary relation over machine configurations such that:
\begin{enumerate}
    \item For every circuit $C$ and value $V$, $(C,V,\epsilon)\converges$,
    
    \item If $(C,M,S)\amred(D,N,R)$ and $(D,N,R)\converges$, then $(C,M,S)\converges$.
\end{enumerate}
We say that a configuration $(C,M,S)$ is \emph{converging} when $(C,M,S)\converges$.
\end{df}

\begin{df}[Deadlocking Machine Configuration]
Let $\deadlocks$ be the smallest unary relation over machine configurations such that:
\begin{enumerate}
    \item If there exists no $(D,N,R)$ such that $(C,M,S) \amred (D,N,R)$ and $S\neq \epsilon$, then $(C,M,S)\deadlocks$,
    
    \item If $(C,M,S)\amred(D,N,R)$ and $(D,N,R)\deadlocks$, then $(C,M,S)\deadlocks$.
\end{enumerate}
We say that a configuration $(C,M,S)$ \emph{goes into deadlock} when $(C,M,S)\deadlocks$.
\end{df}

\begin{df}[Diverging Machine Configuration]
Let $\diverges$ be the largest unary relation over machine configurations such that whenever $(C,M,S)\diverges$ there exists $(D,N,R)$ such that $(C,M,S)\amred (D,N,R)$ and $(D,N,R)\diverges$. We say that a configuration $(C,M,S)$ is \emph{diverging} when $(C,M,S)\diverges$.
\end{df}

As was the case with the previous semantics, the convergence, deadlock and divergence relations are mutually exclusive and total on machine configurations.

\begin{theoremEnd}{prop}\label{machine mutex}
The relations $\converges, \deadlocks$ and $\diverges$ are mutually exclusive over machine configurations. That is, for every machine configuration $(C,M,S)$, the following are true:
\begin{enumerate}
    \item If $(C,M,S)\converges$, then $(C,M,S)\ndeadlocks$,
    \item If $(C,M,S)\converges$, then $(C,M,S)\ndiverges$,
    \item If $(C,M,S)\deadlocks$, then $(C,M,S)\ndiverges$.
\end{enumerate}
\end{theoremEnd}
\begin{proofEnd}
We prove each claim separately:
\begin{enumerate}
    \item We proceed by induction on $(C,M,S)\converges$:
    \begin{itemize}
        \item Case of $M \equiv V$ and $S=\epsilon$. Since $(C,V,\epsilon)$ is irreducible, but the stack is empty, there is no way for $(C,V,\epsilon)$ to go into deadlock, so we conclude $(C,V,\epsilon)\ndeadlocks$.
        
        \item Case of $(C,M,S) \amred (D,N,S')$ and $(D,N,S')\converges$. Since $(C,M,S)$ is reducible, it must be that $(C,M,S) \amred (D,N,S')$ ($\amred$ is deterministic) and $(D,N,S')\deadlocks$ in order for $(C,M,S)$ to go into deadlock. However, by inductive hypothesis we know that $(D,N,S')\ndeadlocks$, so this is impossible and we conclude $(C,M,S)\ndeadlocks$.
    \end{itemize}
    
    \item We proceed by induction on $(C,M,S)\converges$:
    \begin{itemize}
        \item Case of $M \equiv V$ and $S=\epsilon$. Since $(D,V,\epsilon)$ is irreducible, there is no way for it to diverge, so we conclude $(D,V,\epsilon)\ndiverges$.
        
        \item Case of $(C,M,S) \amred (D,N,S')$ and $(D,N,S')\converges$. Since $(D,N,S')\ndiverges$ by inductive hypothesis and $(C,M,S)$ cannot reduce to any other configuration ($\amred$ is deterministic), there is no way for $(C,M,S)$ to diverge, so we conclude $(C,M,S)\ndiverges$.
    \end{itemize}
    
    \item We proceed by induction on $(C,M,S)\deadlocks$:
    \begin{itemize}
        \item Case in which $(C,M,S)$ irreducible and $S \neq \epsilon$. Since $(C,M,S)$ is irreducible, there is no way for it to diverge, so we conclude $(C,M,S)\ndiverges$.
        
        \item Case of $(C,M,S) \amred (D,N,S')$ and $(D,N,S')\deadlocks$. Since $(D,N,S')\ndiverges$ by inductive hypothesis and $(C,M,S)$ cannot reduce to any other configuration ($\amred$ is deterministic), there is no way for $(C,M,S)$ to diverge, so we conclude $(C,M,S)\ndiverges$.
    \end{itemize}
\end{enumerate}
\end{proofEnd}

\begin{theoremEnd}[all end]{lem}\label{machine irreducible implies M value}
If $(C,M,S)$ is irreducible, then $M$ is a value.
\end{theoremEnd}
\begin{proofEnd}
We prove the contrapositive. That is, if $M$ is not a value then $(C,M,S)$ is reducible. We proceed by cases on the form of $M$:
\begin{itemize}
    \item Case $M\equiv x$. This case is impossible, since by the definition of machine configuration $M$ must contain no free variables.
    \item Case $M\equiv \vec\ell$. In this case $M$ is a value and the claim is vacuously true.
    \item Case $M\equiv \lambda x.N$. In this case $M$ is a value and the claim is vacuously true.
    \item Case $M\equiv NP$. In this case $(C,NP,S)$ can be reduced by the \textit{app-split} rule.
    \item Case $M\equiv \tuple{N,P}$. If $N$ and $P$ are both values then $\tuple{N,P}$ is a value too and the claim is vacuously true. Otherwise, $(C,\tuple{N,P},S)$ can be reduced by the \textit{tuple-split} rule.
    \item Case $M\equiv \letin{\tuple{x,y}}{N}{P}$. In this case $(C,\letin{\tuple{x,y}}{N}{P},S)$ can be reduced by the \textit{let-split} rule.
    \item Case $M\equiv \lift N$. In this case $M$ is a value and the claim is vacuously true.
    \item Case $M\equiv \force N$. In this case $(C,\force N,S)$ can be reduced by the \textit{force-open} rule.
    \item Case $M\equiv \boxt_T N$. In this case $(C,\boxt_T N,S)$ can be reduced by the \textit{box-open} rule.
    \item Case $M\equiv \apply(N,P)$.  In this case $(C,\apply(N,P),S)$ can be reduced by the \textit{apply-split} rule.
    \item Case $M\equiv (\vec\ell,D,\vec{\ell'})$. In this case $M$ is a value and the claim is vacuously true.
\end{itemize}
\end{proofEnd}

\begin{theoremEnd}{prop}\label{machine totality}
Every machine configuration $(C,M,S)$ either converges, goes into deadlock or diverges, that is:
$$(C,M,S)\converges \vee (C,M,S)\deadlocks \vee (C,M,S)\diverges.$$
\end{theoremEnd}
\begin{proofEnd}
Let $\clen$ be a function that, given a machine configuration, returns the number of reduction steps that can be taken starting from that configuration. The $\clen$ function is defined as the least fixed point of the following equation on functions from machine configurations to $\mathbb{N}^\infty$:
$$
\clen(C,M,S) = \begin{cases}
\clen(D,N,R)+1 & \textnormal{if } (C,M,S)\amred(D,N,R), \\
0 & \textnormal{otherwise}.
\end{cases}
$$
\emergencystretch=15pt
If $\clen(C,M,S)=\infty$, that means that $(C,M,S) \amred (D,N,R)$ and $\clen(D,N,R)=\infty$. Because $\diverges$ is defined as the largest relation such that $(C,M,S)\diverges$ implies $(C,M,S) \amred (D,N,R)$ and $(D,N,R)\diverges$, we conclude that $(C,M,S)\diverges$. On the other hand, if $\clen(C,M,S)\in\mathbb{N}$, we proceed by induction on $\clen(C,M,S)$:
        \begin{itemize}
            \item Case $\clen(C,M,S) = 0$. In this case $(C,M,S)$ is irreducible and by Lemma \ref{machine irreducible implies M value} we know $M\equiv V$. If $S=\epsilon$, then we trivially conclude $(C,V,\epsilon)\converges$. Otherwise, if $S\neq\epsilon$, we trivially conclude $(C,V,S)\deadlocks$.
            
            \item Case $\clen(C,M,S) = n+1$. In this case we know that $(C,M,S)\amred(D,N,R)$ and $\clen(D,N,R)=n$. By inductive hypothesis, we know that either $(D,N,R)\converges$ or $(D,N,R)\deadlocks$ or $(D,N,R)\diverges$. We exclude the last option, because if $(D,N,R)\diverges$, then $\clen(D,N,R)$ (and as a consequence $\clen(C,M,S)$) would be undefined, contradicting the hypothesis. If $(D,N,R)\converges$, we have $(C,M,S)\amred (D,N,R)$ and we conclude $(C,M,S)\converges$, whereas if $(D,N,R)\deadlocks$ we conclude $(C,M,S)\deadlocks$ by the same reasoning.
        \end{itemize}
\end{proofEnd}

\section{Correspondence Results}

We are at a point where we have presented three different semantics for the evaluation of Proto-Quipper-M programs and we have analyzed their individual properties. In this section, we start looking at the various semantics in relationship with each other, in order to eventually prove that the machine semantics that we arrived to in Section \ref{towards machine semantics} is ultimately equivalent to the small-step semantics given in Section \ref{proto-quipper-m} and, as a consequence, to the original big-step operational semantics given by Rios and Selinger. To this effect, we first focus on the relationship between the stacked semantics and the other two semantics. Then, in light of the respective results, we show that the computations carried out in the small-step and machine semantics are equivalent by proving that they are simulated by the same computation in the stacked semantics.

\subsection{From Small-step to Stacked Semantics}
\label{section from small-step to stacked}
\label{section from small-step to stacked configurations}
\label{from small-step to stacked reductions}
\label{small-step to stacked convergence, deadlock and divergence}

Since we plan to use the stacked semantics as a ``middle ground'' on which the small-step and machine semantics ought to agree, when we give results about the relationship between the small-step and stacked semantics we mainly focus on the direction that goes from the former to the latter. In Section \ref{stacked initiality and reachability} we alluded to some sort of relationship between small-step configurations of the form $(C,M)$ and stacked configurations of the form $(C,M)^\emptyset.\epsilon$. To formalize this relationship, we define the following function:

$$
\fromsmallstep(C,M) = (C,M)^{\emptyset}.\epsilon.
$$

\noindent Note that although we focus on one direction of this relationship, the set of small-step configurations and the set $\initialstacked$ of initial stacked configurations are effectively in bijection, since $\fromsmallstep$ is trivially invertible. The most relevant result that we give in this sub-section is that a single reduction step in the small-step semantics can always be simulated by a reduction sequence in the stacked semantics.

\begin{theoremEnd}{lem}\label{small step to stacked correspondence}
Let $\stred^+$ be the transitive closure of $\stred$. If $(C,M)\red(C',M')$, then for every $\vec k$ and $X$ we have 
$$
(C,M)^{\vec k}.X \stred^+ (C',M')^{\vec k}.X.
$$
Furthermore, if $M\equiv E[\boxt_T(\lift N)]$ then we also have $C'\equiv C, M'\equiv E[(\vec\ell, D, \vec{\ell'})]$ and 
$$(id_Q,N\vec\ell)^{\vec\ell}.X
    \stred^+ (D,\vec{\ell'})^{\vec\ell}.X,
$$
for all $X$ and for $(Q,\vec\ell)=\freshlabels(N,T)$.
\end{theoremEnd}
\begin{proofEnd}
By induction on the derivation of $(C,M)\red(C',M')$. We distinguish two cases. If $M \not\equiv E[\boxt_T(\lift N)]$, then we immediately conclude$(C,M)^{\vec k}.X\stred(C',M')^{\vec k}.X$, for all $\vec k$ and $X$, by the \textit{head} rule. Otherwise, if $M \equiv E[\boxt_T(\lift N)]$, we proceed by cases on $E$:
    \begin{itemize}
    \item Case $E\equiv[\cdot]$.  In this case we have
    $$
    \frac{(Q,\vec\ell)=\freshlabels(N,T) \quad (id_Q,N\vec\ell) \red \dots \red (D,\vec{\ell'})}
    {(C,\boxt_T(\lift N)) \red (C, (\vec\ell,D,\vec{\ell}))}
    \textit{box}
    $$
    \emergencystretch=15pt
    Let $(C_1,M_1)\red(C_2,M_2),(C_2,M_2)\red(C_3,M_3),\dots, (C_{n-1},M_{n-1})\red(C_n,M_n)$ be the sequence of reduction steps denoted by $(id_Q,N\vec\ell) \red \dots \red (D,\vec{\ell'})$, such that $(id_Q,N\vec\ell) \equiv (C_1,M_1)$ and $(C_n,M_n)\equiv (D,\vec{\ell'})$. By applying the inductive hypothesis to each of these reduction steps, we get $(id_Q,N\vec\ell)^{\vec \ell}.X\stred^+(C_2,M_2)^{\vec \ell}.X,(C_2,M_2)^{\vec \ell}.X\stred^+(C_3,M_3)^{\vec \ell}.X,$ up until $(C_{n-1},M_{n-1})^{\vec \ell}.X\stred^+(D,\vec{\ell'})^{\vec \ell}.X$, for all $X$, and therefore $(id_Q,N\vec\ell)^{\vec \ell}.X\stred^+(D,\vec{\ell'})^{\vec \ell}.X$, for all $X$, by the transitivity of $\stred^+$. We therefore consider the following derivation:
    $$
    \frac{(Q,\vec\ell) = \freshlabels(N,T)}
    {(C,\boxt_T(\lift N))^{\vec k}.X' \stred (id_Q,N\vec\ell)^{\vec\ell}.(C,\boxt_T(\lift N))^{\vec k}.X'
    }\textit{step-in}
    $$
    \emergencystretch=20pt
    and by appropriately setting $X = (C,\boxt_T(\lift N))^{\vec k}.X'$ we get $(C,\boxt_T(\lift N))^{\vec k}.X' \stred^+ (D,\vec{\ell'})^{\vec\ell}.(C,\boxt_T(\lift N))^{\vec k}.X'$, for all $X'$, thanks to the transitivity of $\stred^+$. Next, we consider the following derivation:
    $$
    \frac{\void}
    {(D,\vec{\ell'})^{\vec\ell}.(C,\boxt_T(\lift N))^{\vec k}.X' \stred (C,(\vec\ell,D,\vec{\ell'}))^{\vec k}.X'
    }\textit{step-out}
    $$
    by which we conclude $(C,\boxt_T(\lift N))^{\vec k}.X' \stred^+ (C,(\vec\ell, D, \vec{\ell'}))^{\vec k}.X'$, for all $\vec k$ and $X'$, thanks to the transitivity of $\stred^+$.
    
    \item Case $E\equiv FP$. In this case we have
    $$
    \frac{(C, F[\boxt_T(\lift N)]) \red (C',M')}
    {(C,F[\boxt_T(\lift N)]P) \red (C',M'P)
    }\textit{ctx-app-left}
    $$
    \emergencystretch=20pt
    By inductive hypothesis we know that $C\equiv C'$, that $M' \equiv F[(\vec\ell, D, \vec{\ell'})]P$ and that $(id_Q,N\vec\ell)^{\vec\ell}.X
    \stred^+ (D,\vec{\ell'})^{\vec\ell}.X$, for all $X$ and for $(Q,\vec\ell) = \freshlabels(N,T)$. By the same reasoning used in the $E\equiv[\cdot]$ case we can derive $(C,E[\boxt_T(\lift N)])^{\vec k}.X' \stred(id_Q,N\vec\ell)^{\vec\ell}.(C,E[\boxt_T(\lift N)])^{\vec k}.X'$ by the \textit{step-in} rule. Similarly, by the \textit{step-out} rule we can derive $(D,\vec{\ell'})^{\vec\ell}.(C,E[\boxt_T(\lift N)])^{\vec k}.X' \stred (C,E[(\vec\ell,D,\vec{\ell'})])^{\vec k}.X'$. Therefore, by setting $X=(C,E[\boxt_T(\lift N)])^{\vec k}.X'$ we conclude $(C,E[\boxt_T(\lift N)])^{\vec k}.X' \stred^+ (C,E[(\vec\ell, D, \vec{\ell'})])^{\vec k}.X'$, for all $\vec k$ and $X'$, thanks to the transitivity of $\stred^+$.
    \item Case $E\equiv VF$. This case is proven in the same way as the $E\equiv FP$ case, except with \textit{ctx-app-right} instead of \textit{ctx-app-left}.
    
    \item Case $E\equiv \tuple{F,P}$. This case is proven in the same way as the $E\equiv FP$ case, except with \textit{ctx-tuple-left} instead of \textit{ctx-app-left}.
    
    \item Case $E\equiv \tuple{V,F}$. This case is proven in the same way as the $E\equiv FP$ case, except with \textit{ctx-tuple-right} instead of \textit{ctx-app-left}.
    
    \item Case $E\equiv \letin{\tuple{x,y}}{F}{P}$. This case is proven in the same way as the $E\equiv FP$ case, except with \textit{ctx-let} instead of \textit{ctx-app-left}.
    
    \item Case $E\equiv \force F$. This case is proven in the same way as the $E\equiv FP$ case, except with \textit{ctx-force} instead of \textit{ctx-app-left}.
    
    \item Case $E\equiv \boxt_U F$. This case is proven in the same way as the $E\equiv FP$ case, except with \textit{ctx-box} instead of \textit{ctx-app-left}.
    
    \item Case $E\equiv \apply(F,P)$. This case is proven in the same way as the $E\equiv FP$ case, except with \textit{ctx-apply-left} instead of \textit{ctx-app-left}.
    
    \item Case $E\equiv \apply(V,F)$. This case is proven in the same way as the $E\equiv FP$ case, except with \textit{ctx-apply-right} instead of \textit{ctx-app-left}.
    
    \end{itemize}
\end{proofEnd}

\begin{cor}
\label{small step to stacked correspondence cor}
Suppose $(C,M)$ and $(C',M')$ are two small-step configurations. If $(C,M) \red (C',M')$, then $\fromsmallstep(C,M) \stred^+ \fromsmallstep(C',M')$.
\end{cor}
\begin{proof}
The claim follows immediately from Lemma \ref{small step to stacked correspondence} by setting $\vec k = \emptyset$ and $X = \epsilon$.
\end{proof}

A weaker result holds in the other direction. Specifically, a reduction sequence in the stacked semantics can be simulated by a reduction sequence in the small-step semantics only if the former begins and ends on configurations of the same length and all of the intermediate configurations have length greater or equal to that of the endpoints. In other words, the small-step semantics can only simulate computations that begin and end in the same stack frame. This is to be expected, since in the small-step semantics the stack of sub-reductions introduced by the boxing operator exists only in the derivation tree of the reduction sequence, and not in the reduction sequence itself, so a computation $(C,M)\red^*(D,N)$ where $(C,M)$ and $(D,N)$ belong to different sub-reductions is effectively meaningless.

\begin{theoremEnd}{lem}\label{same-level correspondence}
\emergencystretch=10pt
Let $(C,M)^{\vec k}.X$ and $(D,N)^{\vec k}.X$ be two stacked configurations of equal length $m$. If $(C,M)^{\vec k}.X \stred^+ (D,N)^{\vec k}.X$ and all the intermediate configurations in this reduction have length $m$ or greater, then $(C,M) \red^+ (D,N)$.
\end{theoremEnd}
\begin{proofEnd}
By induction on the length of the reduction $(C,M)^{\vec k}.X \stred^+ (D,N)^{\vec k}.X$:
\begin{itemize}
    \item Case of $1$. In this case we have $(C,M)^{\vec k}.X \stred (D,N)^{\vec k}.X$. The only rule which is consistent with the hypothesis is \textit{head}, so we know that $(C,M) \red (D,N)$ and the claim is trivially true.
    
    \item Case of $n+1$. In this case we have $(C,M)^{\vec k}.X \stred (C',M')^{\vec\ell}.X' \stred^+ (D,N)^{\vec k}.X$ and we proceed by cases on the introduction of $(C,M)^{\vec k}.X \stred (C',M')^{\vec \ell}.X'$:
    \begin{itemize}
        \item Case of \textit{head}. In this case we have $C'\equiv C, \vec {\ell} = \vec k$ and $X'=X$ and we know that $(C,M)\red (C',M')$.By inductive hypothesis we also know that $(C',M')\red^+ (D,N)$, so by the transitivity of $\red^+$ we conclude $(C,M)\red^+ (D,N)$.
        
        \emergencystretch=20pt
        \item Case of \textit{step-in}. In this case we know that $M\equiv E[\boxt_T(\lift N)]$ and we have that $(C,E[\boxt_T(\lift N)])^{\vec k}.X \stred (id_Q,N\vec\ell)^{\vec\ell}.(C,E[\boxt_T(\lift N)])^{\vec k}.X$, where $(Q,\vec\ell) = \freshlabels(N,T)$. We know that $(id_Q,N\vec\ell)^{\vec\ell}.(C,E[\boxt_T(\lift N)])^{\vec k}.X$ is a configuration of length $m+1$, so in order for it to eventually reduce to $(D,N)^{\vec k}.X$ there must be a shortening reduction somewhere between $(id_Q,N\vec\ell)^{\vec\ell}.(C,E[\boxt_T(\lift N)])^{\vec k}.X$ and $(D,N)^{\vec k}.X$. Because the only shortening reductions are introduced by the \textit{step-out} rule we must have
        \begin{align*}
            (id_Q,N\vec\ell)^{\vec\ell}.(C,E[\boxt_T(\lift N)])^{\vec k}.X
            &\stred^+ (D',\vec{\ell'})^{\vec\ell}.(C,E[\boxt_T(\lift N)])^{\vec k}.X\\
            &\stred (C,E[(\vec\ell,D',\vec{\ell'})])^{\vec k}.X\\
            &\stred^* (D,N)^{\vec k}.X,
        \end{align*}
        where every intermediate configuration in $(id_Q,N\vec\ell)^{\vec\ell}.(C,E[\boxt_T(\lift N)])^{\vec k}.X \stred^+ (D',\vec{\ell'})^{\vec\ell}.(C,E[\boxt_T(\lift N)])^{\vec k}.X$ has length greater or equal than $m+1$. By applying the inductive hypothesis on this reduction, which has length less than $n$, we get $(id_Q,N\vec\ell) \red^+ (D',\vec{\ell'})$ and therefore $(C,\boxt_T(\lift N)) \red (C,(\vec\ell,D',\vec{\ell'}))$ by the \textit{box} rule. By applying Theorem \ref{context reduction lemma} we also get $(C,E[\boxt_T(\lift N)]) \red (C,E[(\vec\ell,D',\vec{\ell'})])$. At this point, if $(C,E[(\vec\ell,D',\vec{\ell'})])^{\vec k}.X=(D,N)^{\vec k}.X$ then trivially  $(C,E[(\vec\ell,D',\vec{\ell'})])=(D,N)$ and the claim is proven. Otherwise, we have a reduction sequence $(C,E[(\vec\ell,D',\vec{\ell'})])^{\vec k}.X \stred^+ (D,N)^{\vec k}.X,$ of length less than $n$, so by inductive hypothesis we get $(C,E[(\vec\ell,D',\vec{\ell'})])\red^+(D,N)$ and conclude $(C,E[\boxt_T(\lift N)])\stred^+(D,N)$.
        
        \item Case of \textit{step-out}. This case is impossible since it immediately violates the hypothesis that every intermediate configuration is of length $m$ or greater.
    \end{itemize}
\end{itemize}
\end{proofEnd}

\begin{cor}
\label{same-level correspondence cor}
\emergencystretch=10pt
Suppose $(C,M)$ and $(D,N)$ are two small-step configurations. Whenever $\fromsmallstep(C,M) \stred^+ \fromsmallstep(D,N)$, we have $(C,M)\red^+ (D,N)$.
\end{cor}
\begin{proof}
The claim follows immediately from Lemma \ref{same-level correspondence} by setting $\vec k = \emptyset$ and $X = \epsilon$.
\end{proof}

These two results will play a relevant role in the rest of the paper, and they can be summarized graphically in the following diagram:

\begin{center}
    \begin{tikzcd}[column sep = 4cm]
    (C,M) \arrow{r}[pos=1]{+} \arrow[d, leftrightarrow, dashed] & (D,N) \arrow[d, leftrightarrow, dashed]\\
    (C,M)^{\vec\ell}.X \arrow[harpoon]{r}[pos=1]{+} &(D,N)^{\vec\ell}.X
    \end{tikzcd}
\end{center}    

\paragraph{}In order to prove the equivalence between the small-step and machine semantics, it is necessary to prove that $\fromsmallstep$ preserves convergence and that whenever $(C,M)$ goes into deadlock, then $\fromsmallstep(C,M)$ also goes into deadlock. If we consider that to converge essentially means to evaluate to a term of a certain form, we can see that the first result is a trivial consequence of the two previous lemmata, as can be seen by this specific instance of the diagram that we just presented:

\begin{center}
    \begin{tikzcd}[column sep = 4cm]
    (C,M) \arrow{r}[pos=1]{+} \arrow[d, equal, "\fromsmallstep" description] & (D,N) \arrow[d, equal, "\fromsmallstep" description]\\
    (C,M)^\emptyset.\epsilon \arrow[harpoon]{r}[pos=1]{+} & (D,N)^\emptyset.\epsilon
    \end{tikzcd}
\end{center} 
More formally, we give the following result.

\begin{prop}\label{preservation of convergence fromsmallstep}
$(C,M)\converges$ if and only if $\fromsmallstep(C,M)\converges$.
\end{prop}
\begin{proof}
It is easy to see that a small-step configuration $(C,M)$ converges if and only if $(C,M)\red^*(D,V)$ for some $D,V$. It is equally easy to see that a stacked configuration $(C,M)^\emptyset.\epsilon$ converges if and only if $(C,M)^\emptyset.\epsilon \stred^* (D,V)^{\emptyset}.\epsilon$ for some $D,V$. As a result, the claim follows trivially from corollaries \ref{small step to stacked correspondence cor} and \ref{same-level correspondence cor}.
\end{proof}

A similar intuition applies to the deadlocking case. However, because small-step configurations can go into deadlock because of a sub-reduction introduced by \textit{box}, the proof of this second result is not a trivial consequence of the aforementioned diagram.

\begin{theoremEnd}{lem}\label{small-step deadlock to stacked deadlock}
\emergencystretch=15pt
Suppose $(C,M)$ is a small-step configuration. Whenever $(C,M)\deadlocks$, we have $\fromsmallstep(C,M)\deadlocks$.
\end{theoremEnd}
\begin{proofEnd}
We proceed by induction on $(C,M)\deadlocks$:
\begin{itemize}
    \item Case in which $(C,M)$ is irreducible, $M\not\equiv V$ and $M\not\equiv E[\boxt_T(\lift N)]$. In this case consider $\fromsmallstep(C,M)=(C,M)^\emptyset.\epsilon$. This configuration cannot be reduced by the \textit{head} rule, since this would imply the reducibility of $(C,M)$, nor by the \textit{step-in} rule, since $M\not\equiv E[\boxt_T(\lift N)]$, nor by \textit{step-out}, since $M\not\equiv V$ and the rest of the stack is empty. Therefore $(C,M)^\emptyset.\epsilon$ is irreducible and because $M\not\equiv V$ we conclude $(C,M)^\emptyset.\epsilon \deadlocks$.
    
    \emergencystretch=30pt
    \item Case of $(C,M) \red (D,N)$ and $(D,N) \deadlocks$. By Lemma \ref{small step to stacked correspondence} we know $\fromsmallstep(C,M) \stred^+ \fromsmallstep(D,N)$. By inductive hypothesis we know that $\fromsmallstep(D,N) \deadlocks$ and conclude that $\fromsmallstep(C,M)\deadlocks$.
    
    \item Case of $M\equiv E[\boxt_T(\lift N)]$ and $(id_Q,N\vec\ell)\deadlocks$, where $(Q,\vec\ell)=\freshlabels(N,T)$. In this case we have $\frommachine(C,E[\boxt_T(\lift N)]) = (C,E[\boxt_T(\lift N)])^\emptyset.\epsilon \stred (id_Q,N\vec\ell)^{\vec\ell}.(C,E[\boxt_T(\lift N)])^\emptyset.\epsilon$ by the \textit{step-in} rule and inductive hypothesis we know that $\frommachine(id_Q,N\vec\ell) = (id_Q,N\vec\ell)^\emptyset.\epsilon \deadlocks$. From this and Lemma \ref{adding a stack does not solve deadlock} we get that $(id_Q,N\vec\ell)^{\vec k}.X \deadlocks$ for all $\vec k$ and $X$, including $\vec k = \vec \ell$ and $X = (C,E[\boxt_T(\lift N)])^\emptyset.\epsilon$, so we know $(id_Q,N\vec\ell)^{\vec\ell}.(C,E[\boxt_T(\lift N)])^\emptyset.\epsilon \deadlocks$ and conclude $(C,E[\boxt_T(\lift N)])^\emptyset.\epsilon \deadlocks$.
    
    \item Case of $M\equiv E[\boxt_T(\lift N)]$ and $(id_Q,N\vec\ell) \red^* (D,V)$, where $(Q,\vec\ell)=\freshlabels(N,T)$ and $V$ is not a label tuple. In this case we have $\frommachine(C,E[\boxt_T(\lift N)]) = (C,E[\boxt_T(\lift N)])^\emptyset.\epsilon \stred (id_Q,N\vec\ell)^{\vec\ell}.(C,E[\boxt_T(\lift N)])^\emptyset.\epsilon$ by the \textit{step-in} rule and we know $(id_Q,N\vec\ell)^{\vec\ell}.(C,E[\boxt_T(\lift N)])^\emptyset.\epsilon \stred^* (D,V)^{\vec\ell}.(C,E[\boxt_T(\lift N)])^\emptyset.\epsilon$ by Lemma \ref{small step to stacked correspondence}. Since $(D,V)^{\vec\ell}.(C,E[\boxt_T(\lift N)])^\emptyset.\epsilon$ cannot be reduced by the \textit{head} rule (since $V$ is a value), nor by the \textit{step-in} rule (since $V$ cannot be of the form $F[\boxt_T(\lift P)]$ for any $F,P$), nor by the \textit{step-out} rule (since $V$ is not a label tuple), and the rest of the stack is not empty, we get that $(D,V)^{\vec\ell}.(C,E[\boxt_T(\lift N)])^\emptyset.\epsilon \deadlocks$ and conclude that $(C,E[\boxt_T(\lift N)])^\emptyset.\epsilon \deadlocks$.
\end{itemize}
\end{proofEnd}

\subsection{From Machine to Stacked Semantics}
\label{subsection machine to stacked}
\label{machine to stacked reductions}

We now explore the ``other side'' of the equivalence between the small-step and machine semantics, that is, the relationship between the machine semantics and the stacked semantics, with emphasis on the direction that goes from the former to the latter. Whereas the relationship between small-step and stacked configurations is trivial and mainly concerns initial configurations, the relationship between machine and stacked configurations is at the same time more pervasive (it holds for all configurations, not just initial ones) and slightly more complicated to define. For this purpose, we formalize this relationship via a $\frommachine$ function which maps machine configurations into stacked configurations:

\begin{align*}
    \build(C,M,\epsilon) &=(C,M)^\emptyset.\epsilon\\
    \build(C,M,\contfarg(N).S) &= \build(C,MN,S)\\
    \build(C,M,\contfapp(V).S) &= \build(C,VM,S)\\
    \build(C,M,\contalabel(N).S) &= \build(C,\apply(M,N),S)\\
    \build(C,M,\contacirc(V).S) &= \build(C,\apply(V,M),S)\\
    \build(C,M,\conttright(N).S) &= \build(C,\tuple{M,N},S)\\
    \build(C,M,\conttleft(V).S) &= \build(C,\tuple{V,M},S)\\
    \build(C,M,\contbox(Q,\vec\ell).S) &= \build(C,\boxt_T M,S) \textnormal{ where } \emptyset;Q\vdash \vec\ell:T\\
    \build(C,M,\contsub(D,N,\vec\ell,T).S) &= (C,M)^{\vec\ell}.\build(D,\boxt_T(\lift N),S)\\
    \build(C,M,\contlet(x,y,N).S) &= \build(C,\letin{\tuple{x,y}}{M}{N},S)\\
    \build(C,M,\contforce.S) &= \build(C,\force M,S).
\end{align*}

\noindent Before we discuss this definition in greater detail, it is worth noting that by restricting the domain of $\frommachine$ to the set $\initialmachine$ of initial machine configurations, we trivially have a bijection between $\initialmachine$ and the set $\initialstacked$ of initial stacked configurations, which is analogous to the one we had in Section \ref{section from small-step to stacked configurations} between small-step configurations and initial stacked configurations.

\paragraph{}Informally, the $\frommachine$ function takes the term that the machine configuration is currently focused on and gradually unwinds the stack to rebuild the term being evaluated in its entirety. The only case in which we actually change the current stack frame in the resulting stacked configuration is, unsurprisingly, when we encounter a continuation of type $\contsub$. As a result, the machine stack is encoded in the resulting stacked configuration partly as the configuration stack itself, and partly as the structure of the terms contained within each individual stack frame. This structure is not arbitrary. In fact, we have that it always corresponds to the structure of an evaluation context, as stated in the following proposition.

\begin{theoremEnd}{prop}\label{stack to context lemma}
If $\build(C,M,S) = (D,N)^{\vec\ell}.X$, then $C\equiv D$ and $N$ is of the form $E[M]$ for some evaluation context $E$.
\end{theoremEnd}
\begin{proofEnd}
By induction on $|S|$, the size of $S$, defined in the natural manner. If $|S|=0$, it means that $S=\epsilon$. In this case we have $\build(C,M,\epsilon)=(C,M)^\emptyset.\epsilon$ and the claim is true for $E\equiv [\cdot]$. If $|S|=n+1$, suppose $S\equiv H.S'$, where $|S'|=n$. We proceed by cases on $H$:
\begin{itemize}
    \emergencystretch=20pt
    \item Case $H\equiv \contfarg(N)$. In this case we have that $\build(C,M,\contfarg(N).S')=\build(C,MN,S')$. We know that $MN\equiv E[M]$ for $E\equiv [\cdot]N$. Furthermore, by inductive hypothesis we get that $\build(C,E[M],S')=(C,E'[E[M]])^{\vec\ell}.X$ for some $E'$ and by Proposition \ref{context propagation} we conclude $\build(C,M,\contfarg(N).S')=(C,E''[M])^{\vec\ell}.X$ for some $E''$.
    
    \emergencystretch=25pt
    \item Case $H\equiv \contfapp(V)$. In this case we have that $\build(C,M,\contfapp(V).S')=\build(C,VM,S')$. We know that $VM\equiv E[M]$ for $E\equiv V[\cdot]$. Furthermore, by inductive hypothesis we get that $\build(C,E[M],S')=(C,E'[E[M]])^{\vec\ell}.X$ for some $E'$ and by Proposition \ref{context propagation} we eventually conclude that $\build(C,M,\contfapp(V).S') = (C,E''[M])^{\vec\ell}.X$ for some $E''$.
    
    \emergencystretch=20pt
    \item Case $H\equiv \contalabel(N)$. In this case we have that $\build(C,M,\contalabel(N).S')=\build(C,\apply(M,N),S')$. We know that $\apply(M,N)\equiv E[M]$ for $E\equiv \apply([\cdot],N)$. Furthermore, by inductive hypothesis we get $\build(C,E[M],S')=(C,E'[E[M]])^{\vec\ell}.X$ for some $E'$ and by Proposition \ref{context propagation} we conclude that $\build(C,M,\contalabel(N).S')=(C,E''[M])^{\vec\ell}.X$ for some $E''$.
    
    \emergencystretch=30pt
    \item Case $H\equiv \contacirc(V)$. In this case we have $\build(C,M,\contacirc(V).S')=\build(C,\apply(V,M),S')$. We know that $\apply(V,M) \equiv E[M]$ for $E\equiv \apply(V,[\cdot])$. Furthermore, by inductive hypothesis we get $\build(C,E[M],S')=(C,E'[E[M]])^{\vec k}.X$ for some $E'$ and by Proposition \ref{context propagation} we can conclude that $\build(C,M,\contacirc(V).S')=(C,E''[M])^{\vec k}.X$ for some $E''$.
    
    \item Case $H\equiv \conttright(N)$. In this case we have $\build(C,M,\conttright(N).S')=\build(C,\tuple{M,N},S')$. We know that $\tuple{M,N}\equiv E[M]$ for $E\equiv \tuple{[\cdot],N}$. Furthermore, by inductive hypothesis we get that $\build(C,E[M],S')=(C,E'[E[M]])^{\vec\ell}.X$ for some $E'$ and by Proposition \ref{context propagation} we eventually conclude that $\build(C,M,\conttright(N).S')=(C,E''[M])^{\vec\ell}.X$ for some $E''$.
    
    \item Case $H\equiv \conttleft(V)$. In this case we have $\build(C,M,\conttleft(V).S')=\build(C,\tuple{V,M},S')$. We know that $\tuple{V,M} \equiv E[M]$ for $E\equiv \tuple{V,[\cdot]}$. Furthermore, by inductive hypothesis we get that $\build(C,E[M],S')=(C,E'[E[M]])^{\vec\ell}.X$ for some $E'$ and by Proposition \ref{context propagation} we eventually conclude that $\build(C,M,\conttleft(V).S')=(C,E''[M])^{\vec \ell}.X$ for some $E''$.
    
    \emergencystretch=20pt
    \item Case $H\equiv \contbox(Q,\vec\ell)$. In this case we have that $\build(C,M,\contbox(Q,\vec\ell).S')=\build(C,\boxt_T M,S')$, where $\emptyset;Q\vdash \vec\ell:T$. We know that $\boxt_T M \equiv E[M]$ for $E\equiv \boxt_T \,[\cdot]$. Furthermore, by inductive hypothesis we get $\build(C,E[M],S')=(C,E'[E[M]])^{\vec\ell}.X$ for some $E'$ and by Proposition \ref{context propagation} we eventually conclude that $\build(C,M,\contbox(Q,\vec\ell').S')=(C,E''[M])^{\vec k}.X$ for some $E''$.
    
    \item Case $H\equiv \contsub(D,N,\vec\ell,T)$.  In this case $\build(C,M,\contsub(D,N,\vec\ell,T).S')=(C,M)^{\vec\ell}.\build(D, \boxt_T(\lift N), S')$ and the claim is trivially true for $E\equiv [\cdot]$.
    
    \emergencystretch=10pt
    \item Case $H\equiv \contlet(x,y,N)$. In this case we have $\build(C,M,\contlet(x,y,N).S')= \build(C,\letin{\tuple{x,y}}{M}{N},S')$. We know that $\letin{\tuple{x,y}}{M}{N}\equiv E[M]$ for $E\equiv \letin{\tuple{x,y}}{[\cdot]}{N}$. Furthermore, by inductive hypothesis we get $\build(C,E[M],S')=(C,E'[E[M]])^{\vec\ell}.X$ for some $E'$ and by Proposition \ref{context propagation} we eventually conclude that $\build(C,M,\contlet(x,y,N).S')=(C,E''[M])^{\vec\ell}.X$ for some $E''$.
    
    \emergencystretch=35pt
    \item Case $H\equiv \contforce$. In this case we have that $\build(C,M,\contforce.S')=\build(C,\force M,S')$. We also know that $\force M \equiv E[M]$ for $E\equiv \force \,[\cdot]$. Furthermore, by inductive hypothesis we get that $\build(C,E[M],S')=(C,E'[E[M]])^{\vec\ell}.X$ for some $E'$ and by Proposition \ref{context propagation} we eventually conclude that $\build(C,M,\contforce.S')=(C,E''[M])^{\vec k}.X$ for some $E''$.
\end{itemize}
\end{proofEnd}

\paragraph{}The following result goes even further, as it guarantees that the evaluation context $E$ that we introduced in the previous lemma, as well as the locally available labels $\vec\ell$ and the rest of the stack $X$ all depend \textit{exclusively} on the machine stack $S$.

\begin{prop}\label{stack to similarities lemma}
If two machine configurations $(C,M,S)$ and $(D,N,S)$ share the same stack $S$, then $\build(C,M,S)=(C,E[M])^{\vec\ell}.X$ and $\build(D,N,S)=(D,E[N])^{\vec\ell}.X$ for the same $E$, $\vec\ell$ and $X$.
\end{prop}
\begin{proof}
The existence of $E$ is guaranteed by Proposition \ref{stack to context lemma}. The identity of $E,\vec\ell$ and $X$ can be proven trivially by induction on the length of $S$.
\end{proof}

\paragraph{}Because the machine semantics is, intuitively, more fine-grained than the stacked semantics, it comes as no surprise that distinct machine configurations can be mapped by $\frommachine$ to the same stacked configuration. As an example, take the configurations $(C,V,\contfarg(W).\epsilon)$ and $(C,W,\contfapp(V).\epsilon)$, for any values $V,W$. We have
\begin{align*}
    &\frommachine(C,V,\contfarg(W).\epsilon) = \frommachine(C,VW,\epsilon)=(C,VW)^\emptyset.\epsilon,\\
    &\frommachine(C,W,\contfapp(V).\epsilon) = \frommachine(C,VW,\epsilon)=(C,VW)^\emptyset.\epsilon.
\end{align*}
Intuitively, this is due to the fact that the two configurations represents two different phases in the evaluation of the same application $VW$, which in contrast can be evaluated in a single step in the stacked semantics. In fact, most of the rules of the machine semantics only serve to decompose or move around terms (e.g, the \textit{split} or \textit{shift} rules) and because they do not actually evaluate anything, they have no appreciable effect on the corresponding stacked configuration.
This will be a crucial aspect to consider in the coming results, so we ought to formalize it. Let $\amred_b$ be the proper subset of $\amred$ that can be derived by only using rules for $\amred$ whose conclusion $(C,M,S) \amred (D,N,R)$ is such that $\frommachine(C,M,S) = \frommachine(D,N,R)$. These rules are specifically \textit{app-split, app-shift, apply-split, apply-shift, let-split, tuple-split, tuple-shift, tuple-join, box-open, force-open}. Also, let $\amred_r$ be the subset of $\amred$ that can be derived by only using the remaining rules, that is, \textit{app-join, apply-join, box-sub, box-close, let-join, force-close}, such that
$$
(C,M,S) \amred (D,N,R) \iff (C,M,S) \amred_b (D,N,R) \vee (C,M,S) \amred_r (D,N,R).
$$

\noindent The most essential property that we must guarantee is that there is a limit to the number of times we can reduce a machine configuration $(C,M,S)$ without causing any change in the corresponding stacked configuration $\frommachine(C,M,S)$. In other words, we must prove that $\amred_b$ is strongly normalizing.

\begin{theoremEnd}{lem}\label{strong normalization of =>b}
The reduction relation $\amred_b$ is strongly normalizing.
\end{theoremEnd}
\begin{proofEnd}
Let $\operatorname{L}_t$ be a function on terms defined as such:
\begin{align*}
    \operatorname{L}_t(V) &= 0,\\
    \operatorname{L}_t(\boxt_T M) = \operatorname{L}_t(\force M) = \operatorname{L}_t(\letin{\tuple{x,y}}{M}{N}) &= \operatorname{L}_t(M)+1, \\
    \operatorname{L}_t(MN) = \operatorname{L}_t(\apply(M,N)) &= \operatorname{L}_t(M)+\operatorname{L}_t(N)+2, \\
    \operatorname{L}_t(\tuple{M,N}) &= \operatorname{L}_t(M)+\operatorname{L}_t(N)+3.
\end{align*}
Now let $\operatorname{L}_s$ be a function on machine stacks defined as such:
\begin{align*}
    \operatorname{L}_s(\epsilon) = \operatorname{L}_s(\contsub(D,N,\vec\ell,T).S') &= 0, \\
    \operatorname{L}_s(\contbox(Q,\vec\ell).S') = \operatorname{L}_s(\contforce.S') = \operatorname{L}_s(\contlet(x,y,M).S') &= \operatorname{L}_s(S'), \\
    \operatorname{L}_s(\contfarg(M).S') = \operatorname{L}_s(\contalabel(M).S') &= \operatorname{L}_t(M)+ 1 + \operatorname{L}_s(S'),\\
    \operatorname{L}_s(\conttright(M).S') &= \operatorname{L}_t(M)+ 2 + \operatorname{L}_s(S'),\\
    \operatorname{L}_s(\contfapp(V).S') = \operatorname{L}_s(\contacirc(V).S') &= \operatorname{L}_s(S'),\\
    \operatorname{L}_s(\conttleft(V).S') &= 1 + \operatorname{L}_s(S').
\end{align*}
Finally, let $\operatorname{L}(C,M,S)=\operatorname{L}_t(M)+\operatorname{L}_s(S)$. It is trivial to see that both $\operatorname{L}_t$ and $\operatorname{L}_s$, and as a consequence $\operatorname{L}$, are non-negative. Because of this, it is sufficient to show that whenever $(C,M,S) \amred_b (D,N,S')$ then $\operatorname{L}(D,N,S') < \operatorname{L}(C,M,S)$ to prove that $\amred_b$ is strongly normalizing. We proceed by cases on the introduction of $(C,M,S) \amred_b (D,N,S')$:
\begin{itemize}
    \emergencystretch=15pt
    \item Case of \textit{app-split}. In this case $(C,MN,S) \amred_b (C,M,\contfarg(N).S)$. We have $\operatorname{L}(C,MN,S) = \operatorname{L}_t(M) + \operatorname{L}_t(N) + 2 + \operatorname{L}_s(S)$ and $\operatorname{L}(C,M,\contfarg(N).S)=\operatorname{L}_t(M) + \operatorname{L}_t(N) + 1 + \operatorname{L}_s(S)$ and the claim is proven.
    
    \item Case of \textit{app-shift}. In this case $(C,V,\contfarg(N).S) \amred_b (C,N,\contfapp(V).S)$. We have $\operatorname{L}(C,V,\contfarg(N).S) = \operatorname{L}_t(N) + 1 + \operatorname{L}_s(S)$ and $\operatorname{L}(C,N,\contfapp(V).S) =\operatorname{L}_t(N) + \operatorname{L}_s(S)$ and the claim is proven.
    
    \emergencystretch=30pt
    \item Case of \textit{apply-split}. In this case $(C,\apply(M,N),S) \amred_b (C,M,\contalabel(N).S)$. We have $\operatorname{L}(C,\apply(M,N),S) = \operatorname{L}_t(M) + \operatorname{L}_t(N) + 2 + \operatorname{L}_s(S)$ and $\operatorname{L}(C,M,\contalabel(N).S)=\operatorname{L}_t(M) + \operatorname{L}_t(N) + 1 + \operatorname{L}_s(S)$ and the claim is proven.
    
    \item Case of \textit{apply-shift}. In this case $(C,V,\contalabel(N).S) \amred_b (C,N,\contacirc(V).S)$. We have $\operatorname{L}(C,V,\contalabel(N).S) = \operatorname{L}_t(N) + 1 + \operatorname{L}_s(S)$ and $\operatorname{L}(C,N,\contacirc(V).S) =\operatorname{L}_t(N) + \operatorname{L}_s(S)$ and the claim is proven.
    
    \item Case of \textit{let-split}. In this case $(C,\letin{\tuple{x,y}}{M}{N}, S) \amred_b (C,M,\contlet(x,y,N).S)$. We have $\operatorname{L}(C,\letin{\tuple{x,y}}{M}{N}, S) = \operatorname{L}_t(M) + 1 + \operatorname{L}_s(S)$ and $\operatorname{L}(C,M,\contlet(x,y,N).S) = \operatorname{L}_t(M) + \operatorname{L}_s(S)$ and the claim is proven.
    
    \item Case of \textit{tuple-split}. In this case $(C,\tuple{M,N},S) \amred_b (C,M,\conttright(N).S)$. We have $\operatorname{L}(C,\tuple{M,N},S) = \operatorname{L}_t(M) + \operatorname{L}_t(N) + 3 + \operatorname{L}_s(S)$ and $\operatorname{L}(C,M,\conttright(N).S)=\operatorname{L}_t(M) + \operatorname{L}_t(N) + 2 + \operatorname{L}_s(S)$ and the claim is proven.
    
    \item Case of \textit{tuple-shift}. In this case $(C,V,\conttright(N).S) \amred_b (C,N,\conttleft(V).S)$. We have $\operatorname{L}(C,V,\conttright(N).S) = \operatorname{L}_t(N) + 2 + \operatorname{L}_s(S)$ and $\operatorname{L}(C,N,\conttleft(V).S) =\operatorname{L}_t(N) + 1 + \operatorname{L}_s(S)$ and the claim is proven.
    
    \item Case of \textit{tuple-join}. In this case $(C,W,\conttleft(V).S) \amred_b (C,\tuple{V,W},S)$. We have $\operatorname{L}(C,W,\conttleft(V).S) = \operatorname{L}_s(S) + 1$ and $\operatorname{L}(C,\tuple{V,W},S) = \operatorname{L}_s(S)$ and the claim is proven.
    
    \item Case of \textit{box-open}. In this case $(C,\boxt_T M, S) \amred_b (C,M,\contbox(Q,\vec\ell).S)$. We have $\operatorname{L}(C,\boxt_T M, S) = \operatorname{L}_t(M) + 1 + \operatorname{L}_s(S)$ and $\operatorname{L}(C,M,\contbox(Q,\vec\ell).S) = \operatorname{L}_t(M) + \operatorname{L}_s(S)$ and the claim is proven.
    
    \item Case of \textit{force-open}. In this case $(C,\force M, S) \amred_b (C,M,\contforce.S)$. We have $\operatorname{L}(C,\force M, S) = \operatorname{L}_t(M) + 1 + \operatorname{L}_s(S)$ and $\operatorname{L}(C,M,\contforce.S) = \operatorname{L}_t(M) + \operatorname{L}_s(S)$ and the claim is proven.
    
\end{itemize}
\end{proofEnd}

\paragraph{}Now we can give the most significant result of this sub-section, which is similar to the one we gave in Lemma \ref{small step to stacked correspondence}. Namely, we show that a single reduction step in the machine semantics can always be simulated by zero or one steps in the stacked semantics.

\begin{theoremEnd}{lem}\label{machine to stacked correspondence}
\emergencystretch=10pt
Suppose $(C,M,S)$ and $(C',M',S')$ are two machine configurations such that $(C,M,S) \amred (C',M',S')$. The following hold:
\begin{enumerate}
    \item If $(C,M,S)\amred_b(C',M',S')$, then $\frommachine(C,M,S) = \frommachine(C',M',S'),$
    \item If $(C,M,S)\amred_r(C',M',S')$, then $\frommachine(C,M,S) \stred \frommachine(C',M',S').$
\end{enumerate}
\end{theoremEnd}
\begin{proofEnd}
By cases on the introduction of $(C,M,S)\amred (C',M',S')$:
\begin{itemize}
    \item Case of \textit{app-split}. In this case we have $(C,NP,S)\amred_b(C,N,\contfarg(P).S)$ and we immediately conclude $\build(C,N,\contfarg(P).S) = \build(C,NP,S)$ by the definition of $\build$.
    
    \item Case of \textit{app-shift}. In this case we have $(C,V, \contfarg(P).S) \amred (C, P, \contfapp(V).S)$ and we immediately conclude
    \begin{align*}
        \build(C,V, \contfarg(P).S) &= \build(C,VP, S) \\
        &= \build(C, P, \contfapp(V).S).
    \end{align*}
    
    \item Case of \textit{app-join}. In this case we have $(C,V,\contfapp(\lambda x. N).S) \amred_r (C,N[V/x],S)$. By propositions \ref{stack to context lemma} and \ref{stack to similarities lemma} we know that
    \begin{align*}
        \build(C,V,\contfapp(\lambda x. N).S) &=
        \build(C,(\lambda x. N)V,S)\\
        &=(C,E[(\lambda x. N)V])^{\vec\ell}.X,\\
        \\
        \build(C,N[V/x],S) &= (C,E[N[V/x]])^{\vec\ell}.X.
    \end{align*}
    \emergencystretch=20pt
    Because $(C,(\lambda x. N)V)\red(C,N[V/x])$ by the \textit{app} rule, we get $(C,E[(\lambda x. N)V])\red(C,E[N[V/x]])$ by Theorem \ref{context reduction lemma}. Furthermore, by Corollary \ref{context exclusion lemma redex} we know that $E[(\lambda x. N)V]\not\equiv E'[\boxt_T(\lift P)]$ for any $E',P$ and therefore by the \textit{head} rule we conclude
    $$
    (C,E[(\lambda x. N)V])^{\vec\ell}.X \stred (C,E[N[V/x]])^{\vec\ell}.X.
    $$
    
    \emergencystretch=20pt
    \item Case of \textit{apply-split}. In this case we have $(C,\apply(N,P),S)\amred_b(C,N,\contalabel(P).S)$ and we conclude $\build(C,N,\contalabel(P).S) = \build(C,\apply(N,P),S)$ by the definition of $\build$.
    
    \item Case of \textit{apply-shift}. In this case we have $(C,V, \contalabel(P).S) \amred_b (C, P, \contacirc(V).S)$ and we immediately conclude
    \begin{align*}
        \build(C,V, \contalabel(P).S) &= \build(C,\apply(V, P),S) \\
        &= \build(C, P, \contacirc(V).S).
    \end{align*}
    
    \item Case of \textit{apply-join}. In this case we have $(C,\vec k,\contacirc((\vec\ell, D, \vec{\ell'})).S) \amred_r (C',\vec{k'},S)$, where $(C',\vec{k'})=\append(C,\vec k, \vec\ell, D, \vec{\ell'})$. By propositions \ref{stack to context lemma} and \ref{stack to similarities lemma} we know that
    \begin{align*}
        \build(C,\vec k,\contacirc((\vec\ell, D, \vec{\ell'})).S) &= \build(C,\apply((\vec\ell, D, \vec{\ell'}),\vec k),S)\\
        &=(C,E[\apply((\vec\ell, D, \vec{\ell'}),\vec k)])^{\vec{\ell''}}.X,\\
        \\
        \build(C,\vec{k'},S) &= (C,E[\vec{k'}])^{\vec{\ell''}}.X.
    \end{align*}
    \emergencystretch=20pt
    Because $(C,\apply((\vec\ell, D, \vec{\ell'}),\vec k))\red(C,\vec{k'})$ by the \textit{apply} rule, by Theorem \ref{context reduction lemma} we know that $(C,E[\apply((\vec\ell, D, \vec{\ell'}),\vec k)])\red(C,E[\vec{k'}])$. Furthermore, by Corollary \ref{context exclusion lemma redex} we know that $E[\apply((\vec\ell, D, \vec{\ell'}),\vec k)] \not\equiv E'[\boxt_T(\lift P)]$ for any $E',P$ and therefore by the \textit{head} rule we conclude
    $$
    (C,E[\apply((\vec\ell, D, \vec{\ell'}),\vec k)])^{\vec{\ell''}}.X \stred (C,E[\vec{k'}])^{\vec{\ell''}}.X.
    $$
    
    \item Case of \textit{tuple-split}. In this case we have $(C,\tuple{N,P},S)\amred_b(C,N,\conttright(P).S)$ and we immediately conclude $\build(C,N,\conttright(P).S)=\build(C,\tuple{N,P},S)$ by the definition of $\build$.
    
    \item Case of \textit{tuple-shift}. In this case we have $(C,V, \conttright(P).S) \amred_b (C, P, \conttleft(V).S)$ and we immediately conclude
    \begin{align*}
        \build(C,V, \conttright(P).S) &= \build(C,\tuple{V,P}, S) \\&= \build(C, P, \conttleft(V).S).
    \end{align*}
    
    \item Case of \textit{tuple-join}. In this case we have $(C,W,\conttleft(V).S) \amred_b (C,\tuple{V,W},S)$ and we immediately conclude $\build(C,W,\conttleft(V).S) = \build(C,\tuple{V,W},S)$ by the definition of $\build$.
    
    \item Case of \textit{box-open}. In this case we have $(C,\boxt_T N,S) \amred_b (C,N,\contbox(Q,\vec\ell).S)$, where $(Q,\vec\ell)=\freshlabels(N,T)$. Since we know by the definition of $\freshlabels$ that $\emptyset;Q\vdash \vec\ell:T$, we immediately conclude $\build(C,N,\contbox(Q,\vec\ell).S)=\build(C,\boxt_T N,S)$.
    
    \item Case of \textit{box-sub}. In this case $(C,\lift N, \contbox(Q,\vec\ell).S) \amred_r (id_Q, N\vec\ell, \contsub(C,N,\vec\ell,T).S)$, where $\emptyset;Q\vdash \vec\ell:T$. By propositions \ref{stack to context lemma} and \ref{stack to similarities lemma} we know that
    \begin{align*}
        \build(C,\lift N, \contbox(Q,\vec\ell).S) &= \build(C,\boxt_T(\lift N),S)\\
        &=(C,E[\boxt_T(\lift N)])^{\vec{\ell'}}.X,\\
        \\
        \build(id_Q, N\vec\ell,\contsub(C,N,\vec\ell,T).S) &= (id_Q,N\vec\ell)^{\vec\ell}.\build(C,\boxt_T(\lift N),S)\\
        &= (id_Q,N\vec\ell)^{\vec\ell}.(C,E[\boxt_T(\lift N)])^{\vec{\ell'}}.X,
    \end{align*}
    and by the \textit{step-in} rule we conclude
    $$
    (C,E[\boxt_T(\lift N)])^{\vec{\ell'}}.X \stred (id_Q,N\vec\ell)^{\vec\ell}.(C,E[\boxt_T(\lift N)])^{\vec{\ell'}}.X.
    $$

    \item Case of \textit{box-close}. In this case we have $(D,\vec{\ell'},\contsub(C,N,\vec\ell,T).S) \amred_r (C,(\vec\ell,D,\vec{\ell'}),S)$. By propositions \ref{stack to context lemma} and \ref{stack to similarities lemma} we know that
    \begin{align*}
        \build(D,\vec{\ell'},\contsub(C,N,\vec\ell,T).S) &=
        (D,\vec{\ell'})^{\vec\ell}.\build(C,\boxt_T(\lift N),S)\\
        &=(D,\vec{\ell'})^{\vec\ell}.(C,E[\boxt_T(\lift N)])^{\vec k}.X,\\
        \\
        \build(C,(\vec\ell,D,\vec{\ell'}),S) &=
        (C,E[(\vec\ell,D,\vec{\ell'})])^{\vec k}.X,
    \end{align*}
    and by the \textit{step-out} rule we conclude
    $$
    (D,\vec{\ell'})^{\vec\ell}.(C,E[\boxt_T(\lift N)])^{\vec k}.X \stred (C,E[(\vec\ell,D,\vec{\ell'})])^{\vec k}.X.
    $$
    
    \emergencystretch=20pt
    \item Case of \textit{let-split}. In this case we have $(C,\letin{\tuple{x,y}}{N}{P},S)\amred_b(C,N,\contlet(x,y,P).S)$ and we conclude $\build(C,N,\contlet(x,y,P).S) = \build(C,\letin{\tuple{x,y}}{N}{P},S)$ by the definition of $\build$.
    
    \item Case of \textit{let-join}. In this case we have $(C,\tuple{V,W},\contlet(x,y,N).S) \amred_r (C,N[V/x][W/y],S)$. By propositions \ref{stack to context lemma} and \ref{stack to similarities lemma} we know that
    \begin{align*}
        \build(C,\tuple{V,W},\contlet(x,y,N).S) &= \build(C,\letin{\tuple{x,y}}{\tuple{V,W}}{N},S)\\
        &= (C,E[\letin{\tuple{x,y}}{\tuple{V,W}}{N}])^{\vec\ell}.X,\\
        \\
        \build(C,N[V/x][W/y],S) &= (C,E[N[V/x][W/y]])^{\vec\ell}.X.
    \end{align*}
    Because $(C,\letin{\tuple{x,y}}{\tuple{V,W}}{N}) \red (C,N[V/x][W/y])$ by the \textit{let} rule, we get $(C,E[\letin{\tuple{x,y}}{\tuple{V,W}}{N}]) \red (C,E[N[V/x][W/y]])$ by Theorem \ref{context reduction lemma} and by the \textit{head} rule we conclude
    $$
    (C,E[\letin{\tuple{x,y}}{\tuple{V,W}}{N}])^{\vec\ell}.X \stred (C,E[N[V/x][W/y]])^{\vec\ell}.X.
    $$
    
    \item Case of \textit{force-open}. In this case we have $(C,\force N, S) \amred_b (C, N, \contforce.S)$ and we immediately conclude $\build(C,N,\contforce.S)=\build(C, \force N, S)$ by the definition of $\build$.
    
    \item Case of \textit{force-close}. In this case we have $(C,\lift N, \contforce.S) \amred_r (C, N, S)$. By propositions \ref{stack to context lemma} and \ref{stack to similarities lemma} we know that
    \begin{align*}
        \build(C,\lift N, \contforce.S)&= \build(C,\force (\lift N),S)\\
        &= (C,E[\force(\lift N)])^{\vec\ell}.X,\\
        \\
        \build(C,N,S)&=(C,E[N])^{\vec\ell}.X.
    \end{align*}
    Because $(C,\force(\lift N)) \red (C,N)$ by the \textit{force} rule, we get $(C,E[\force(\lift N)]) \red  (C,E[N])$ by Theorem \ref{context reduction lemma} and by the \textit{head} rule we conclude
    $$
    (C,E[\force(\lift N)])^{\vec\ell}.X \stred (C,E[N])^{\vec\ell}.X.
    $$
\end{itemize}
\end{proofEnd}

\paragraph{}Like we did in Section \ref{from small-step to stacked reductions}, it is useful to represent the result of the lemma that we just proved in a diagrammatic way, as follows:

\begin{center}
    \begin{tikzcd}[column sep=huge]
    (C,E[M])^{\vec\ell}.X \arrow[r, harpoon] & (D,F[N])^{\vec k}.Y\\
    (C,M,S) \arrow[u, dashed, "\frommachine" description]& (D,N,R) \lar[Leftarrow][pos=0]{r} \arrow[u, dashed, "\frommachine" description]
    \end{tikzcd}
    \begin{tikzcd}[column sep=tiny]
    & (C,E[M])^{\vec\ell}.X  &\\
    (C,M,S) \arrow[ur, dashed, "\frommachine" description] && (D,N,R) \arrow[Leftarrow]{ll}[pos=0]{b} \arrow[ul, dashed, "\frommachine" description]
    \end{tikzcd}
\end{center}
Or, more generally and synthetically, as a single diagram:
\begin{center}
    \begin{tikzcd}[column sep=4cm]
    (C,E[M])^{\vec\ell}.X \arrow[harpoon]{r}[pos=1]{*} & (D,F[N])^{\vec k}.Y\\
    (C,M,S) \arrow[Rightarrow]{r}[pos=1]{+} \arrow[u, dashed, "\frommachine" description] & (D,N,R) \arrow[u, dashed, "\frommachine" description]
    \end{tikzcd}
\end{center}

\paragraph{}In order to prove the equivalence between the small-step and machine semantics in the coming Section \ref{subsection equivalence small-step machine}, it is necessary to show that $\frommachine$ preserves the convergence of reachable configurations and that whenever $(C,M,S)$ goes into deadlock then also $\frommachine(C,M,S)$ goes into deadlock. Let us start with deadlock. Informally, a machine configuration goes into deadlock when it evaluates to an irreducible configuration whose stack is not empty. Similarly, a stacked configuration goes into deadlock when it evaluates to an irreducible configuration in which either the stack is not empty, or the term in the head is not a value. Because we know that a computation in the machine semantics can always be simulated by a computation in the stacked semantics, to show that whenever a machine configuration goes into deadlock then the corresponding stacked configuration also goes into deadlock it is sufficient to show that every irreducible machine configuration $(C,M,S)$ such that $S$ is not empty is mapped by $\frommachine$ to an irreducible stacked configuration $(C,E[M])^{\vec\ell}.X$ such that either $E[M]$ is not a value or $X$ is not empty.

\begin{theoremEnd}{lem}\label{machine deadlocked implies stacked deadlocked base}
Suppose $(C,V,S)$ is a machine configuration. If $(C,V,S)$ is irreducible and $S\neq \epsilon$, then $\frommachine(C,V,S)=(C,E[V])^{\vec\ell}.X$ is irreducible and either $E[V]$ is not a value or $\vec\ell \neq \emptyset, X \neq \epsilon$.
\end{theoremEnd}
\begin{proofEnd}
Suppose $S=H.S'$. We proceed by cases on $H$:
\begin{itemize}
    \item Case of $\contfarg(N)$. This case is impossible, since it would allow $(C,V,S)$ to be reduced by the \textit{app-shift} rule, regardless of $V$, contradicting the hypothesis.
    
    \emergencystretch=10pt
    \item Case of $\contfapp(W)$. In this case $W\not\equiv \lambda x. N$, since otherwise $(C,V,\contfapp(\lambda x.N).S')$ would be reducible by the \textit{app-join} rule. We have $\frommachine(C,V,\contfapp(W).S') = \frommachine(C,WV,S') = (C,E[WV])^{\vec k}.X$. This configuration cannot be reduced by the \textit{head} rule (since by Theorem \ref{context reduction lemma} the reducibility of $(C,E[WV])$ would imply the reducibility of $(C,WV)$, which is absurd since $W\not\equiv \lambda x. N$), nor by the \textit{step-in} rule (since $WV$ is a proto-redex and therefore by Proposition \ref{context exclusion theorem proto} $E[WV] \not\equiv F[\boxt_T(\lift P)]$, for any $F,P$), nor by the \textit{step-out} rule (since $E[WV]$ cannot be a value). Since $(C,E[WV])^{\vec k}.X$ is irreducible and $E[WV]$ is not a value, the claim is proven.
    
    \item Case of $\contalabel(N)$. This case is impossible, since it would allow $(C,V,S)$ to be reduced by the \textit{apply-shift} rule, regardless of $V$, contradicting the hypothesis.
    
    \emergencystretch=20pt
    \item Case of $\contacirc(W)$. In this case $W\not\equiv (\vec\ell, D, \vec{\ell'})$, or else $(C,V,\contacirc((\vec\ell, D, \vec{\ell'})).S')$ would be reducible by the \textit{apply-join} rule. We have $\frommachine(C,V,\contacirc(W).S') = \frommachine(C,\apply(W,V),S') = (C,E[\apply(W,V)])^{\vec k}.X$. This configuration cannot be reduced by the \textit{head} rule (since by Theorem \ref{context reduction lemma} the reducibility of $(C,E[\apply(W,V)])$ would imply the reducibility of $(C,\apply(W,V))$, which is absurd since $W\not\equiv (\vec\ell,D,\vec{\ell'})$), nor by the \textit{step-in} rule (since $\apply(W,V)$ is a proto-redex and therefore by Proposition \ref{context exclusion theorem proto} $E[\apply(W,V)] \not\equiv F[\boxt_T(\lift P)]$, for any $F,P$), nor by the \textit{step-out} rule (since $E[\apply(W,V)]$ cannot be a value). Since $(C,E[\apply(W,V)])^{\vec k}.X$ is irreducible and $E[\apply(W,V)]$ is not a value, the claim is proven.
    
    \item Case of $\conttright(N)$. This case is impossible, since it would allow $(C,V,S)$ to be reduced by the \textit{tuple-shift} rule, regardless of $V$, contradicting the hypothesis.
    
    \item Case of $\conttleft(W)$. This case is impossible, since it would allow $(C,V,S)$ to be reduced by the \textit{tuple-join} rule, regardless of $V$, contradicting the hypothesis.
    
    \item Case of $\contbox(Q,\vec\ell)$. In this case $V\not\equiv \lift N$, since otherwise $(C,\lift N,\contbox(Q,\vec\ell).S')$ would be reducible by the \textit{box-sub} rule. We have $\frommachine(C,V,\contbox(Q,\vec\ell).S') = \frommachine(C,\boxt_T V,S') = (C,E[\boxt_T V])^{\vec k}.X$, where $\emptyset;Q\vdash \vec\ell : T$. This configuration cannot be reduced by the \textit{head} rule (since by Theorem \ref{context reduction lemma} the reducibility of $(C,E[\boxt_T V])$ would imply the reducibility of $(C,\boxt_T V)$, which is absurd since $V\not\equiv \lift N$), nor by the \textit{step-in} rule (since $V\not\equiv \lift N$ and $\boxt_T V$ is a proto-redex, so by Proposition \ref{context exclusion theorem proto} $E[\boxt_T V] \not\equiv F[\boxt_T(\lift P)]$ for any other $F,P$), nor by the \textit{step-out} rule (since $E[\boxt_T(\lift N)]$ cannot be a value). Since $(C,E[\boxt_T V])^{\vec k}.X$ is irreducible and $E[\boxt_T V]$ is not a value, the claim is proven.
    
    \item Case of $\contsub(D,N,\vec\ell,T)$. In this case $V\not\equiv \vec{\ell'}$, or else $(C,\vec{\ell'},\contsub(D,N,\vec\ell,T).S')$ would be reducible by the \textit{box-close} rule. We have $\frommachine(C,V,\contsub(D,N,\vec\ell,T).S') = (C,V)^{\vec\ell}.\frommachine(D,\boxt_T(\lift N),S') = (C,V)^{\vec\ell}.(D,E[\boxt_T(\lift N)])^{\vec k}.X$. This configuration cannot be reduced by the \textit{head} rule (since $V$ is a value), nor by the \textit{step-in} rule (since $V$ cannot be of the form $F[\boxt_T(\lift P)]$, for any $F,P$), nor by the \textit{step-out} rule (since $V\not\equiv \vec{\ell'}$). Since $(C,V)^{\vec\ell}.(D,E[\boxt_T(\lift N)])^{\vec k}.X$ is irreducible and $(D,E[\boxt_T(\lift N)])^{\vec k}.X \neq \epsilon$, the claim is proven.
    
    \emergencystretch=35pt
    \item Case of $\contlet(x,y,N)$. In this case we know $V\not\equiv \tuple{V',V''}$, or else $(C,\tuple{V',V''},\contlet(x,y,N).S')$ would be reducible by the \textit{let-join} rule. We have that $\frommachine(C,V,\contlet(x,y,N).S') = \frommachine(C,\letin{\tuple{x,y}}{V}{N},S') = (C,E[\letin{\tuple{x,y}}{V}{N}])^{\vec k}.X$. This configuration cannot be reduced by the \textit{head} rule (since by Theorem \ref{context reduction lemma} the reducibility of $(C,E[\letin{\tuple{x,y}}{V}{N}])$ would imply the reducibility of $(C,\letin{\tuple{x,y}}{V}{N})$, which is absurd since $V\not\equiv \tuple{V',V''}$), nor by the \textit{step-in} rule (since $\letin{\tuple{x,y}}{V}{N}$ is a proto-redex and therefore by Proposition \ref{context exclusion theorem proto} $E[\letin{\tuple{x,y}}{V}{N}] \not\equiv F[\boxt_T(\lift P)]$, for any $F,P$), nor by the \textit{step-out} rule (since $E[\letin{\tuple{x,y}}{V}{N}]$ cannot be a value). Since $(C,E[\letin{\tuple{x,y}}{V}{N}])^{\vec k}.X$ is irreducible and $E[\letin{\tuple{x,y}}{V}{N}]$ is not a value, the claim is proven.
    
    \item Case of $\contforce$. In this case we know $V\not\equiv \lift N$, since otherwise $(C,\lift N,\contforce.S')$ would be reducible by the \textit{force-close} rule. We have $\frommachine(C,V,\contforce.S') = \frommachine(C,\force V,S') = (C,E[\force V])^{\vec k}.X$. This configuration cannot be reduced by the \textit{head} rule (since by Theorem \ref{context reduction lemma} the reducibility of $(C,E[\force V])$ would imply the reducibility of $(C,\force V)$, which is absurd since $V\not\equiv \lift N$), nor by the \textit{step-in} rule (since $\force V$ is a proto-redex and therefore by Proposition \ref{context exclusion theorem proto} $E[\force N] \not\equiv F[\boxt_T(\lift P)]$, for any $F,P$), nor by the \textit{step-out} rule (since $E[\force V]$ cannot be a value). Since $(C,E[\force V])^{\vec k}.X$ is irreducible and $E[\force V]$ is not a value, the claim is proven.
\end{itemize}
\end{proofEnd}

\begin{lem}\label{machine deadlocked implies stacked deadlocked}
\emergencystretch=20pt
Suppose $(C,M,S)$ is a machine configuration. Whenever $(C,M,S)\deadlocks$, we have $\frommachine(C,M,S)\deadlocks$.
\end{lem}
\begin{proof}
It is easy to see that a machine configuration $(C,M,S)$ goes into deadlock if and only if $(C,M,S)\amred^*(D,V,R)$ for some irreducible $(D,V,R)$ such that $R\neq \epsilon$. It is also easy to see that a stacked configuration $(C,M)^{\vec\ell}.X$ goes into deadlock if and only if $(C,M)^{\vec\ell}.X \stred^* (D,N)^{\vec k}.Y$ for some irreducible $(D,N)^{\vec k}.Y$ such that either $N$ is not a value or $\vec k\neq \emptyset, Y\neq \epsilon$. By a finite number of applications of Lemma \ref{machine to stacked correspondence} we know that $\frommachine(C,M,S)\stred^*\frommachine(D,V,R)$ and by Lemma \ref{machine deadlocked implies stacked deadlocked base} we know that $\frommachine(D,V,R)=(D,E[V])^{\vec k}.Y$ for some $E,\vec k, Y$ such that either $E[V]$ is not a value or $\vec k\neq\emptyset, Y\neq \epsilon$. We therefore conclude that $\frommachine(C,M,S)\deadlocks$.
\end{proof}

Next, we show that $\frommachine$ preserves the convergence of reachable configurations.

\begin{theoremEnd}[all end]{lem}\label{machine stacked reachability}
\emergencystretch=10pt
If $(C,M,S)$ is a reachable machine configuration, then $\frommachine(C,M,S)$ is a reachable stacked configuration.
\end{theoremEnd}
\begin{proofEnd}
By induction on the reachability of $(C,M,S)$. In the case in which $S=\epsilon$ we have $\frommachine(C,M,\epsilon)=(C,M)^\emptyset.\epsilon \in \initialstacked$, which is reachable. In the case in which there exists $(D,N,S')$ such that $(D,N,S')$ is reachable and $(D,N,S') \amred (C,M,S)$ we proceed by cases on the introduction of $(D,N,S') \amred (C,M,S)$:
\begin{itemize}
    \item Case of \textit{app-split}. In this case we have $(D,NP,S')\amred (D,N,\contfarg(P).S')$. By inductive hypothesis we know that $\frommachine(D,NP,S')$ is reachable, so we immediately conclude that $\frommachine(D,N,\contfarg(P).S') = \frommachine(D,NP,S')$ is reachable.
    
    \item Case of \textit{app-shift}. In this case we have $(D,V,\contfarg(P).S')\amred (D,P,\contfapp(V).S')$. By inductive hypothesis we know that $\frommachine(D,V,\contfarg(P).S')$ is reachable. Also, by the definition of $\frommachine$ we have
    \begin{align*}
        \frommachine(D,V,\contfarg(P).S') &= \frommachine(D,VP,S') \\&= \frommachine(D,P,\contfapp(V).S'),
    \end{align*}
    so we conclude that $\frommachine(D,P,\contfapp(V).S')$ is reachable.
    
    \emergencystretch=25pt
    \item Case of \textit{app-join}. In this case we have that $(D,V,\contfarg(\lambda x.N).S') \amred (D,N[V/x],S')$. By propositions \ref{stack to context lemma} and \ref{stack to similarities lemma} we know that $\frommachine(D,V,\contfarg(\lambda x.N).S')=\frommachine(D,(\lambda x.N)V,S')=(D,E[(\lambda x.N)V])^{\vec\ell}.X, \frommachine(D,N[V/x],S')=(D,E[N[V/x]])^{\vec\ell}.X$ for the same $E,\vec\ell,X$, and by inductive hypothesis we know that $(D,E[(\lambda x.N)V])^{\vec\ell}.X$ is reachable. Because $(D,(\lambda x.N)V) \red (D,N[V/x])$ by the \textit{$\beta$-reduction} rule, by Theorem $\ref{context reduction lemma}$ we get $(D,E[(\lambda x.N)V]) \red (D,E[N[V/x]])$. Also, because $(\lambda x.N)V$ is a redex, we know by Corollary \ref{context exclusion lemma redex} that $E[(\lambda x.N)V] \not\equiv E'[\boxt_T(\lift P)]$ for any $E'$, so we get $(D,E[(\lambda x.N)V])^{\vec\ell}.X \stred (D,E[N[V/x]])^{\vec\ell}.X$ by the \textit{head} rule and conclude that $(D,E[N[V/x]])^{\vec\ell}.X$ is reachable.
    
    \emergencystretch=20pt
    \item Case of \textit{apply-split}. In this case we have $(D,\apply(N,P),S')\amred (D,N,\contalabel(P).S')$. By inductive hypothesis we know that $\frommachine(D,\apply(N,P),S')$ is reachable, so we conclude that $\frommachine(D,N,\contalabel(P).S') = \frommachine(D,\apply(N,P),S')$ is reachable.
    
    \item Case of \textit{apply-shift}. In this case we have $(D,V,\contalabel(P).S')\amred (D,P,\contacirc(V).S')$. By inductive hypothesis we know that $\frommachine(D,V,\contalabel(P).S')$ is reachable. Also, by the definition of $\frommachine$ we have
    \begin{align*}
        \frommachine(D,V,\contalabel(P).S') &= \frommachine(D,\apply(N,P),S') \\&= \frommachine(D,P,\contacirc(V).S'),
    \end{align*}
    so we conclude that $\frommachine(D,P,\contacirc(V).S')$ is reachable.
    
    \item Case of \textit{apply-join}. Case of \textit{app-join}. In this case we have $(D,\vec k,\contacirc((\vec\ell,D',\vec{\ell'})).S') \amred (C,\vec{k'},S')$, where $(C,\vec{k'})=\append(D,\vec k, \vec\ell, D', \vec{\ell'})$. By propositions \ref{stack to context lemma} and \ref{stack to similarities lemma} we know $\frommachine(D,\vec k,\contacirc((\vec\ell,D',\vec{\ell'})).S')=\frommachine(D,\apply((\vec\ell,D',\vec{\ell'}),\vec k),S')=(D,E[\apply((\vec\ell,D',\vec{\ell'}),\vec k)])^{\vec{\ell''}}.X$ and $\frommachine(C,\vec{k'},S')=(C,E[\vec{k'}])^{\vec{\ell''}}.X$, for the same $E,\vec{\ell''},X$, and by inductive hypothesis we know $(D,E[\apply((\vec\ell,D',\vec{\ell'}),\vec k)])^{\vec{\ell''}}.X$ is reachable. Because $(D,\apply((\vec\ell,D',\vec{\ell'}),\vec k)) \red (C,\vec{k'})$ by the \textit{apply} rule, by Theorem $\ref{context reduction lemma}$ we get $(D,E[\apply((\vec\ell,D',\vec{\ell'}),\vec k)]) \red (C,E[\vec{k'}])$. Also, because $\apply((\vec\ell,D',\vec{\ell'}),\vec k)$ is a redex, we know by Corollary \ref{context exclusion lemma redex} that $E[\apply((\vec\ell,D',\vec{\ell'}),\vec k)] \not\equiv E'[\boxt_T(\lift P)]$ for any $E'$, so we get $(D,E[\apply((\vec\ell,D',\vec{\ell'}),\vec k)])^{\vec{\ell''}}.X \stred (C,E[\vec{k'}])^{\vec{\ell''}}.X$ by the \textit{head} rule and conclude that $(C,E[\vec{k'}])^{\vec{\ell''}}.X$ is reachable.
    
    \item Case of \textit{tuple-split}. In this case we have $(D,\tuple{N,P},S')\amred (D,N,\conttright(P).S')$. By inductive hypothesis we know that $\frommachine(D,\tuple{N,P},S')$ is reachable, so we immediately conclude that $\frommachine(D,N,\conttright(P).S') = \frommachine(D,\tuple{N,P},S')$ is reachable.
    
    \item Case of \textit{tuple-shift}. In this case we have $(D,V,\conttright(P).S')\amred (D,P,\conttleft(V).S')$. By inductive hypothesis we know that $\frommachine(D,V,\conttright(P).S')$ is reachable. Also, by the definition of $\frommachine$ we have
    \begin{align*}
        \frommachine(D,V,\conttright(P).S') &= \frommachine(D,\tuple{V,P},S') \\&= \frommachine(D,P,\conttleft(V).S'),
    \end{align*}
    so we conclude that $\frommachine(D,P,\conttleft(V).S')$ is reachable.
    
    \emergencystretch=25pt
    \item Case of \textit{tuple-join}. In this case we have $(D,W,\conttleft(V).S') \amred (D,\tuple{V,W},S')$. By inductive hypothesis we immediately know that $\frommachine(D,W,\conttleft(V).S')=\frommachine(D,\tuple{V,W},S')$ is reachable and the claim is trivially true.
    
    \item Case of \textit{box-open}. In this case we have $(D,\boxt_T N,S')\amred (D,N,\contbox(Q,\vec\ell).S')$, where $(Q,\vec\ell)=\freshlabels(N,T)$. By inductive hypothesis we know that the configuration $\frommachine(D,\boxt_T N,S')$ is reachable, so we immediately conclude that $\frommachine(D,N,\contbox(Q,\vec\ell).S') = \frommachine(D,\boxt_T N,S')$ is reachable.
    
    \emergencystretch=50pt
    \item Case of \textit{box-sub}. In this case $(D,\lift N,\contbox(Q,\vec\ell).S')\amred (id_Q,N\vec\ell,\contsub(D,N,\vec\ell,T).S')$. By inductive hypothesis we know that $\frommachine(D,\lift N,\contbox(Q,\vec\ell).S')=\frommachine(D,\boxt_T(\lift N),S')$ is reachable. We also know that
    \begin{align*}
        &\frommachine(id_Q,N\vec\ell,\contsub(D,N,\vec\ell,T).S')\\ &= (id_Q,N\vec\ell)^{\vec\ell}.\frommachine(D,\boxt_T(\lift N),S').
    \end{align*}
    By Proposition \ref{stack to context lemma} we get $\frommachine(D,\boxt_T(\lift N),S')=(D,E[\boxt_T(\lift N)])^{\vec k}.X$ for some $E, \vec k$ and $X$. Finally, by the \textit{step-in} rule we have $(D,E[\boxt_T(\lift N)])^{\vec\ell}.X \stred (id_Q,N\vec\ell)^{\vec\ell}.(D,E[\boxt_T(\lift N)])^{\vec\ell}.X$, where $(Q,\vec\ell)=\freshlabels(N,T)$, and conclude that the latter is reachable.
    
    \item Case of \textit{box-close}. In this case we have $(D,\vec{\ell'},\contsub(C,N,\vec\ell,T).S') \amred (C,(\vec\ell,D,\vec{\ell'}),S')$. By propositions \ref{stack to context lemma} and \ref{stack to similarities lemma} we know that $\frommachine(D,\vec{\ell'},\contsub(C,N,\vec\ell,T).S')=(D,\vec{\ell'})^{\vec\ell}.\frommachine(C,\boxt_T(\lift N),S')=(D,\vec{\ell'})^{\vec\ell}.(C,E[\boxt_T(\lift N)])^{\vec k}.X$ and $\frommachine(C,(\vec\ell,D,\vec{\ell'}),S')=(C,E[(\vec\ell,D,\vec{\ell'})])^{\vec k}.X$, for the same $E,\vec k, X$, and by inductive hypothesis we know that $(D,\vec{\ell'})^{\vec\ell}.(C,E[\boxt_T(\lift N)])^{\vec k}.X$ is reachable. Because we have $(D,\vec{\ell'})^{\vec\ell}.(C,E[\boxt_T(\lift N)])^{\vec k}.X \stred (C,E[(\vec\ell,D,\vec{\ell'})])^{\vec k}.X$ by the \textit{step-out} rule, we conclude that $(C,E[(\vec\ell,D,\vec{\ell'})])^{\vec k}.X$ is reachable.
    
    \item Case of \textit{let-split}. In this case we have $(D,\letin{\tuple{x,y}}{N}{P},S')\amred (D,N,\contlet(x,y,P).S')$. By inductive hypothesis we know  $\frommachine(D,\letin{\tuple{x,y}}{N}{P},S')$ is reachable, so we conclude that $\frommachine(D,N,\contlet(x,y,P).S') = \frommachine(D,\letin{\tuple{x,y}}{N}{P},S')$ is reachable.
    
    \emergencystretch=20pt
    \item Case of \textit{let-join}. In this case $(D,\tuple{V,W},\contlet(x,y,N).S') \amred (D,N[V/x][W/y],S')$. By propositions \ref{stack to context lemma} and \ref{stack to similarities lemma} we know that $\frommachine(D,\tuple{V,W},\contlet(x,y,N).S')=\frommachine(D,\letin{\tuple{x,y}}{\tuple{V,W}}{N},S')=(D,E[\letin{\tuple{x,y}}{\tuple{V,W}}{N}])^{\vec\ell}.X$ and $\frommachine(D,N[V/x][W/y],S')=(D,E[N[V/x][W/y]])^{\vec\ell}.X$, for the same $E,\vec\ell,X$, and by inductive hypothesis we know that $(D,E[\letin{\tuple{x,y}}{\tuple{V,W}}{N}])^{\vec\ell}.X$ is reachable. Because $(D,\letin{\tuple{x,y}}{\tuple{V,W}}{N}) \red (D,N[V/x][W/y])$ by the \textit{let} rule, by Theorem $\ref{context reduction lemma}$ we get $(D,E[\letin{\tuple{x,y}}{\tuple{V,W}}{N}]) \red (D,E[N[V/x][W/y]])$. Also, because $\letin{\tuple{x,y}}{\tuple{V,W}}{N}$ is a redex, we know by Corollary \ref{context exclusion lemma redex} that $E[\letin{\tuple{x,y}}{\tuple{V,W}}{N}] \not\equiv E'[\boxt_T(\lift P)]$ for any $E'$, so we get $(D,E[\letin{\tuple{x,y}}{\tuple{V,W}}{N}])^{\vec\ell}.X \stred (D,E[N[V/x][W/y]])^{\vec\ell}.X$ by the \textit{head} rule and conclude that $(D,E[N[V/x][W/y]])^{\vec\ell}.X$ is reachable.
    
    \item Case of \textit{force-open}. In this case we have $(D,\force N,S')\amred (D,N,\contforce.S')$. By inductive hypothesis we know that $\frommachine(D,\force N,S')$ is reachable, so we immediately conclude that $\frommachine(D,N,\contforce.S') = \frommachine(D,\force N,S')$ is reachable.
    
    \emergencystretch=35pt
    \item Case of \textit{force-close}. In this case we have $(D,\lift N,\contforce.S') \amred (D,N,S')$. By propositions \ref{stack to context lemma} and \ref{stack to similarities lemma} we know $\frommachine(D,\lift N,\contforce.S')=\frommachine(D,\force(\lift N),S')=(D,E[\force(\lift N)])^{\vec\ell}.X, \frommachine(D,N,S')=(D,E[N])^{\vec\ell}.X$ for the same $E,\vec\ell,X$, and by inductive hypothesis we know that $(D,E[\force(\lift N)])^{\vec\ell}.X$ is reachable. Because $(D,\force(\lift N)) \red (D,N)$ by the \textit{force} rule, by Theorem $\ref{context reduction lemma}$ we get $(D,E[\force(\lift N)]) \red (D,E[N])$. Also, because $\force(\lift N)$ is a redex, we know by Corollary \ref{context exclusion lemma redex} that $E[\force(\lift N)] \not\equiv E'[\boxt_T(\lift P)]$ for any $E'$, so we get $(D,E[\force(\lift N)])^{\vec\ell}.X \stred (D,E[N])^{\vec\ell}.X$ by the \textit{head} rule and conclude that $(D,E[N])^{\vec\ell}.X$ is reachable.
    
\end{itemize}
\end{proofEnd}

\begin{theoremEnd}{prop}\label{preservation of convergence frommachine}
Suppose $(C,M,S)$ is a reachable machine configuration. $(C,M,S)\converges$ if and only if $\frommachine(C,M,S)\converges$.
\end{theoremEnd}
\begin{proofEnd}
We start by proving that if $(C,M,S)\converges$, then $\frommachine(C,M,S)\converges$. We proceed by induction on $(C,M,S)\converges$:
\begin{itemize}
    \item Case $M\equiv V$ and $S=\epsilon$. In this case $\frommachine(C,V,\epsilon)=(C,V)^\emptyset.\epsilon$ and the claim is trivially true.
    
    \emergencystretch=20pt
    \item Case $(C,M,S)\amred (D,N,S')$ and $(D,N,S')\converges$. In this case we know by inductive hypothesis that $\frommachine(D,N,S')\converges$. By Lemma \ref{machine to stacked correspondence} we get that $\frommachine(C,M,S)\stred^* \frommachine(D,N,S')$ and by the definition of converging stacked configuration we conclude $\frommachine(C,M,S)\converges$.
\end{itemize}
We now need to prove that if $\frommachine(C,M,S)\converges$, then $(C,M,S)\converges$. We proceed by induction on $\frommachine(C,M,S)\downarrow$:
\begin{itemize}
    \item Case $\frommachine(C,M,S) = (C,V)^\emptyset.\epsilon$. We know $\frommachine(C,M,S) = (C,E[M])^{\vec\ell}.X$ by Proposition \ref{stack to context lemma}. For $E[M]\equiv V$ to be true we must have $E\equiv [\cdot]$ and $M\equiv V$. We must also have $\vec\ell=\emptyset$ and $ X=\epsilon$. The only way to have $E\equiv[\cdot],\vec\ell=\emptyset$ and $X=\epsilon$ is to have $S=\epsilon$. If $S=H.S'$ were non-empty, we would either have $\vec\ell\neq\emptyset$ and $X\neq \epsilon$ (in case of an $H$ of type $\contsub$) or $E\neq [\cdot]$ (for any other $H$), which would contradict the hypothesis. Therefore we have $M\equiv V$ and $S=\epsilon$ and we conclude $(C,V,\epsilon)\downarrow$.
    
    \item Case $\frommachine(C,M,S) \stred (D,N)^{\vec {k}}.X'$ and $(D,N)^{\vec {k}}.X'\converges$. Let $(C',M',S')$ be the normal form of $(C,M,S)$ with respect to $\amred_b$. We distinguish two cases:
    \begin{itemize}
        \item If $(C',M',S')$ is also normal with respect to $\amred$, then either $(C,M,S)\deadlocks$ or $(C,M,S)\converges$. Because by Lemma \ref{machine deadlocked implies stacked deadlocked} $(C,M,S)\deadlocks$ would imply $\frommachine(C,M,S)\deadlocks$, contradicting Proposition \ref{stacked mutex}, we conclude $(C,M,S)\converges$
        
        \emergencystretch=30pt
        \item If $(C',M',S')$ is not normal with respect to $\amred$, then we have $(C',M',S') \amred_r (D',N',S'')$ and by Lemma \ref{machine to stacked correspondence} we know that $\frommachine(C',M',S') \stred \frommachine(D',N',S'')$. At the same time, by the definition of $\amred_b$ we know that $\frommachine(C',M',S') = \frommachine(C,M,S)$. Because $\stred$ is deterministic, this necessarily entails $\frommachine(D',N',S'') = (D,N)^{\vec {k}}.X'$ and consequently $\frommachine(D',N',S'') \converges$. By inductive hypothesis we get $(D',N',S'') \converges$ and conclude $(C,M,S) \converges$ by the definition of converging machine configuration.
    \end{itemize}
    
\end{itemize}
\end{proofEnd}

\subsection{Equivalence Between Small-step and Machine Semantics}
\label{subsection equivalence small-step machine}

Finally, in this last section we cover the relationship between the small-step semantics and the machine semantics. In particular, as we anticipated, we show that they are essentially equivalent. Concretely, this means that converging computations in the small-step semantics translate to converging computations in the machine semantics (which converge to the same values), that computations that go into deadlock in the small-step semantics translate to computations that go into deadlock in the machine semantics, and -- last but not least -- that diverging computations in the small-step semantics translate to diverging computations in the machine semantics, and vice-versa. In order to prove these results, let us formalize the relationship between small-step and machine configurations through a simple $\loadmachine$ function, defined as such:
$$\loadmachine(C,M) = (C,M,\epsilon).$$
This function is very similar to $\fromsmallstep$. Like $\fromsmallstep$ it is trivially invertible, and as such is establishes a bijection between small-step configurations and the set $\initialmachine$ of initial machine configurations. In addition, $\loadmachine$ has the following property:
$$\frommachine \circ \loadmachine = \fromsmallstep,$$
which will be essential in the coming proofs. The property is actually trivial to prove, as for every small-step configuration $(C,M)$ we have:
$$\frommachine(\loadmachine(C,M))=\frommachine(C,M,\epsilon)=(C,M)^\emptyset.\epsilon = \fromsmallstep(C,M).$$

\subsubsection{Convergence}

As we mentioned at the beginning of this section, we prove that corresponding small-step and machine configurations evaluate to the same circuit and value in the respective semantics by proving that the two computations are simulated by the same computation in the stacked semantics. To this effect, lemmata \ref{small step to stacked correspondence} and \ref{machine to stacked correspondence} are going to play a paramount role. Specifically, we want the diagrams that we introduced in sections \ref{from small-step to stacked reductions} and \ref{machine to stacked reductions} to compose as follows:

\begin{center}
    \begin{tikzcd}[column sep = 4cm]
    (C,M) \arrow{r}[pos=1]{*} \arrow[d, equal, "\fromsmallstep" description] \arrow[dd, bend right = 75, equal, "\loadmachine" description] & (D,V) \arrow[d, equal, "\fromsmallstep" description] \arrow[dd, bend left = 75, equal, "\loadmachine" description] \\
    (C,M)^\emptyset.\epsilon \arrow[harpoon]{r}[pos=1]{*} & (D,V)^\emptyset.\epsilon\\
    (C,M,\epsilon) \arrow[Rightarrow]{r}[pos=1]{*} \arrow[u, equal, "\frommachine" description] & (D,V,\epsilon) \arrow[u, equal, "\frommachine" description]
    \end{tikzcd}
\end{center} 
For simplicity, we prove the two directions of the equivalence separately and then put them together to prove our goal.

\begin{theoremEnd}[normal]{lem}\label{small-step to machine correspondence}
Suppose $(C,M)$ and $(D,V)$ are small-step configurations. If $(C,M)\red^*(D,V)$, then $\loadmachine(C,M) \amred^* \loadmachine(D,V)$.
\end{theoremEnd}
\begin{proofEnd}
The case in which $(C,M)=(D,V)$ is trivially true, so let us consider the case in which $(C,M)\red^+(D,V)$.
First of all, by Lemma \ref{small step to stacked correspondence} we get that $(C,M)^\emptyset.\epsilon \stred^+ (D,V)^\emptyset.\epsilon$.
In parallel, by propositions \ref{preservation of convergence fromsmallstep} and \ref{preservation of convergence frommachine} we have that $\loadmachine(C,M)=(C,M,\epsilon)$ converges, since $\fromsmallstep(C,M) = (C,M)^\emptyset.\epsilon = \frommachine(C,M,\epsilon)$ and $(C,M,\epsilon)$ is trivially reachable. This entails that there exists a normal form $(D',V',\epsilon)$ such that $(C,M,\epsilon) \amred ^* (D',V',\epsilon)$. By Lemma \ref{machine to stacked correspondence} this implies that $\frommachine(C,M,\epsilon) \stred^* \frommachine(D',V',\epsilon)$, or $(C,M)^\emptyset.\epsilon \stred^* (D',V')^\emptyset.\epsilon$. We now have $(C,M)^\emptyset.\epsilon \stred^+ (D,V)^\emptyset.\epsilon$ and $(C,M)^\emptyset.\epsilon \stred^* (D',V')^\emptyset.\epsilon$. Because $\stred$ is deterministic and because $(D,V)^\emptyset.\epsilon$ and $(D',V')^\emptyset.\epsilon$ are both normal forms, we get that $(D,V)^\emptyset.\epsilon = (D',V')^\emptyset.\epsilon$, that is, $D=D'$ and $V=V'$, and we conclude $\loadmachine(C,M)\amred^*(D,V,\epsilon)$.
\end{proofEnd}

\begin{theoremEnd}[normal]{lem}\label{machine to small-step correspondence}
\emergencystretch=20pt
Suppose $(C,M)$ and $(D,V)$ are two small-step configurations. If $\loadmachine(C,M) \amred^* \loadmachine(D,V)$, then $(C,M)\red^*(D,V)$.
\end{theoremEnd}
\begin{proofEnd}
The case in which $\loadmachine(C,M)=\loadmachine(D,V)$ is trivially true, so let us consider the case in which $\loadmachine(C,M)\amred^+ \loadmachine(D,V)$.
By Proposition \ref{preservation of convergence frommachine} we know that  $\frommachine(\loadmachine(C,M)) = \fromsmallstep(C,M)$ converges, and by Proposition \ref{preservation of convergence fromsmallstep} we know that $(C,M)$ converges too. That is, $(C,M)\red^*(D',V')$. By Lemma \ref{small-step to machine correspondence} this entails $\loadmachine(C,M)\amred^*\loadmachine(D',V')$. Since $\loadmachine(D,V)$ and $\loadmachine(D',V')$ are both normal forms, $(D,V)\neq(D',V')$ would contradict the determinism of $\amred$. As a result, we know that $(D,V)=(D',V')$ and conclude $(C,M)\red^*(D,V)$.
\end{proofEnd}

\begin{thm}[Equivalence in Convergence]
\label{equivalence between small-step and machine}
Suppose $(C,M)$ and $(D,V)$ are small-step configurations. We have that $(C,M)\red^*(D,V)$ if and only if $\loadmachine(C,M) \amred^* \loadmachine(D,V)$.
\end{thm}
\begin{proof}
The claim follows immediately from lemmata \ref{small-step to machine correspondence} and \ref{machine to small-step correspondence}.
\end{proof}

\begin{cor}\label{preservation of convergence loadmachine}
Suppose $(C,M)$ is a small-step configuration. We have that $(C,M)\converges$ if and only if $\loadmachine(C,M)\converges$.
\end{cor}
\begin{proof}
The claim follows immediately from Theorem \ref{equivalence between small-step and machine} and the definition for converging small-step and machine configurations. Alternatively, it follows from propositions \ref{preservation of convergence fromsmallstep} and \ref{preservation of convergence frommachine}.
\end{proof}

\subsubsection{Deadlock and Divergence}

The equivalence between the small-step and machine semantics is stronger than the one between the big-step and the small-step semantics, as we now show that whenever a small-step computation goes into deadlock or diverges, then the corresponding machine computation goes into deadlock or diverges, respectively.

\begin{theoremEnd}{lem} \label{small-step deadlock to machine deadlock}
\emergencystretch=20pt
Suppose $(C,M)$ is a small-step configuration. If $(C,M)\deadlocks$, then $\loadmachine(C,M)\deadlocks$.
\end{theoremEnd}
\begin{proofEnd}
\emergencystretch=30pt
First of all, by Lemma \ref{small-step deadlock to stacked deadlock} we know that $\fromsmallstep(C,M)\deadlocks$. Now, suppose $\loadmachine(C,M)\ndeadlocks$. By Proposition \ref{machine totality} we know that either $\loadmachine(C,M)\converges$ or $\loadmachine(C,M)\diverges$. Because $\loadmachine(C,M)=(C,M,\epsilon)$ is trivially reachable, by Proposition \ref{preservation of convergence frommachine} we know that if $\loadmachine(C,M)\converges$ then $\frommachine(\loadmachine(C,M))=\fromsmallstep(C,M)\converges$, which contradicts Proposition \ref{stacked mutex}, so $\loadmachine(C,M)\nconverges$. On the other hand, $\loadmachine(C,M)\diverges$ would entail an infinite computation starting from $\loadmachine(C,M)$. By lemmata \ref{machine to stacked correspondence} and \ref{strong normalization of =>b} this would entail an infinite computation starting from $\frommachine(\loadmachine(C,M))=\fromsmallstep(C,M)$ too, which would contradict $\fromsmallstep(C,M)\deadlocks$, since $\fromsmallstep(C,M)\deadlocks$ implies that the same computation is finite ($\stred$ is deterministic). Because $\loadmachine(C,M)\ndeadlocks$ ultimately leads to a contradiction, we conclude $\loadmachine(C,M)\deadlocks$.
\end{proofEnd}

\begin{theoremEnd}{lem}\label{small-step divergence to machine divergence}
\emergencystretch=20pt
Suppose $(C,M)$ is a small-step configuration. If $(C,M)\diverges$, then $\loadmachine(C,M)\diverges$.
\end{theoremEnd}
\begin{proofEnd}
\emergencystretch=20pt
First of all, because the computation starting from $(C,M)$ is infinite, by Lemma \ref{small step to stacked correspondence} we know that the computation starting from $\fromsmallstep(C,M)$ is also infinite ($\stred$ is deterministic). Now, suppose $\loadmachine(C,M)\ndiverges$. By Proposition \ref{machine totality} we know that either $\loadmachine(C,M)\converges$ or $\loadmachine(C,M)\deadlocks$. Because $\loadmachine(C,M)$ is trivially reachable, by Proposition \ref{preservation of convergence frommachine} we know that if $\loadmachine(C,M)\converges$ then $\frommachine(\loadmachine(C,M))=\fromsmallstep(C,M)\converges$. This contradicts the fact that the computation starting from $\fromsmallstep(C,M)$ is infinite, since $\fromsmallstep(C,M)\converges$ implies that the same computation is finite, so $\loadmachine(C,M)\nconverges$.
On the other hand, if $\loadmachine(C,M) \deadlocks$, by Lemma \ref{machine deadlocked implies stacked deadlocked} we know that $\frommachine(\loadmachine(C,M))=\fromsmallstep(C,M) \deadlocks$. However, this contradicts the fact that the computation starting from $\fromsmallstep(C,M)$ is infinite, since $\fromsmallstep(C,M)\deadlocks$ implies that the same computation is finite, so $\loadmachine(C,M)\ndeadlocks$. Because $\loadmachine(C,M)\ndiverges$ ultimately leads to a contradiction, we conclude $\loadmachine(C,M)\diverges$.  
\end{proofEnd}
\begin{lem}\label{machine deadlock to small-step deadlock}
Suppose $(C,M)$ is a small-step configuration. If $\loadmachine(C,M)\deadlocks$, then $(C,M)\deadlocks$.
\end{lem}
\begin{proof}
Suppose $(C,M)\ndeadlocks$. By Proposition \ref{small-step totality} we know that either $(C,M)\converges$ or $(C,M)\diverges$. However, by Corollary \ref{preservation of convergence loadmachine} $(C,M)\converges$ entails $\loadmachine(C,M)\converges$, while by Lemma \ref{small-step divergence to machine divergence} $(C,M)\diverges$ entails $\loadmachine(C,M)\diverges$. Because both these conclusions contradict Proposition \ref{machine mutex}, we conclude that $(C,M)\deadlocks$.
\end{proof}

\begin{lem}\label{machine divergence to small-step divergence}
Suppose $(C,M)$ is a small-step configuration. If $\loadmachine(C,M)\diverges$, then $(C,M)\diverges$.
\end{lem}
\begin{proof}
Suppose $(C,M)\ndiverges$. By Proposition \ref{small-step totality} we know that either $(C,M)\converges$ or $(C,M)\deadlocks$. However, by Corollary \ref{preservation of convergence loadmachine} $(C,M)\converges$ entails $\loadmachine(C,M)\converges$, while by Lemma \ref{small-step deadlock to machine deadlock} $(C,M)\deadlocks$ entails $\loadmachine(C,M)\deadlocks$. Because both these conclusions contradict Proposition \ref{machine mutex}, we conclude that $(C,M)\diverges$.
\end{proof}

Eventually, the four lemmata can be summarized in the following two theorems, which, together with Theorem \ref{equivalence between small-step and machine}, complete the equivalence between the small-step and machine semantics. Figure \ref{fig:final} in the next page illustrates the full picture of this equivalence.

\begin{thm}[Equivalence in Deadlock]
Suppose $(C,M)$ is a small-step configuration. We have that $(C,M)\deadlocks$ if and only if $\loadmachine(C,M)\deadlocks$.
\end{thm}
\begin{proof}
The claim follows immediately from lemmata \ref{small-step deadlock to machine deadlock} and \ref{machine deadlock to small-step deadlock}.
\end{proof}

\begin{thm}[Equivalence in Divergence]
Suppose $(C,M)$ is a small-step configuration. We have that $(C,M)\diverges$ if and only if $\loadmachine(C,M)\diverges$.
\end{thm}
\begin{proof}
The claim follows immediately from lemmata \ref{small-step divergence to machine divergence} and \ref{machine divergence to small-step divergence}.
\end{proof}

\newpage
\begin{figure}[H]
    \centering
    \includegraphics[]{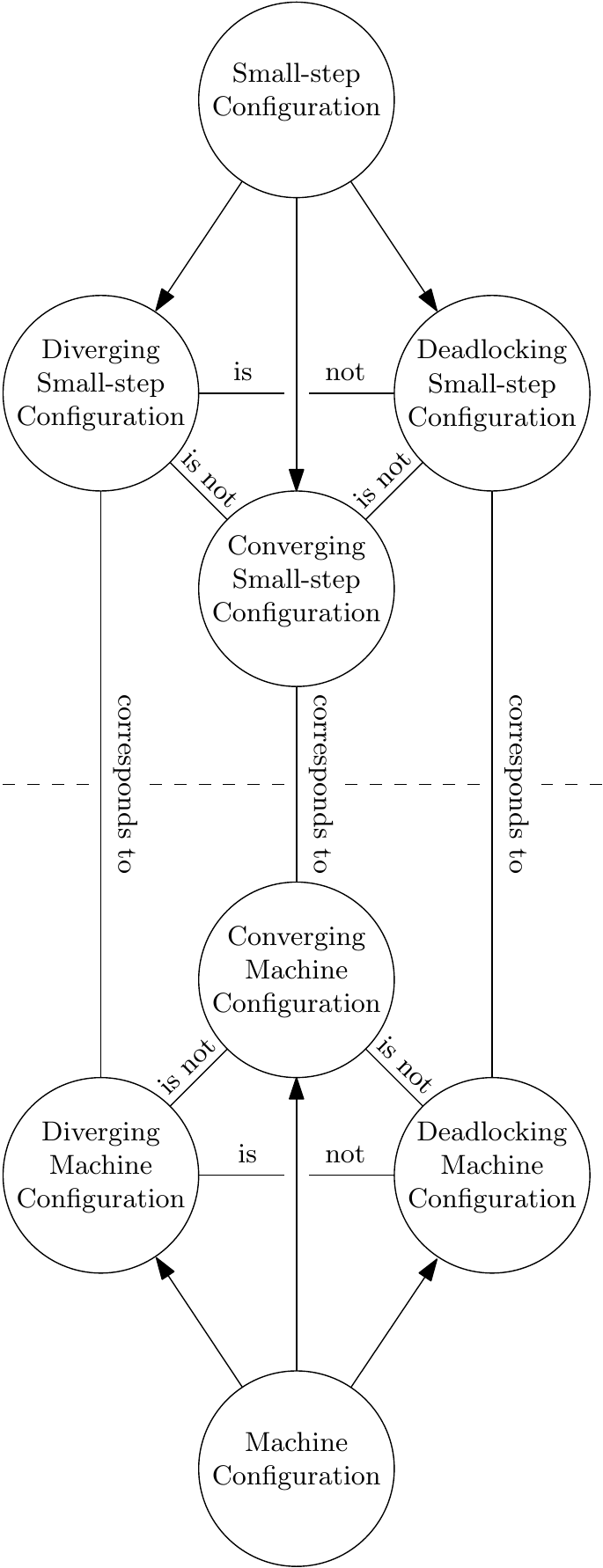}
    \caption{The final relationship between the small-step semantics modelled after Proto-Quipper-M's big-step semantics and the machine semantics proposed in this paper.}
    \label{fig:final}
\end{figure}

\newpage
\section{Conclusions and Future Work}
In the first part of the paper we reviewed Proto-Quipper-M, a member of the Proto-Quipper research language family, which aims at formalizing relevant fragments of Quipper in a type-safe way. We introduced the intuition behind Proto-Quipper-M, its categorical model for quantum circuits and its syntax. We also discussed its linear type system, which allows the enforcement of the no-cloning property of quantum states at compile time, effectively overcoming one of Quipper's greatest weaknesses. Last, but not least, we presented Proto-Quipper-M's big-step operational semantics.

\paragraph{} By rewriting the big-step rules of Proto-Quipper-M into small-step rules, we obtained an equivalent semantics which is small-step save for the case of circuit boxing. We showed that this semantics behaves well with respect to Proto-Quipper-M's type system by proving subject reduction and progress results. In the second part of the paper, we defined a \textit{stacked semantics} for Proto-Quipper-M, which overcomes the aforementioned problems with circuit boxing by organizing all of the sub-reductions introduced by a boxing operation in an explicit stack. We used this semantics as an intermediate step in the definition of a \textit{machine semantics} for Proto-Quipper-M, which takes this approach even further. Inspired by abstract machines such as the CEK machine, this semantics models every phase of the evaluation of a program as a continuation on a stack. Lastly, we concluded by proving that the proposed machine semantics is equivalent to the initial small-step semantics and -- as a consequence -- to the original big-step semantics given by Rios and Selinger.

\subsection{Future Work}
The point of arrival of our work is a minimal abstract machine which accurately models the operational semantics of the Proto-Quipper-M language, and therefore formalizes a fundamental fragment of the behavior of Quipper itself. From here, we can expect most of the future work to be focused on one of two directions.

\paragraph{} The first direction is that of expanding the current machine specification to progressively model a larger and larger portion of Quipper. First and foremost, we have that a considerable number of language features included in the original Proto-Quipper-M specifications by Rios and Selinger have been omitted in our work, for the sake of feasibility. These are not domain specific features, and include things such as sequencing operators, sum types, pattern matching, naturals, lists, and so on. Although not essential for circuit building, these are the features that usually make a programming language practical and, as a consequence, useful. Therefore, it would be appropriate, although unchallenging (and probably tedious) to extend the current machine specification with these programming constructs and to show that this extension does not compromise the results that we gave in this paper. More interestingly, the proposed machine semantics could be used as a starting point to model some of Quipper's most advanced features, which have no counterpart in Proto-Quipper-M. A prime example of such a feature is \textit{dynamic lifting}, which refers to the ability to measure the intermediate state of qubits in the midst of the execution of a circuit and to use the resulting classical information to build the remaining portion of the circuit on the fly.

\paragraph{}The second direction is one that we briefly mentioned in the introduction of the paper, and it is not completely orthogonal to the first one. The research direction in question is the one that focuses on the static analysis of interesting properties of Quipper programs. In this case, our machine semantics could be used as a reference model to define concepts such as the time needed to construct a circuit, or the number of qubits required by it. The estimation of the latter quantity, in particular, would be extremely valuable at a time where quantum resources are still scarce, and it is not trivial to compute, especially if dynamic lifting is involved.

\paragraph{} Lastly, as a side note, the machine itself could be made more concrete than it currently is. As we mentioned in Section \ref{an abstract machine for proto-quipper-m}, when designing our machine we chose to keep relying on an abstract substitution function for reasons of simplicity. Taking further inspiration from the CEK machine, an explicit substitution algorithm could be implemented by endowing our own machine with an environment component and rules to look up variables inside an environment. One major obstacle in this approach is that our environments would contain linear resources, and thus would have to be handled differently from how they are treated in the CEK machine. Note that this concretization operation is not essential, but it would be particularly beneficial to any research focusing -- for example -- on the static estimation of the time requirements of circuit generation.

\emergencystretch=0pt
\printbibliography

\newpage
\section*{Appendix}
\printProofs

\end{document}